\documentclass{article}

\usepackage[dvipsnames]{xcolor}
\usepackage{float}
\usepackage{macros}
\usepackage[section]{placeins}
\usepackage{pgfplots}
\usepackage{multirow}
\usepackage{authblk}
\usepackage{hyperref}

\begin{document}
  \title{Code-based Signatures from New Proofs of Knowledge for the Syndrome Decoding Problem} 

  \author{Loïc Bidoux$^1$, Philippe Gaborit$^2$, Mukul Kulkarni$^1$, Victor Mateu$^1$}
  \affil{$^1$ Technology Innovation Institute, UAE \\ $^2$ University of Limoges, France}
  \date{}

\maketitle

\begin{abstract} 
In this paper, we study code-based signatures constructed from Proof of Knowledge (PoK).
This line of work can be traced back to Stern who introduces the first efficient PoK for the syndrome decoding problem in 1993 \cite{Stern93}. 
Afterward, different variations were proposed in order to reduce signature's size. 
In practice, obtaining a smaller signature size relies on the interaction of two main considerations: (i) the underlying protocol and its soundness error and (ii) the type of optimizations which are compatible with a given protocol. 
In particular, optimizations related to the possibility to use random seeds instead of mere vectors have a great impact on the final signature length.
  Over the years, different variations were proposed to improve the Stern scheme such as the Veron scheme (with public key a noisy codeword rather than a syndrome) \cite{Veron97}, the AGS scheme which is a 5-pass protocol with cheating probability asymptotically equal to 1/2 \cite{AGS11} and more recently the FJR approach which permits to decrease the cheating probability to 1/N but induces a performance overhead \cite{FJR21}.
Overall the length of the signature depends on a trade-off between: the 
scheme in itself, the possible optimizations and the cost of the 
implementation. For instance, depending on the application one may prefer 
a 30\% shorter signature at the cost a ten times slower implementation 
rather than a longer signature but a faster implementation. The recent 
approaches which increase the cost of the implementation opens the door 
to many different type of trade-offs.

In this paper we propose three new schemes and different trade-offs, 
which are all interesting in themselves, since depending on potential 
future optimizations a scheme may eventually become more efficient than 
another. All the schemes we propose use a trusted helper:
a first scheme permits to get a 1/2 cheating probability,  a second 
scheme permits to decrease the cheating probability in 1/N but with a 
different approach than the recent FJR scheme and at last a third scheme 
propose a Veron-like adaptation of the FJR scheme in which the public 
key is a noisy codeword rather than a syndrome.
We provide an extensive comparison table which lists various trade-offs 
  between our schemes and previous ones. The table 
shows the interest of our constructions for certain type of 
trade-offs.
\end{abstract}

%--------------------------------------------------------------------%
\section{Introduction}
%--------------------------------------------------------------------%

The goal of post-quantum cryptography is to provide cryptographic schemes that are secure against adversaries using both classical and quantum computers.
Code-based cryptography was introduced by McEliece in 1978 \cite{McEliece78} and is nowadays one of the main alternative to classical cryptography.
This is illustrated by the ongoing NIST Post-Quantum Cryptography standardization process~\cite{NISTPQC} whose round~3 features three code-based Key Encapsulation Mechanisms (KEM) \cite{ClassicMcEliece, BIKE, HQC}.
Additional KEM were also considered during the round~2 of the competition; see \cite{LEDA,ROLLO,RQC}.
Unlike the code-based KEM, designing digital signatures from coding theory has historically been challenging.
Two approaches have been studied in this regard namely signatures from the hash-and-sign paradigm and signatures based on proofs of knowledge (PoK).
Regarding signatures based on the hash-and-sign paradigm, a first (although inefficient) construction was proposed in 2001 \cite{CFS01}.
The Wave construction \cite{wave19} follows the same approach and features small signature sizes.
Regarding code-based signatures from PoK, the Schnorr-Lyubashevsky \cite{Schnorr, Lyubashevsky} approach has been successfully used in the rank metric setting by the Durandal scheme \cite{Durandal}.
In this paper, we focus on the Fiat-Shamir paradigm \cite{FS, PS96} which relies on zero knowledge PoK.
In this approach, one transforms an honest verifier zero-knowledge interactive PoK into a signature scheme using the so called Fiat-Shamir heuristic.

The first efficient PoK for the syndrome decoding ($\SD$) problem over $\Ft$ was introduced by Stern in 1993~\cite{Stern93}.
In 1997, Véron improved the Stern protocol by designing a protocol based on the general syndrome decoding ($\GSD$) problem rather than the $\SD$ one \cite{Veron97}.
The $\SD$ and $\GSD$ problem are equivalent and only differ in the way used to represent the underlying code namely using a parity-check matrix in the former and using a generator matrix in the latter.
Stern and Véron protocols feature a soundness error equal to $2/3$ and as such need to be repeated several times in order to achieve a negligible soundness error.
In 2011, two 5-round code-based PoK reducing the soundness error close to $1/2$ (hence leading to smaller signature sizes) were proposed.
The first one (CVE) relies on the $\SD$ problem over $\Fq$ \cite{CVE11} while the second one (AGS) relies on the $\QCGSD$ problem over $\Ft$ namely the quasi-cyclic variant of the $\GSD$ problem \cite{AGS11}.
A zero-knowledge issue impacting $\GSD$ based protocols (Véron and AGS) was identified in \cite{ZKIssue} and fixed in \cite{ISIT21}.
In addition, the AGS protocol have been improved by the BGS proposal by using an optimization specifically tailored to the $\QCSD$ problem \cite{BGS21}.
Furthermore, it has been shown recently (see related work section bellow for additional details) that one can design a protocol achieving an arbitrarily small soundness error in \cite{GPS21} and \cite{FJR21}.
Some of the aforementioned protocols have been adapted to the rank metric setting, see \cite{Chen95, RankStern11, RankVeron19}.

Recently,  Katz, Kolesnikov and Wang~\cite{KKW18}, designed a signature
scheme based on PoK  using the MPC-in-the-head paradigm introduced
by~\cite{IKOS07}. An important highlight of this design is that
it allows one to achieve much smaller soundness errors which can results in shorter 
signatures in our case. While this benefit comes at the cost of
slightly involved protocols and performance overhead, with careful analysis and parameter selection it is possible to design signature schemes with acceptable performance and shorter sized using this framework.
Beullens generalized the work of \cite{KKW18} by introducing the notion of PoK with trusted helper and designing new PoK for the Multivariate Quadratic (MQ) problem, Permuted Kernel Problem (PKP) and Short Integer Solution (SIS) problem \cite{Beullens20}.
In this work, we propose a new code-based proof of knowledge (PoK) systems
with trusted helper for the syndrome decoding problem, and later construct signature schemes from the PoK by using the Fiat-Shamir transformation.

\vspace{\baselineskip}
\noindent \textbf{Contributions.} We introduce three new PoK with trusted helper for the syndrome decoding problem over $\Ft$.
The first one (denoted PoK~1) is a PoK for the $\SD$ problem achieving a soundness error equal to $1/2$ without any extra assumption such as using quasi-cyclic variants of the problem or working over $\Fq$.
The second one (denoted PoK~2) is a PoK for the $\SD$ problem over $\Ft$ (using its $\GSD$ form to be precise) achieving an arbitrarily small soundness error.
The third one (denoted PoK~3) is a variant of PoK~2 using some ideas from \cite{FJR21}.
Our proofs (as well as the \cite{FJR21} one) can leverage the quasi-cyclic variants of the $\SD$ or $\GSD$ problems to improve their performances although this is not mandatory.
Table~\ref{table:sota} compares our new PoK to existing ones with respect to their soundness error as well as their underlying security assumption.
The soundness error is directly linked to the resulting signature size (once the Fiat-Shamir heuristic have been applied) hence the smaller the soundness error is, the more compact the signature can be.
Regarding security assumptions, the $\SD$ problem over $\Ft$ has been arguably more studied than its counterpart over $\Fq$ hence can be considered as a slightly more conservative assumption.
The $\QCSD$ and $\QCGSD$ constitutes structured variants of the initial $\SD$ and $\GSD$ problems and as such are less conservative than the latter although they are believed to be hard by the community.
Similarly to our PoK, protocols over $\Fq$ could also use quasi-cyclicity to improve their performances which is not depicted in Table~\ref{table:sota}.

\vspace{0.5\baselineskip}
\begin{table}[ht]
\begin{center}
{\renewcommand{\arraystretch}{1.35}
  \begin{tabular}{|c|l|l|l|}
    \hline
    {\small \textbf{Soundness}} & \multicolumn{1}{c|}{$\SD$/$\GSD$ over $\Ft$} & \multicolumn{1}{c|}{$\mathsf{QC}$-$\SD$/$\GSD$ over $\Ft$} & \multicolumn{1}{c|}{$\SD$/$\GSD$ over $\Fq$} \\ \hline
    \multirow{2}{*}{2/3} & \cite{Stern93} & & \\
    & \cite{Veron97},\,\cite{ISIT21} & & \\
    \hline
    \multirow{3}{*}{1/2} & \multirow{3}{*}{PoK 1 (Section \ref{sec:pok1})} & \cite{AGS11},\,\cite{ISIT21} & \multirow{3}{*}{\hspace{-3pt} \cite{CVE11}} \\
    & & \cite{BGS21} & \\
    & & PoK~1 (Section \ref{sec:pok1}) & \\
    \hline
    \multirow{3}{*}{1/N} & \cite{FJR21} & \cite{FJR21} & \multirow{3}{*}{\hspace{-3pt} \cite{GPS21}} \\ 
     & PoK~2 (Section \ref{sec:pok2}) & PoK~2 (Section \ref{sec:pok2}) & \\ 
     & PoK~3 (Section \ref{sec:pok3}) & PoK~3 (Section \ref{sec:pok3}) & \\ 
    \hline
  \end{tabular}
  \caption{PoK for $\SD$/$\GSD$ or $\QCSD$/$\QCGSD$ problems in Hamming metric} \label{table:sota}
}
\end{center}
\end{table}
\vspace{-0.5\baselineskip}

In addition, we explain how to transform our PoK with trusted helper into 3-round PoK without helper or 5-round PoK without helper.
These two transformations offer different trade-offs between communication cost and security.
Using 5-round PoK lead to smaller signature sizes however the security proof of the Fiat-Shamir heuristic is less tight in this case.
In practice, this means that one have to take into account attacks such as the one from \cite{KZ20}.
We consider both transformations for our PoK~1 and denote the resulting PoK without helper by 3-round PoK~1 and 5-round PoK~1 respectively.
For the PoK~2 and PoK~3, we only consider the first transformation thus leading to PoK without helper denoted 3-round PoK~2 and 3-round PoK~3 respectively.
Furthermore, we present several optimizations for these PoK and describe how to convert them into signature schemes.
Our first signature is built from our 3-round PoK~1 and is denoted Sig~1~(3-round).
It features the most conservative design possible as it relies on the $\SD$ problem over $\Ft$ along with an underlying 3-round structure.
Our second signature is built from our 5-round PoK~1 and is denoted Sig~1~(5-round).
Both signatures can be instantiated using the plain $\SD$ problem or its quasi-cyclic variant $\QCSD$.
Sig~1~(3-round) and Sig~1~(5-round) both outperforms the Stern \cite{Stern93}, \cite{Veron97} and \cite{AGS11} schemes with respect to signature size for comparable settings at the cost of a small performance overhead. 
In addition, they achieve similar performances to \cite{BGS21} while relying on more conservative security assumptions.
Sig~2 and Sig~3 are constructed from 3-round PoK~2 and 3-round PoK~3 respectively and feature even smaller signature size however at the cost of a bigger performance overhead.
Sig~2 outperforms the recent proposal from \cite{GPS21} but is outperformed by the proposal from \cite{FJR21}.
Finally, Sig~3 close this performance gap by mixing PoK~2 with the shared permutation idea from \cite{FJR21}.

\vspace{\baselineskip}
\noindent \textbf{Related Work.}
Gueron, Persichetti and Santini have recently proposed a new code-based signature built from a PoK with trusted helper for the $\SD$ problem over $\Fq$ that achieve an arbitrarily small soundness error \cite{GPS21}.
Recently in an independent and concurrent work, Feneuil, Joux and Rivain have proposed a code-based signature based on a PoK for the $\SD$ problem over $\Ft$ that achieves an arbitrarily small soundness error \cite{FJR21}.
These works present some similarities with our PoK~2 and its associated signature Sig~2.
A PoK for the $\SD$ problem requires to prove two statements: (i) there exists a value $\bm{x}$ such that $\bm{H} \bm{x}^{\top} = \bm{y}^{\top}$ and (ii) the weight of $\bm{x}$ is small.
To achieve an arbitrarily small soundness error, one need to prove both statements at once which is challenging to do while preserving the zero-knowledge property of the underlying proof.
Indeed, one generally proves the first property by masking $\bm{x}$ using $\bm{u} + \bm{x}$ with a uniform random value $\bm{u}$ and prove the second property by masking $\bm{x}$ using $\pi[\bm{x}]$ for some random permutation $\pi$ while reconciling the two parts of the proof thanks to a third value such as $\pi[\bm{u} + \bm{x}]$.
The authors of \cite{GPS21} solve this issue by \emph{revealing the permutation and later canceling it} in their proof.
As such, their proposal reveals $\pi$ rather than $\pi[\bm{x}]$ contrarily to existing protocols.
The authors of \cite{FJR21} solve the aforementioned issue by introducing what they called a \emph{shared permutation} namely by masking the permutation during the computation of some permuted vectors.
Doing so, they are able to compute a value related to $\pi[\bm{u} + \bm{x}]$ from $\bm{u} + \bm{x}$ without revealing anything on $\pi$.
Our PoK~2 relies on another approach by \emph{introducing several permutations and revealing all of them but one} in order to prove the knowledge of the solution of a permuted $\SD$ problem instance.

We briefly discuss the main differences between these three approaches.
The PoK from \cite{GPS21} relies on the $\SD$ problem over $\Fq$ while our PoK~2 relies on the $\SD$ problem over $\Ft$.
As the $\SD$ problem over $\Ft$ has been arguably more studied than its counterpart over $\Fq$, it can be considered as a more conservative assumption.
%In addition, arithmetic over $\Ft$ is generally easier and faster to perform than arithmetic over $\Fq$.
%One should nonetheless note that this can be mitigated on platforms providing specific instructions sets such as the GF-NI instructions for recent Intel processors as mentionned in \cite{GPS21}.
Furthermore, using the protocol from \cite{GPS21}, one has to send the permutation $\pi$ (which is fixed hence not replaceable by a seed) to prove the weight of $\bm{x}$ while our protocol only requires to send $\pi[\bm{x}]$ (which is a small weight vector hence can be compressed).
As sending a permutation of a vector of size $n$ over $\Fq$ is costly, the communication cost associated to the GPS proposal is bigger than the communication cost of our PoK~2.
%Besides, this core difference is amplified within the protocols as several executions of the underlying PoK are computed in order to achieve a negligible soundness error.
In practice, this means that for comparable parameters, our Sig~2 outperforms the signature from \cite{GPS21}.
%For instance, considering the first parameter set from \cite{GPS21}, one has to execute $23$ instances of the underlying PoK thus leading to a signature size equal to $27$~kB.
%For a comparable parameter set (also requiring $23$ executions of our PoK~2), our Sig~2 features a size of $16$~kB.
%We defer the reader to Section~\ref{sec:parameters} for a complete comparison between both schemes.

The PoK from \cite{FJR21} and our PoK~2 are more closely related as they are both based on the $\SD$ problem over $\Ft$ and both achieve an arbitrarily small soundness error equal to $1/N$.
As such, they are equivalent from a theoretical point of view.
Nonetheless, the optimized version of the FJR protocol outperforms the optimized version of our PoK~2 in practice.
This is explained by the fact that some optimizations related to commitment compression bring a better improvement for the FJR protocol than for our PoK~2.
As a result, we also introduce PoK~3 which mixes PoK~2 with the shared permutation setting of \cite{FJR21}.
Doing so, one can consider that our PoK~3 is a dual version (Véron-like) of the protocol from \cite{FJR21}.

\vspace{\baselineskip}
\noindent \textbf{Paper Organization.} We start by describing some preliminaries related to code-based cryptography and PoK in Section~\ref{sec:preliminaries}.
We present our new PoK with trusted helper in Section~\ref{sec:pok1} and \ref{sec:pok23} respectively.
Then, we explain how to remove the trusted helper from the aforementioned protocols in Section~\ref{sec:no-helper}.
Several optimizations reducing the bandwidth cost of these PoK are described in Section~\ref{sec:optimizations}.
We explain how to transform our PoK into signature schemes in Section~\ref{sec:signatures}.
Parameters for these new signatures are provided in Section~\ref{sec:parameters} along with a comparison to existing code-based signatures.
To finish, we discuss some generalizations and variants of our PoK in Section~\ref{sec:generalization}.

%--------------------------------------------------------------------%
\section{Preliminaries} \label{sec:preliminaries}
%--------------------------------------------------------------------%

\vspace{0.5\baselineskip}
\noindent \textbf{Notations.} Hereafter, vectors (respectively matrices) are represented using bold lower-case (respectively upper-case) letters.
Also, the vectors are assumed to be row vectors by default, and
we denote the column vectors by transpose of row of vector (such as $\bm{x}^T$).
The Hamming weight (number of non-zero coordinates) of a vector $\bm{x}$ is denoted by $\hw{\bm{x}}$.
For an integer $n > 0$, we use $\hperm{n}$ to denote the symmetric group of all permutations of $n$ elements.
For a finite set $S$, $x \sampler S$ denotes that $x$ is sampled uniformly at random from $S$ while $x \samples{\theta}$ denotes that $x$ is sampled uniformly at random from $S$ using the seed $\theta$.
In addition, we use the acronym $\ppt$ as an abbreviation for the term ``probabilistic polynomial time''.
We also call a function \emph{negligible} and denote it by
$\negl(\cdot)$ if for all sufficiently large $\lambda \in \mathbb{N}$,
$\negl(\lambda) < \lambda^{-c}$, for all constants $c > 0$.

\subsection{Code-based Cryptography} \label{sec:preliminaries:cbc}

We start by defining binary linear codes and quasi-cyclic codes.
Then, we describe the syndrome decoding ($\SD$) and general syndrome decoding ($\GSD$) problems which are hard problems commonly used in code-based cryptography.
These problems are equivalent and differ only in the way used to represent the underlying code namely using a parity-check matrix in the former and using a generator matrix in the latter. 
The $\SD$ problem has been proven NP-complete in \cite{BMVT78}.
In addition, we also introduce the quasi-cyclic problems $\QCSD$ and $\QCGSD$ which are structured variants of the $\SD$ and $\GSD$ problems.

\begin{definition}[Binary Linear Code]
Let $n$ and $k$ be positive integers such that $k < n$. A binary linear $\mathcal{C}$ code (denoted $[n,k]$) is a $k$-dimensional subspace of $\Ftn$.
  $\mathcal{C}$ can be represented in two equivalent ways: by a generator matrix $\bm{G} \in \Ft^{\ktn}$ such that $\mathcal{C} = \{\bm{mG} ~|~ \bm{m} \in \Ftk\}$ or by a parity-check matrix $\bm{H} \in \Ft^{\nmktn}$ such that $\mathcal{C} = \{ \bm{x} \in \Ftn ~|~ \bm{H}\bm{x}^{\top} = 0\}$.
\end{definition}

\begin{definition}[Systematic Binary Quasi-Cyclic Code]
  A systematic binary quasi-cyclic code of index $\ell$ and rate $1/\ell$ is a
  $[n=\ell k, k]$ code that can be represented by a $k \times \ell k = k \times n$ generator matrix $\bm{G} \in \mathcal{QC}(\Ft^{\ktn})$ of the form:
\begin{equation*}
  \mathbf{G}=
\begin{bmatrix}
  \mathbf{I}_k & \mathbf{A}_0 & \cdots & \mathbf{A}_{\ell-2}\\
\end{bmatrix}
\end{equation*}
  where $\mathbf{A}_0,\ldots ,\mathbf{A}_{\ell-2}$ are circulant $k\times k$ matrices. Alternatively, it can be represented by an $(\ell-1)k \times \ell k = (n-k)\times n$ parity check matrix $\bm{H} \in \mathcal{QC}(\Ft^{(n - k) \times k})$ of the form:
\begin{equation*}
  \mathbf{H}=
\begin{bmatrix}
  \mathbf{I}_k & \cdots & 0 & \mathbf{B}_0\\
               & \ddots &   & \vdots \\
  0            & \cdots & \mathbf{I}_k & \mathbf{B}_{\ell-2}
\end{bmatrix}
\end{equation*}
where $\mathbf{B}_0,\ldots ,\mathbf{B}_{\ell-2}$ are circulant $k\times k$ matrices.
\end{definition}

\begin{definition}[$\SD$ problem]
  Given positive integers $n$, $k$, $w$, a random parity-check matrix $\bm{H} \sampler \Ft^{\nmktn}$ and a syndrome $\bm{y} \in \Ft^{\nmk}$, the syndrome decoding problem $\SD(n,k,w)$ asks to find $\bm{x} \in \Ftn$ such that $\bm{Hx}^\top = \bm{y}^\top$ and $\hw{\bm{x}} = w$.
\end{definition}

\begin{definition}[$\GSD$ problem]
  Given positive integers $n$, $k$, $w$, a random generator matrix $\bm{G} \sampler \Ft^{\ktn}$ and a vector $\bm{y} \in \Ftn$, the general syndrome decoding problem $\GSD(n,k,w)$ asks to find $(\bm{x}, \bm{e}) \in \Ftk \times \Ftn$ such that $\bm{xG} + \bm{e} = \bm{y}$ and $\hw{\bm{e}} = w$.
\end{definition}

\begin{definition}[$\ell$-$\QCSD$ problem]
  Given positive integers $n, k, w$, with $n=\ell k$ for some $\ell$, a random parity-check matrix of a quasi-cyclic code $\bm{H} \sampler \mathcal{QC}(\Ft^{(n - k) \times n})$ and a syndrome $\bm{y} \in \Ft^{(n-k)}$, the syndrome decoding problem $\ell$-$\QCSD(n,k,w)$ asks to find $\bm{x} \in \Ft^{n}$, such that $\bm{Hx}^\top = \bm{y}^\top$ and $\hw{\bm{x}} = w$.
\end{definition}

\begin{definition}[$\ell$-$\QCGSD$ problem]
  Given positive integers $n, k, w$, with $n=\ell k$ for some $\ell$, a random generator matrix of a quasi-cyclic code $\bm{G} \sampler \mathcal{QC}(\Ft^{\ktn})$ and a vector $\bm{y} \in \Ftn$, the general syndrome decoding problem $\ell$-$\QCGSD(n,k,w)$ asks to find $(\bm{x}, \bm{e}) \in \Ftk \times \Ftn$ such that $\bm{xG} + \bm{e} = \bm{y}$ and $\hw{\bm{e}} = w$.
\end{definition}

\subsection{Commitments Schemes} \label{sec:preliminaries:cmt}
We now introduce commitment schemes as they are building blocks commonly used to construct proofs of knowledge. 
We require such schemes to be both hiding and binding. 
The former property ensures that the commitment does not leak any information on the committed message while the latter ensures that adversaries can not change their committed messages once the commitment is sent.
We now present the formal definition of the commitment schemes.
%In practice, these commitment schemes are instantiated using hash functions modeled as random oracles.

\begin{definition}[Commitment Scheme] \label{def:cmt-scheme}
  A (non-interactive) commitment scheme with underlying 
  message space $\mathcal{M}$ is tuple of algorithms
  $(\keygen, \mathsf{Com}, \mathsf{Open})$ such that:
  \begin{itemize}
      \item $\mathbf{\keygen}$: Takes the security parameter
      $\lambda$ as input and outputs pair of keys.
      $(\mathsf{gk}, \mathsf{ck}) \samplen \keygen(1^\lambda)$.
      Here $\mathsf{gk}$ is called the \emph{setup key} and it serves as
      an implicit input to the $\mathsf{Com}$ and $\mathsf{Open}$
      algorithms, whereas $\mathsf{ck}$ is called
      \emph{commitment key} and it is given to the sender.
      Note that, the commitment key $\mathsf{ck}$ can be
      set to empty string $\varepsilon$, if only setup key
      $\mathsf{gk}$ is sufficient for committing to the messages.
      
      \item $\mathbf{\mathsf{Com}}$: Takes a message $m \in \mathcal{M}$
      and commitment key $\mathsf{ck}$ as input and output
      a commitment $c$ and opening $d$. Formally, 
      $(c,d) \samplen \commit{\mathsf{ck},m}$.
      
      \item $\mathbf{\mathsf{Open}}$: Takes a commitment $c$,
      opening $d$, message $m \in \mathcal{M}$ as input
      and outputs a bit $b \in \lbrace 0, 1 \rbrace$ indicating
      whether the commitment $c$ is a valid commitment of $m$.
      Formally, $b := \open{c,d,m}$.
      
  \end{itemize}
  
  \noindent The commitment scheme is \emph{perfectly correct} if
  $ \forall \lambda \in \mathbb{N}, \, \forall m \in \mathcal{M} $
  and for all valid key pairs $(\mathsf{gk}, \mathsf{ck})$,
  $$ \Pr[\open{\commit{\mathsf{ck},m},m} = 1] = 1. $$
\end{definition}

The commitment scheme satisfies two
security properties guaranteeing security from
malicious sender (prover) and from
malicious receiver (verifier):

\begin{itemize}
    \item \emph{Hiding: } It is computationally hard
    for an efficient adversary $ \mathcal{A} $ to generate two distinct
    messages $m_0, m_1 \in \mathcal{M}$, such that $\mathcal{A}$
    can distinguish between their respective commitments.
    Formally, for any PPT adversary $\mathcal{A}$ it should hold that,
    $$ \Pr \left[b = b' \quad  \bigg \vert  \quad \begin{matrix}
        (\mathsf{gk}, \mathsf{ck}) \samplen \keygen(1^\lambda), \,
        (m_0, m_1) \samplen \mathcal{A}(\mathsf{gk}, \mathsf{ck})\\ 
        b \sampler \lbrace 0, 1 \rbrace, 
        (c,d) \samplen \commit{m_b, \mathsf{ck}},
        b' \samplen \mathcal{A}(c)
    \end{matrix} \right]  = \frac{1}{2} + \mathsf{negl}(\lambda).$$
    
    \item \emph{Binding: } It is computationally hard for an
    efficient adversary to generate a triple $(c,d,d')$ such that
    \emph{both} $(c,d)$ and $(c, d')$ are valid commitment/opening
    pairs for some $m, m' \in \mathcal{M}$ respectively, where
    $m \neq m'$.
    Formally, for any PPT adversary it should hold that,
    $$\Pr \left[\begin{matrix}
      m \neq m' \, \bigwedge \\ m, m' \in \mathcal{M}
    \end{matrix} \quad \bigg \vert \quad \begin{matrix} 
    (\mathsf{gk}, \mathsf{ck}) \samplen \keygen(1^\lambda), \,
    (c, d, d') \samplen \mathcal{A}(\mathsf{gk}, \mathsf{ck})\\ 
    1 := \open{c,d,m}, 
    1 := \open{c,d',m'} \end{matrix} \right] \leq \mathsf{negl}(\lambda) .$$

\end{itemize}

In this work, we assume that the commitment scheme is implemented using
a collision-resistant hash function $H$ modelled as random oracle.
To commit to a message $m$, we first sample a random
value $r \leftarrow \bit^\lambda$ and compute the commitment as
$c := H(r,m)$. The hiding follows since $H$ is modelled as 
random oracle and the binding follows from the collision-resistance
of $H$. The random value $r$ serves as the opening $d$. The verifier
can simply re-compute $H(r,m)$ on receiving $r$ as opening
and check if $H(r,m)$ equals the commitment $c$.

\ignore{
\begin{definition}[Commitment Scheme]
  A commitment scheme is a function $\mathsf{Com}: \bit^{\lambda} \times \bit^{*} \rightarrow \bit^{2\lambda}$ that takes as input $\lambda$ uniformly random bits $r$ where $\lambda$ is the security parameter as well as a message $m \in \bit^{*}$ and outputs a commitment $\commit{r, m}$ of size $2\lambda$ bits.
\end{definition}

\begin{definition}[Computationally Hiding]
Let $(m_0, m_1)$ be a pair of messages, the advantage of an adversary $\adv$ against the commitment hiding experiment is:
\begin{equation*}
  \mathsf{Adv}^{\mathsf{Hiding}}_{\mathsf{Com}, \adv}(1^\lambda) = 
  \Bigg| \, \prb \left[ \begin{array}{l}
    b = b' \\
  \end{array} \ \middle\vert~ \\ \begin{array}{l}
    b \sampler \bit, \, r \sampler \bit^{\lambda} \\
    b' \samplen \adv.\mathsf{guess}(\commit{r, \, m_b})
  \end{array}\right]
  - \frac{1}{2} \, \Bigg|.
\end{equation*}
A commitment scheme $\mathsf{Com}$ is computationally hiding if for all $\ppt$ adversaries $\adv$ and every pair of messages $(m_0, m_1)$, $\mathsf{Adv}^{\mathsf{Hiding}}_{\mathsf{Com}, \adv}(1^\lambda)$ is negligible in $\lambda$.
\end{definition}

\begin{definition}[Computationally Binding]
Let the advantage of an adversary $\adv$ against the commitment binding experiment be defined as:
\begin{equation*}
  \mathsf{Adv}^{\mathsf{Binding}}_{\mathsf{Com}, \adv}(1^\lambda) = \prb
  \left[ \begin{array}{l}
    m_0 \neq m_1 \\ 
    c_0 = c_1 \\
  \end{array} \ \middle\vert~ \\ \begin{array}{l}
    (r_0, r_1, m_0, m_1) \samplen \adv.\mathsf{choose}(1^{\lambda}) \\
    c_0 = \commit{r_0, \, m_0}, ~ c_1 = \commit{r_1, \, m_1} \\ 
  \end{array}\right].
\end{equation*}
A commitment scheme $\mathsf{Com}$ is computationally binding if for all $\ppt$ adversaries $\adv$, $\mathsf{Adv}^{\mathsf{Binding}}_{\mathsf{Com}, \adv}(1^\lambda)$ is negligible in $\lambda$.
\end{definition}
}

\subsection{Proofs of Knowledge with Helper} \label{sec:preliminaries:pok}

Following the work of Katz, Kolesnikov and Wang \cite{KKW18}, Beullens introduced the notion of sigma protocols with helper in \cite{Beullens20}.
Given a relation $R = (x, w)$, these Honest-Verifier Zero-Knowledge Proofs of Knowledge (HVZK PoK) allow a prover to convince an honest verifier (namely a verifier that follows the protocol as described) that it knows a witness $w$ for the statement $x$ without revealing anything on $w$.
In our context, the relation $R = (x, w)$ is defined by an instance of the $\SD$ problem such that $x = (\bm{H}, \bm{y})$ and $w = \bm{x}$ namely the prover convinces the verifier that he knows a solution to an $\SD$ instance without revealing anything on its solution.
Alternatively, when the $\GSD$ form of the problem is considered, one has $x = (\bm{G}, \bm{y})$ and $w = (\bm{x}, \bm{e})$.

\begin{definition}[Sigma Protocol with Helper \cite{Beullens20}]
  A protocol is a Sigma Protocol with helper for relation $R$ with challenge space $\mathcal{C}$ if it follows the form of Figure~\ref{preliminaries-pok-fig1} and satisfies:

\begin{itemize}
  \item[$\bullet$] \textbf{Completeness.} If all parties ($\helperf$, $\proverf$ and $\verifierf$) follow the protocol on input $(x, w) \in R$, then the verifier always accepts.

  \item[$\bullet$] \textbf{Special soundness.} From an adversary $\adv$ that outputs two valid transcripts $(x, \aux, \com, \alpha, \rsp)$ and $(x, \aux, \com, \alpha', \allowbreak \rsp')$ with $\alpha \neq \alpha'$ and where $\aux = \setup(\theta)$ for some seed value $\theta$ (not necessarily known to the extractor), there exists an extractor $\Ext$ that efficiently extracts a witness $w$ such that $(x, w) \in R$ with probability $1 - \negl(\lambda)$.

  \item[$\bullet$] \textbf{Special honest-verifier zero-knowledge.} There exists a $\ppt$ simulator $\Sim$ that on input $x$, a random seed value $\theta$ and a random challenge $\alpha$ outputs a transcript $(x, \aux, \com, \alpha, \rsp)$ with $\aux = \setup(\theta)$ that is computationally indistinguishable from the probability distribution of transcript of honest executions of the protocol on input $(x, w)$ for some witness $w$ such that $(x, w) \in R$, conditioned on the auxiliary information being equal to $\aux$ and the challenge being equal to $\alpha$.
\end{itemize}
\end{definition}

\begin{figure}[!ht]
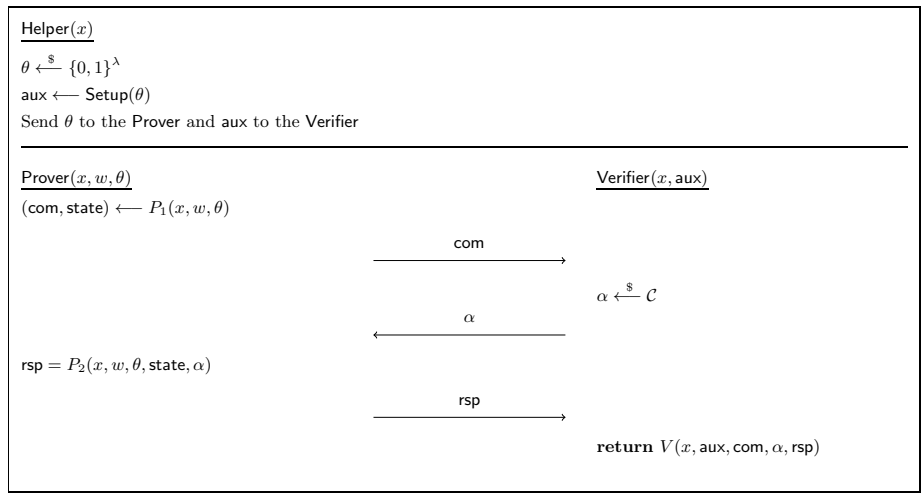
 
  \begin{center}
    \resizebox{1\textwidth}{!}{\fbox{
      \pseudocode{%
        \hspace{160pt} \> \> \hspace{160pt} \\[-0.75\baselineskip] % Hack to force font scaling
        \underline{\helperf(x)} \\[0.1\baselineskip]
        \theta \sampler \bit^{\lambda} \\ 
        \aux \samplen \setup(\theta) \\ 
        \tx{Send } \theta \tx{ to the } \proverf \tx{ and } \aux \tx{ to the } \verifierf \\[0.5\baselineskip][\hline]\\[-0.25\baselineskip]
        \underline{\proverf(x, w, \theta)} \> \> \underline{\verifierf(x, \aux)} \\
        (\com, \state) \samplen P_1(x, w, \theta) \\
        \> \sendmessageright*{\com} \\
        \> \> \alpha \sampler \mathcal{C} \\[-0.5\baselineskip]
        \> \sendmessageleft*{\alpha} \\
        \rsp = P_2(x, w, \theta, \state, \alpha) \\
        \> \sendmessageright*{\rsp} \\
        \> \> \pcreturn V(x, \aux, \com, \alpha, \rsp) \\
      }
    }}
    \caption{Structure of an HVZK PoK with Trusted Helper \cite{Beullens20} \label{preliminaries-pok-fig1}}
  \end{center}
\end{figure}

\subsection{Signature schemes} \label{sec:preliminaries:sig}

In this section, we define signatures based on the Fiat-Shamir transform~\cite{FS} 
and then present the security definitions associated to these schemes.

\begin{definition}[Digital Signature] \label{def:sig}
	A digital signature scheme $\mathsf{SIG}$ is a tuple of algorithms
	$\mathsf{SIG} = (\mathsf{Gen}, \mathsf{Sign}, \mathsf{Verify})$,
	\begin{itemize}
		\item $\mathsf{Gen}$ takes security parameter $\lambda$ as input
		and outputs the key pair $(\mathsf{pk}, \, \mathsf{sk})$.
		\item $\mathsf{Sign}$ takes a message $m$ along with the secret
		(signing) key $\mathsf{sk}$ as input and produces signature $\sigma$ as output.
		\item The verification algorithm $\mathsf{Verify}$ takes the public (verification)
		key $\mathsf{pk}$, message $m'$, and signature $\sigma$ as input
		and returns $\mathsf{accept}$ or $\mathsf{reject}$.
	\end{itemize}
\end{definition}

\noindent A signature scheme $\mathsf{SIG}$ is said to have correctness error $\varepsilon$
if for all $(\mathsf{pk}, \mathsf{sk}) \leftarrow \mathsf{Gen}(1^\lambda)$,
and all messages $m \in \mathcal{M}$, it holds true that
$$ \prb \left[ \mathsf{Verify}(\mathsf{pk}, \mathsf{Sign}(\mathsf{sk},m), m) = \mathsf{reject}  \right] \leq \varepsilon.$$

%We call the signature scheme as \emph{perfectly correct} if the correctness error is $0$.

\begin{definition} \label{def:inst-gen}
	A binary relation $R$ with instance generator $\mathsf{IG}$ is called \emph{hard}
	if for any (quantum) adversary $\adv$, it holds that
	$$ \prb \left[  (x, \tilde{w}) \in R \, \vert \, (x, w) \leftarrow  \mathsf{IG}, \, \tilde{w} \leftarrow \adv(x)  \right] $$
	is negligible, for any $\mathsf{IG}$ that
	always outputs a valid pair $(x,w) \in R$.
\end{definition}

\begin{definition} [Fiat-Shamir Signature] \label{def:sig-FS}
	A Fiat-Shamir signature scheme based on a public-coin interactive proof system
	(or $\Sigma$-protocol) $ \Pi = (\mathsf{Prover}, \mathsf{Verifier})$
	for a hard relation $R$ with instance generator $\mathsf{IG}$, denoted by
	$\mathsf{SIG}[\Pi]$ is a tuple of algorithms
	$\mathsf{SIG}[\Pi] = (\mathsf{GenFS}, \mathsf{SignFS}, \mathsf{VerifyFS})$,
	\begin{itemize}
		\item $\mathsf{GenFS}$ samples 
		 $(x,w) \leftarrow \mathsf{IG}$ then outputs $\mathsf{sk} := (x,w)$ and
		 $\mathsf{pk} := x$.
		 \item $\mathsf{SignFS}^H(\pk, \mathsf{sk}, m)$ outputs $(m, \sigma)$
		 where $ \sigma \leftarrow \mathsf{Prover}^H(x, w, m)$.
		 \item $\mathsf{VerifyFS}^H(\mathsf{pk}, \sigma, \tilde{m})$ runs
		 $\mathsf{Verifier}^H(x, \sigma, \tilde{m})$ and returns its output.
	\end{itemize}
\end{definition}

\noindent Hereafter, we assume that the algorithms have oracle access to hash function $H$
which is modeled as (quantum) random oracle. 

%Note that, $\mathsf{Prover}^H(x||m,w)$ simply ignores the message $m$ and runs $\mathsf{Prover}(x)$ with
%oracle access to $H$. Similarly, $\mathsf{Verifier}^H(x||m,\sigma)$ simply ignores the message $m$ and runs $\mathsf{Verifier}(x,\sigma)$ with
%oracle access to $H$.

\begin{definition} [Strong Existentially Unforgeable Signatures under Chosen Message Attack ($\mathsf{sEUF-CMA}$)] \label{def:strong-euf-cma}
	A signature scheme possesses 
	\emph{strong existential unforgeability under chosen message attack} ($\mathsf{sEUF-CMA}$)
	if for all (quantum) polynomial-time algorithms $\adv$ and for uniformly
	random $H$, it holds that
  $$ \prb \left[ \mathsf{Verify}^H(\mathsf{pk}, \sigma, m) \, \bigwedge \, (m, \sigma) \notin \mathbf{Sign-q} \, \vert \, (\mathsf{pk}, \mathsf{sk}) \leftarrow \mathsf{Gen},  \, (m, \sigma) \leftarrow \adv^{H,\mathbf{Sign}}(\mathsf{pk}) \right] $$
	is negligible. Here, $\mathbf{Sign}$ is a classical oracle which on input $m$
	returns $\mathsf{Sign}^H(\pk, \allowbreak \mathsf{sk}, m)$ and $\mathbf{Sign-q}$ is the list
	of all ($q$) queries made to $\mathbf{Sign}$.
\end{definition}

% Definitions of quantum random oracle: cite{C:DonFehMaj20}.

%--------------------------------------------------------------------%
\section{PoK 1 - Stern Protocol Improvement} \label{sec:pok1}
%--------------------------------------------------------------------%

The first PoK for the $\SD$ problem over $\Ft$ was introduced by Stern in 1993~\cite{Stern93}.
This 3-round protocol features a soundness error equal to $2/3$ and as such needs to be repeated several times in order to achieve a negligible soundness error.
Over the years, 5-round code-based PoK reducing the soundness error to $1/2$ (hence providing smaller communication costs) have been proposed.
Such protocols either rely on the $\SD$ problem over $\Fq$ \cite{CVE11} or leverage the structured $\QCGSD$ and $\QCSD$ problems over $\Ft$ \cite{AGS11, BGS21}.
Hereafter, we introduce a PoK for the $\SD$ problem over $\Ft$ with soundness error equal to $1/2$.
Our new protocol (denoted PoK~1) can be either seen as (i) a modification of the initial Stern protocol leveraging the MPC-in-the-head paradigm along with several optimizations from \cite{BGS21} or as (ii) the BGS protocol \cite{BGS21} in which the quasi-cyclicity is replaced by the use of the MPC-in-head technique.

The initial Stern protocol permits to prove the knowledge of $\bm{x}$ such that $\bm{y}^\top = \bm{H} \bm{x}^\top$ and $\hw{\bm{x}} = w$.
Within the protocol, one proves the knowledge of $\bm{x}$ using $\bm{x} + \bm{u}$ for some random value $\bm{u}$ and prove that $\hw{\bm{x}} = w$ using $\pi[\bm{x}]$ for some random permutation $\pi$.
To this end, the prover starts by generating three commitments related to $\pi, \bm{x}$ and $\bm{u}$.
Next, the verifier samples a random challenge from $\{0, 1, 2\}$ and the prover outputs a response that is specific to the received challenge.
Amongst these three possible responses, one can be computed without knowing the secret $\bm{x}$ hence could be verified using the MPC-in-the-head paradigm,
Doing so, one can reduce the challenge space to $\bit$ thus achieving a soundness error equal to $1/2$.
Our PoK~1 follows this approach and is depicted as a sigma protocol with helper in Figure~\ref{pok1:fig1}.

We explain in Section \ref{sec:no-helper} how to remove the helper from Figure~\ref{pok1:fig1} in order to get both a 3-round HVZK PoK and a 5-round HVZK PoK.
Our 3-round PoK~1 features a very conservative design ($\SD$ assumption over $\Ft$ only, tighter Fiat-Shamir transformation proof thanks to the 3-round structure) therefore is comparable to the Stern \cite{Stern93} and Véron \cite{Veron97} proposals which provide the same security guarantees.
Our 3-round PoK~1 benefits from a smaller signature size than the Stern and Véron protocols at the cost of a small performance overhead due to the use of the MPC-in-the-head.
When coupled with quasi-cyclicity, our 5-round PoK~1 is comparable to the AGS \cite{AGS11} and BGS \cite{BGS21} protocols ($\QCSD$ assumption over $\Ft$, 5-round structure) while being more conservative security-wise as it relies on the $\QCSD$ problem directly rather than the $\DiffSD$ problem contrarily to the AGS and BGS protocols (see \cite{BGS21}, Definition~11 for a description of the $\DiffSD$ problem).
Similarly to the 3-round case, our 5-round PoK~1 features a smaller signature size than the AGS and BGS protocols at the cost of a small performance overhead.

\begin{figure}[!ht]
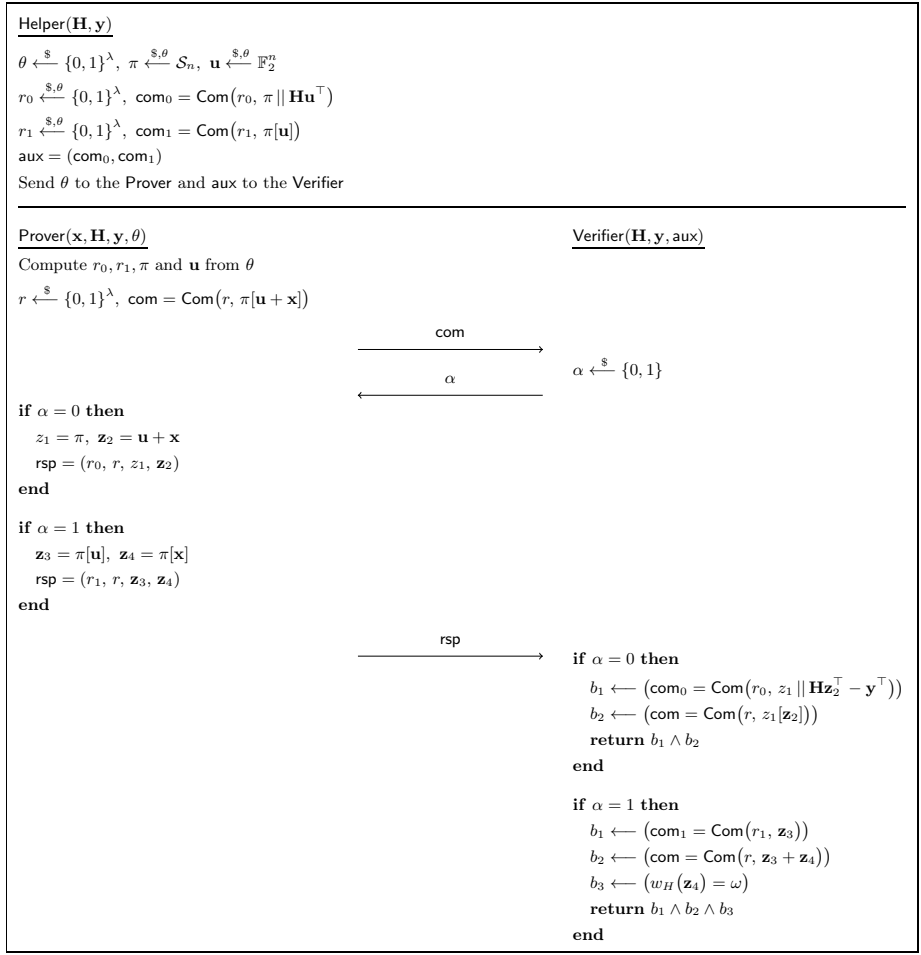
 
  \begin{center}
    \resizebox{1\textwidth}{!}{\fbox{
      \pseudocode{%
        \hspace{160pt} \> \> \hspace{160pt} \\[-0.75\baselineskip] % Hack to force font scaling
        \underline{\helperf(\bm{H}, \bm{y})} \\[0.1\baselineskip]
        \theta \sampler \bit^{\lambda}, ~\pi \samples{\theta} \hperm{n}, ~ \bm{u} \samples{\theta} \Ftn \\ 
        r_0 \samples{\theta} \bit^{\lambda}, ~ \com_{0} = \commit{r_0, \, \pi \, || \, \bm{H} \bm{u}^{\top}} \\
        r_1 \samples{\theta} \bit^{\lambda}, ~ \com_{1} = \commit{r_1, \, \pi[\bm{u}]} \\
        \aux = (\com_0, \com_1) \\
        \tx{Send } \theta \tx{ to the } \proverf \tx{ and } \aux \tx{ to the } \verifierf \\[0.5\baselineskip][\hline]\\[-0.25\baselineskip]
        \underline{\proverf(\bm{x}, \bm{H}, \bm{y}, \theta)} \> \> \underline{\verifierf(\bm{H}, \bm{y}, \aux)} \\
        \tx{Compute } r_0, r_1, \pi \tx{ and } \bm{u} \tx{ from } \theta \\
        r \sampler \bit^{\lambda}, ~ \com = \commit{r, \, \pi[\bm{u} + \bm{x}]}  \\
        \> \sendmessageright*{\com} \\[-0.5\baselineskip]
        \> \> \alpha \sampler \bit \\[-\baselineskip]
        \> \sendmessageleft*{\alpha} \\[-0.5\baselineskip]
        \pcif \alpha = 0 \pcthen \\
        \pcind z_1 = \pi, ~ \bm{z}_2 = \bm{u} + \bm{x} \\
        \pcind \rsp = (r_0, \, r, \, z_1, \, \bm{z}_2) \\
        \pcend \\[0.5\baselineskip]
        \pcif \alpha = 1 \pcthen \\
        \pcind \bm{z}_3 = \pi[\bm{u}], ~ \bm{z}_4 = \pi[\bm{x}] \\ 
        \pcind \rsp = (r_1, \, r, \, \bm{z}_3, \, \bm{z}_4) \\
        \pcend \\
        \> \sendmessageright*{\rsp} \\[-\baselineskip]
        \> \> \pcif \alpha = 0 \pcthen \\
        \> \> \pcind b_1 \samplen \big( \com_0 = \commit{r_0, \, z_1 \, || \, \bm{H} \bm{z}_2^{\top} - \bm{y}^{\top}} \big) \\ 
        \> \> \pcind b_2 \samplen \big( \com = \commit{r, \, z_1[\bm{z}_2]} \big) \\ 
        \> \> \pcind \pcreturn b_1 \wedge b_2 \\
        \> \> \pcend \\[0.5\baselineskip]
        \> \> \pcif \alpha = 1 \pcthen \\
        \> \> \pcind b_1 \samplen \big( \com_1 = \commit{r_1, \, \bm{z}_3} \big) \\ 
        \> \> \pcind b_2 \samplen \big( \com = \commit{r, \, \bm{z}_3 + \bm{z}_4} \big) \\ 
        \> \> \pcind b_3 \samplen \big( \hw{\bm{z}_4} = \omega \big) \\
        \> \> \pcind \pcreturn b_1 \wedge b_2 \wedge b_3 \\
        \> \> \pcend
      }
    }}
    \caption{ZK PoK with Helper for the $\SD$ problem over $\Ft$ \label{pok1:fig1}}
  \end{center}
\vspace{-1.5\baselineskip}
\end{figure}

\begin{theorem}[\textbf{Proof of knowledge with helper}]
  If the commitment used is binding and hiding, then the protocol depicted in Figure~\ref{pok1:fig1} is a proof of knowledge with helper for the $\SD$ problem with challenge space $\mathcal{C}$ such that $|\mathcal{C}| = 2$.
\end{theorem}

\begin{proof} We need to prove that the protocol in Figure~\ref{pok1:fig1}
satisfies the properties of correctness, special soundness, and special honest-verifier zero-knowledge.

\vspace{\baselineskip}
  \noindent \textbf{Correctness.} The correctness follows straightforwardly from the protocol description once the commitments are verified.

\vspace{\baselineskip}
  \noindent \textbf{Special soundness.} Given an adversary $\adv$ that outputs with non negligible probability two valid transcripts $(\bm{H}, \bm{y}, \aux, \com, \alpha, \rsp)$ and $(\bm{H}, \bm{y}, \aux, \com, \alpha', \allowbreak \rsp')$ with $\alpha \neq \alpha'$ and where $\aux = \setup(\theta)$ for some random seed $\theta$, one can easily build a knowledge extractor $\Ext$ that returns a solution to the $\SD$ instance defined by $(\bm{H}, \bm{y})$:

\vspace{0.5\baselineskip} 
\pseudocode{%
  \tx{1. Compute and output } z_1^{-1}[\bm{z}_4].
}
\vspace{0.5\baselineskip} 

\noindent As $\alpha \neq \alpha'$, the extractor $\Ext$ has access to both transcripts $(z_1, \bm{z}_2)$ and $(\bm{z}_3, \bm{z}_4)$ therefore he can output $z_1^{-1}[\bm{z}_4]$.
We now explain why the extractor's output is a solution to the considered $\SD$ problem instance.
Using the binding property of $\com_{0}$ and $\com_{1}$, one has $z_1 = \pi$ and $\bm{H} \bm{z}_2^{\top} - \bm{y}^{\top} = \bm{H} \bm{u}^{\top}$ as well as $z_3 = \pi[\bm{u}]$.
In addition, from the binding property of $\com$, one has $z_1[\bm{z}_2] = \bm{z}_3 + \bm{z}_4$ thus $\bm{z}_2 = \bm{u} + \pi^{-1}[\bm{z}_4]$.
Using this expression within $\bm{H} \bm{z}_2^{\top} - \bm{y}^{\top} = \bm{H} \bm{u}^{\top}$, one can deduce that $\bm{H} (\pi^{-1}[\bm{z}_4])^{\top} = \bm{y}^{\top}$.
Given that $\hw{\bm{z}_4} = \omega$, one also has $\hw{\pi^{-1}[\bm{z}_4]} = \omega$ thus $z_1^{-1}[\bm{z}_4]$ is a solution to the considered $\SD$ problem instance.
Finally, $\Ext$ runs in polynomial time which completes the proof.

\vspace{\baselineskip}
\noindent \textbf{Special Honest-Verifier Zero-Knowledge.} 
We start by explaining why valid transcripts don't leak anything on the secret.
A valid transcript contains either $(\pi, \bm{u} + \bm{x})$ or $(\pi[\bm{u}], \pi[\bm{x}])$ namely the secret $\bm{x}$ is masked either by a random value $\bm{u}$ or by some random permutation $\pi$.
  We now explain how to build a $\ppt$ simulator $\Sim$ that given $(\bm{H}, \bm{y})$, a random seed $\theta$ and a random challenge $\alpha$ outputs a transcript $(\bm{H}, \bm{y}, \aux, \com, \alpha, \rsp)$ such that $\aux = \setup(\theta)$ that is indistinguishable from the probability distribution of transcripts of honest executions of the protocol:

\vspace{0.5\baselineskip} 
\pseudocode{%
  \tx{1. Compute } (r_0, r_1, \bm{u}, \pi) \tx{ from } \theta \\
  \tx{2. If } \alpha = 0, \tx{ compute } \bm{\tilde{x}} \tx{ such that } \bm{H} \tilde{\bm{x}} = \bm{y} \tx{ (without constraint on the weight of $\bm{x}$)} \\
  \hspace{11.5pt} \tx{If } \alpha = 1, \tx{ compute } \bm{\tilde{x}} \sampler \swset{\omega}{\Ftn} \\ 
  \tx{3. Compute } r \sampler \bit^{\lambda}, ~ \tilde{\com} = \commit{r, \, \pi[\bm{u} + \bm{\tilde{x}}]} \\
  \tx{4. If } \alpha = 0, \tx{ compute } z_1 = \pi, ~ \bm{\tilde{z}}_2 = \bm{u} + \bm{\tilde{x}} \tx{ and } \tilde{\rsp} = (r_0, r, z_1, \bm{\tilde{z}}_2) \\ 
  \hspace{11.5pt} \tx{If } \alpha = 1, \tx{ compute } \bm{z}_3 = \pi[\bm{u}], ~ \bm{\tilde{z}}_4 = \pi[\bm{\tilde{x}}] \tx{ and } \tilde{\rsp} = (r_1, r, \bm{z}_3, \bm{\tilde{z}}_4) \\ 
  \tx{5. Output } (\bm{H}, \bm{y}, \aux, \tilde{\com}, \alpha, \tilde{\rsp})\\[-0.5\baselineskip]
}
\vspace{0.5\baselineskip} 

  \noindent The transcript generated by the simulator $\Sim$ is $(\bm{H}, \bm{y}, \aux, \tilde{\com}, \alpha, \tilde{\rsp})$ where $\aux \samplen \setup(\theta)$.
One need to check that $\tilde{\com}$ and $\tilde{\rsp}$ are indistinguishable in the simulation and during the real execution.
If the commitment used is hiding, then $\com$ and $\tilde{\com}$ are indistinguishable in the simulation and during the real execution.
When $\alpha = 0$, one cannot distinguish between $\bm{\tilde{z}}_2$ and $\bm{z}_2$ as $\bm{u}$ is sampled uniformly at random.
When $\alpha = 1$, one cannot distinguish between $\bm{\tilde{z}}_4$ and $\bm{z}_4$ as $\pi[\bm{\tilde{x}}]$ follows the same probability distribution as $\pi[\bm{x}]$, since $\pi$ is a random permutation.
As a consequence, $\rsp$ and $\tilde{\rsp}$ are indistinguishable in the simulation and during the real execution.
Finally, $\Sim$ runs in polynomial time which completes the proof.
\end{proof}

%--------------------------------------------------------------------%
\section{PoK 2 \& 3 - Arbitrarily Small Soundness Error} \label{sec:pok23}
%--------------------------------------------------------------------%

In the previous section, we have leveraged the MPC-in-the-head technique in order to design a PoK for the $\SD$ problem over $\Ft$ achieving a soundness error of $1/2$.
Hereafter, we present two PoK for the $\GSD$ problem over $\Ft$ achieving an arbitrarily small soundness error equal to $1/N$ for some parameter $N$.

\subsection{Reducing soundness using several permutations} \label{sec:pok2}

We start by highlighting a particularity of PoK for the $\SD$ problem namely that it requires to prove two statements: (i) their exists a value $\bm{x}$ such that $\bm{H} \bm{x}^{\top} = \bm{y}^{\top}$ and (ii) the weight of $\bm{x}$ is small.
As a consequence, it is very natural to design these proofs with two parts checking respectively each one of the aforementioned properties (see \cite{AGS11, CVE11, BGS21} as well as our construction from Section~\ref{sec:pok1}) which leads to a soundness error equal to $1/2$.
In order to reduce the soundness error even further, one need to merge the two parts of the proof together which turn out to be challenging to do while preserving the zero-knowledge property of the underlying proof.
Indeed, one generally (see for instance Figure~\ref{pok1:fig1}) prove the first property by masking $\bm{x}$ using $\bm{u} + \bm{x}$ for some random value $\bm{u}$ and prove the second property by masking $\bm{x}$ using $\pi[\bm{x}]$ for some random permutation $\pi$. % while reconciling the two parts of the proof thanks to a third value such as $\pi[\bm{u} + \bm{x}]$.
The verifier can then convince itself by checking some third value such as $\pi[\bm{u} + \bm{x}]$ which 
binds the two parts of the proof together.
To enforce that the same permutation $\pi$ is used in both $\pi[\bm{x}]$ and $\pi[\bm{u} + \bm{x}]$, we can commit to the permutation $\pi$ and reveal it later such that at any given point (during or after the execution of the protocol) the verifier either knows $\pi$ or $\pi[\bm{x}]$ but not both.

Our second PoK (hereafter denoted PoK~2) solves this issue by \emph{introducing several permutations and revealing of all them but one} in order to prove the knowledge of the solution of a permuted syndrome decoding problem instance.
%In practice, this implies that we have to design the protocol in a way that permits to extract the considered permutations whenever one has more than one valid transcripts (soundness) while enforcing that the permutation cannot be retrieved from only one valid transcript (zero-knowledge).
In particular, we need to ensure the protocol guarantees soundness 
(i.e. one can extract the permutations used in the protocol, 
whenever more than one valid transcripts are given), and preserves
the zero-knowledge (i.e. all of the permutations cannot be retrieved
from any given (single) valid transcript).
Our PoK relies on the $\GSD$ problem namely the $\SD$ problem defined with a generator matrix instead of a parity-check matrix.
Given a $\GSD$ instance $(\bm{G}, \bm{y})$, we consider $N$ permuted instances $(\pi_i[\bm{G}], \pi_i[\bm{y}])_{i \in \intoneto{N}}$ satisfying $\pi_i[\bm{x} \bm{G}] + \pi_i[\bm{e}] = \pi_i[\bm{y}]$.
Here, the solution to the $\GSD$ problem $(\bm{x}, \bm{e})$ is the secret witness.
By adding random values $\bm{u}$ and $\bm{v}_i$, one get an equivalent equation namely $\pi_i[(\bm{u} + \bm{x})\bm{G}] + \bm{v}_i + \pi_i[\bm{e}] = \pi_i[\bm{y} + \bm{u}\bm{G}] + \bm{v}_i$.
Using the random mask $\bm{v}$ is mandatory as failing to do so (like in the initial Véron protocol) lead to a zero-knowledge issue that was identified in \cite{ZKIssue} and then fixed in \cite{ISIT21}.
To summarize, by adding random values $\bm{u}$ and $\bm{v}_i$, one get an equivalent equation namely $\pi_i[(\bm{u} + \bm{x})\bm{G}] + \bm{v}_i + \pi_i[\bm{e}] = \pi_i[\bm{y} + \bm{u}\bm{G}] + \bm{v}_i$.
This can be used for the verification while preserving the zero-knowledge.
We now explain how our PoK~2 achieve an arbitrarily small soundness error.
As shown in Figure~\ref{fig:pok2-h}, the helper can compute a commitment of 
$\pi_i[\bm{y} + \bm{u}\bm{G}] + \bm{v}_i$ as this value does not involve any 
secret information.
Thus, one can design a PoK for the $\GSD$ problem by revealing both $\pi_i[\bm{e}]$ and $\pi_i[(\bm{u} + \bm{x})\bm{G}] + \bm{v}_i$.
%We use the cut-and-choose technique on the $N$ permutations of the considered $\GSD$ instance in order to ensure that the latter value has been correctly computed.
%Given some challenge $\alpha$, the prover can reveal $\bm{u} + \bm{x}$, $(\pi_i, \bm{v}_i)_{i \in \intoneto{N} \setminus \alpha}$ as well as $\pi_{\alpha}[(\bm{u} + \bm{x})\bm{G}] + \bm{v}_\alpha$.
%Using the public value $\bm{G}$, the verifier can then recompute $\pi_i[(\bm{u} + \bm{x})\bm{G}] + \bm{v}_i$ for all $i \in \intoneto{N} \setminus \alpha$ which enforces that $\pi_{\alpha}[(\bm{u} + \bm{x})\bm{G}] + \bm{v}_\alpha$ has been correctly generated except with arbitrarily small probability $1/N$.
We now explain how to use the cut-and-choose technique 
on the $N$ permutations of the considered $\GSD$ instance 
in order to ensure that the latter value 
$\pi_i[(\bm{u} + \bm{x})\bm{G}] + \bm{v}_i$ has been correctly computed.
Given some challenge $\alpha \in \intoneto{N} $, the prover can reveal 
(i) the masked secret $\bm{u} + \bm{x}$, 
(ii) all the permutations and random masks $\bm{v}_i$ except for
the instance $\alpha$ (denoted by $(\pi_i, \bm{v}_i)_{i \in \intoneto{N} \setminus \alpha}$) 
as well as
(iii) the value to be checked $\pi_{\alpha}[(\bm{u} + \bm{x})\bm{G}] + \bm{v}_\alpha$, in the instance $\alpha$.
The verifier can then recompute $\pi_i[(\bm{u} + \bm{x})\bm{G}] + \bm{v}_i$ for
all $i \in \intoneto{N} \setminus \alpha$, using the public value $\bm{G}$.  
This enforces that $\pi_{\alpha}[(\bm{u} + \bm{x})\bm{G}] + \bm{v}_\alpha$ has 
been correctly generated except with arbitrarily small probability $1/N$.
As a technicality, our proof also requires $\bm{u}$ to be extractable for the soundness hence we define it as $\bm{u} = \sum\nolimits_{i \in \intoneto{N}} \bm{u}_i$.
The resulting protocol is described in Figure~\ref{fig:pok2-h}.
We explain in Section~\ref{sec:no-helper} how to remove the helper in order to construct a 3-round HVZK PoK.
Our PoK~2 achieves an arbitrarily small soundness error and therefore lead to small signatures however at the cost of a significant overhead on performances.

\begin{figure}[!ht]
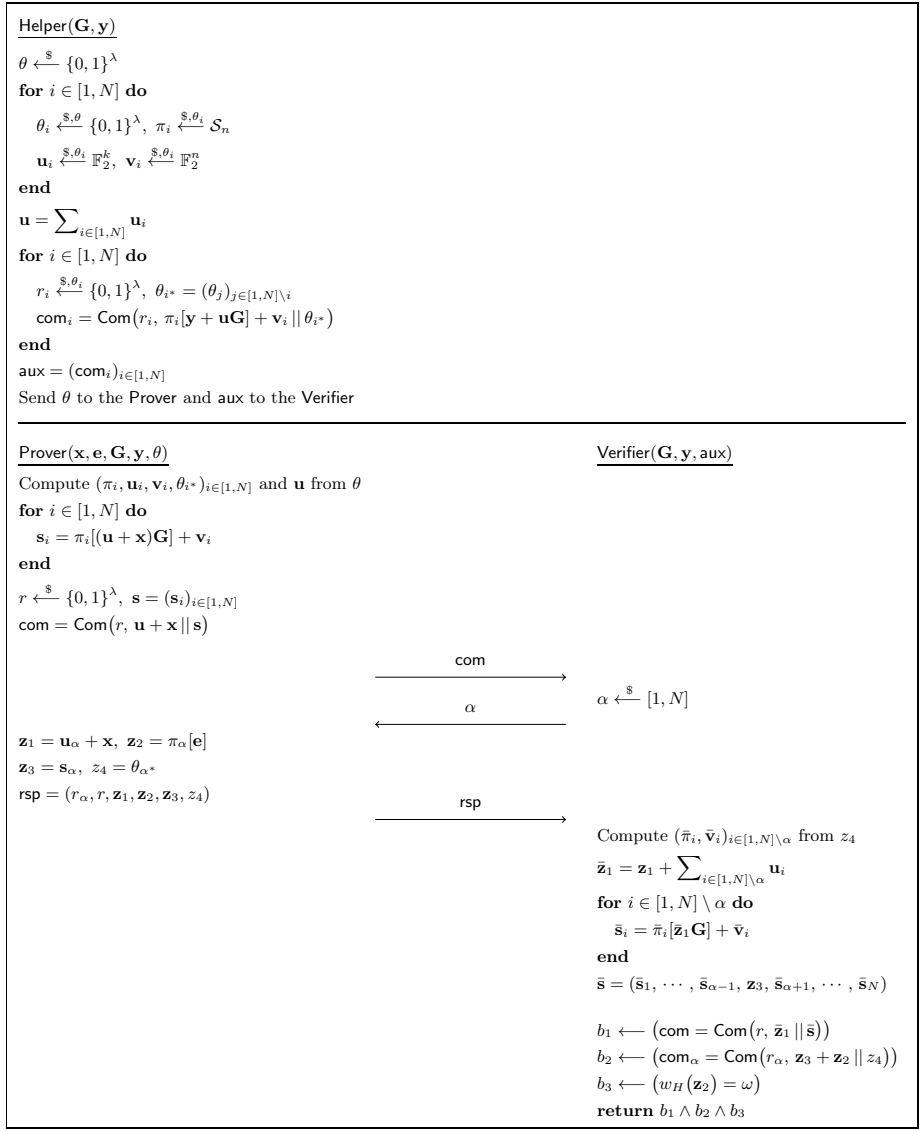
 
  \begin{center}
    \resizebox{1\textwidth}{!}{\fbox{
      \pseudocode{%
        \hspace{160pt} \> \> \hspace{160pt} \\[-0.75\baselineskip] % Hack to force font scaling
        \underline{\helperf(\bm{G}, \bm{y})} \\[0.1\baselineskip]
        \theta \sampler \bit^{\lambda} \\
        \pcfor i \in \intoneto{N} \pcdo \\
        \pcind \theta_{i} \samples{\theta} \bit^{\lambda}, ~ \pi_i \samples{\theta_i} \hperm{n} \\
        \pcind \bm{u}_i \samples{\theta_i} \Ftk, ~ \bm{v}_i \samples{\theta_i} \Ftn \\
        \pcend \\
        \bm{u} = \sum\nolimits_{i \in \intoneto{N}} \bm{u}_i \\ 
        \pcfor i \in \intoneto{N} \pcdo \\
        \pcind r_{i} \samples{\theta_i} \bit^{\lambda}, ~ \theta_{i^*} = (\theta_j)_{j \in \intoneto{N} \setminus i} \\ 
        \pcind \com_{i} = \commit{r_i, \, \pi_i[\bm{y} + \bm{u} \bm{G}] + \bm{v}_i \, || \, \theta_{i^*}} \\
        \pcend \\
        \aux = (\com_i)_{i \in \intoneto{N}} \\
        \tx{Send } \theta \tx{ to the } \proverf \tx{ and } \aux \tx{ to the } \verifierf \\[0.5\baselineskip][\hline]\\[-0.25\baselineskip]
        \underline{\proverf(\bm{x}, \bm{e}, \bm{G}, \bm{y}, \theta)} \> \> \underline{\verifierf(\bm{G}, \bm{y}, \aux)} \\
        \tx{Compute } (\pi_i, \bm{u}_i, \bm{v}_i, \theta_{i^*})_{i \in \intoneto{N}} \tx{ and } \bm{u} \tx{ from } \theta \\
        \pcfor i \in [1, N] \pcdo \\
        \pcind \bm{s}_i = \pi_i[(\bm{u} + \bm{x}) \bm{G}] + \bm{v}_i \\ 
        \pcend \\
        r \sampler \bit^{\lambda}, ~ \bm{s} = (\bm{s}_i)_{i \in \intoneto{N}} \\
        \com = \commit{r, \, \bm{u} + \bm{x} \, || \, \bm{s}}  \\
        \> \sendmessageright*{\com} \\[-0.5\baselineskip]
        \> \> \alpha \sampler \intoneto{N} \\[-\baselineskip]
        \> \sendmessageleft*{\alpha} \\[-0.5\baselineskip]
        \bm{z}_1 = \bm{u}_\alpha + \bm{x}, ~ \bm{z}_2 = \pi_{\alpha}[\bm{e}] \\
        \bm{z}_3 = \bm{s}_\alpha, ~ z_4 = \theta_{\alpha^*} \\
        \rsp = (r_{\alpha}, r, \bm{z}_1, \bm{z}_2, \bm{z}_3, z_4) \\[-\baselineskip]
        \> \sendmessageright*{\rsp} \\[-0.5\baselineskip]
        \> \> \tx{Compute } (\bar{\pi}_i, \bar{\bm{v}}_i)_{i \in \intoneto{N} \setminus \alpha} \tx{ from } z_4 \\
        \> \> \bar{\bm{z}}_1 = \bm{z}_1 + \sum\nolimits_{i \in \intoneto{N} \setminus \alpha} \bm{u}_i \\ 
        \> \> \pcfor i \in [1, N] \setminus \alpha \pcdo \\
        \> \> \pcind \bar{\bm{s}}_i = \bar{\pi}_i[\bar{\bm{z}}_1 \bm{G}] + \bar{\bm{v}}_i \\
        \> \> \pcend \\
        \> \> \bar{\bm{s}} = (\bar{\bm{s}}_1, \, \cdots, \, \bar{\bm{s}}_{\alpha - 1}, \, \bm{z}_3, \, \bar{\bm{s}}_{\alpha + 1}, \, \cdots, \, \bar{\bm{s}}_N) \\[0.75\baselineskip]
        \> \> b_1 \samplen \big( \com = \commit{r, \, \bar{\bm{z}}_1 \, || \, \bar{\bm{s}}} \big) \\ 
        \> \> b_2 \samplen \big( \com_\alpha = \commit{r_{\alpha}, \, \bm{z}_3 + \bm{z}_2 \, || \, z_4} \big) \\
        \> \> b_3 \samplen \big( \hw{\bm{z}_2} = \omega \big) \\
        \> \> \pcreturn b_1 \wedge b_2 \wedge b_3
      }
    }}
    \caption{ZK PoK with Helper for the $\GSD$ problem over $\Ft$ \label{fig:pok2-h}}
  \end{center}
\end{figure}

\newpage
~
\newpage
\begin{theorem}[\textbf{Proof of knowledge with helper}] \label{thm:pok2}
  If the commitment scheme is computationally binding and computationally hiding, then the protocol depicted in Figure~\ref{fig:pok2-h} is a proof of knowledge with helper for the $\GSD$ problem with challenge space $\mathcal{C}$ such that $|\mathcal{C}| = N$, with computational soundness error $1/N$ and
  honest-verifier computational zero-knowledge.
  \footnote{Proof of Knowledge systems with computational soundness
  are also called Arguments of Knowledge. Our PoK achieves
  computational ZK since the random masks $(\bm{u}, \bm{v})$ 
  added to hide the secrets are generated from seeds with the help
  of pseudorandom objects such as XOF.}

\end{theorem}

\begin{proof} We prove the correctness, special soundness and special honest-verifier zero-knowledge properties below.

\vspace{\baselineskip}
\noindent \textbf{Correctness.}
One prove the correctness by showing that the input to the
$\com$ and $\com_{\alpha}$ used by the verifier in the final steps of
Figure~\ref{fig:pok2-h} are same as the input provided by the prover
while generating $\com$ and $\com_{\alpha}$. Once the inputs are
shown to be identical, the correctness follows from the correctness
of the commitment scheme.
We begin by considering the inputs to $\com$, we need to show that
$\bar{\bm{z}}_1 = \bm{u} + \bm{x}$ and $\bar{\bm{s}} = \bm{s}$.
Note that the verifier computes $\bar{\bm{z}}_1$ as $\bm{z}_1 + \sum\nolimits_{i \in \intoneto{N} \setminus \alpha} \bm{u}_i$.
Here, $\bm{z}_1$ is computed by the prover as $(\bm{u}_{\alpha} + \bm{x})$,
and the second term is summation of all $\bm{u}_i$ except $\bm{u}_{\alpha}$.
Therefore, it is easy to see that
$\bar{\bm{z}}_1 = (\bm{u}_{\alpha} + \bm{x}) + \sum\nolimits_{i \in \intoneto{N} \setminus \alpha} \bm{u}_i = \bm{u} + \bm{x}$.
Next, note that the verifier has access to all seeds $\theta_i$ except
$\theta_\alpha$ from the $z_4$ sent by the prover in the response $\rsp$.
The verifier can therefore compute all the permutation $\pi_i$ and
the random masks $\bm{v}_i$ except $\pi_\alpha$ and $\bm{v}_\alpha$.
As shown in previous step, the verifier also knows the value
$\bar{\bm{z}}_1 = \bm{u} + \bm{x}$. The verifier
can therefore compute 
$\bar{\bm{s}}_i = \pi_i[\bar{\bm{z}}_1\bm{G}] + \bm{v}_i$ for all $i \in \intoneto{N} \setminus \alpha$ since $\bm{G}$ is public. Which is same as
$\bar{\bm{s}}_i = \pi_i[(\bm{u} + \bm{x})\bm{G}] + \bm{v}_i$ for all $i \in \intoneto{N} \setminus \alpha$, but this is exactly how
the $\bm{s}_i$ values are computed by the prover.
Therefore, we have shown that $\bar{\bm{s}}_i = \bm{s}_i$ for all $i \neq \alpha$. However, $\bar{\bm{s}}_\alpha = \bm{z}_3 = \bm{s}_\alpha$
since $\bm{z}_3$ is computed by the prover.
Therefore we have shown that $\bar{\bm{s}} = \bm{s}$.
This concludes the part related the commitment $\com$, since we have shown that
both the inputs are identical.
We now show the same for $\com_{\alpha}$. Here, we need to show that
(i) $\bm{z}_3 + \bm{z}_2 = \pi_{\alpha}[\bm{y} + \bm{u} \bm{G}] + \bm{v}_{\alpha}$ and, 
(ii) $z_4 = \theta_{\alpha^*}$ where $\theta_{\alpha^*}$ denotes all the seeds
$\theta_i$ except for $i = \alpha$.
It is easy to verify and has been discussed earlier
that $z_4$ computed by the prover as $\theta_{\alpha^*}$ and
sent as part of the response to the verifier. 
Thus, $z_4 = \theta_{\alpha^*}$ by just inspecting the prover's
response.
Recall, that $\bm{z}_3 = \bm{s}_\alpha$ and 
$\bm{s}_{\alpha} = \pi_{\alpha}[(\bm{u} + \bm{x}) \bm{G}] + \bm{v}_{\alpha}$.
Adding $\bm{z}_2$ to both sides and substituting $\bm{z}_2 = \pi_{\alpha}[\bm{e}]$, one get
$\bm{z}_{3} + \bm{z}_2 = \pi_{\alpha}[(\bm{u} + \bm{x}) \bm{G}] + \bm{v}_{\alpha} + \pi_{\alpha}[\bm{e}] = \pi_{\alpha}[\bm{x} \bm{G} + \bm{e} + \bm{u} \bm{G}] + \bm{v}_{\alpha} = \pi_{\alpha}[\bm{y} + \bm{u} \bm{G}] + \bm{v}_{\alpha}.$
Hence we have shown that both the inputs to $\com_\alpha$ are also
identical.
As mentioned earlier, since the inputs to $\com$ and $\com_\alpha$
are identical to their counterparts computed by the prover, the correctness
of the protocol follows from the correctness of the commitment scheme.

\vspace{\baselineskip}
  \noindent \textbf{Special soundness.} 
  In order to prove the special soundness,
  one need to build an efficient knowledge extractor $\Ext$ which
  returns a solution of the $\GSD$ instance defined by $(\bm{G}, \bm{y})$
  with high probability,
  when provided with two valid transcripts $(\bm{G}, \bm{y}, \aux, \com, \alpha, \rsp)$ and $(\bm{G}, \bm{y}, \aux, \com, \allowbreak \alpha', \rsp')$ with $\alpha \neq \alpha'$ 
  generated by a PPT adversary (malicious prover) $\adv$,
  where $\aux = \setup(\theta)$ for some random seed $\theta$.
  The knowledge extractor $\Ext$ computes the solution as:

\vspace{0.5\baselineskip} 
\pseudocode{%
  \tx{1. Compute } (\pi_i, \bm{u}_i)_{i \in \intoneto{N}} \tx{ from } z_4 \tx{ and } z_4' \\
  \tx{2. Output } (\bm{z}_1 - \bm{u}_{\alpha}, \pi_{\alpha}^{-1}[\bm{z}_2])
}
\vspace{0.5\baselineskip} 

\noindent We now show that this can be computed efficiently by $\Ext$, and
then prove that the output is indeed a solution to the given $\GSD$ problem.
Recall from Figure~\ref{fig:pok2-h}, that the prover's response is
of the form $\rsp = (r_{\alpha}, r, \bm{z}_1, \bm{z}_2, \bm{z}_3, z_4)$
and $\rsp' = (r_{\alpha'}, r', \bm{z}_1', \bm{z}_2', \bm{z}_3', z_4')$.
Here, $z_4$ denotes all the seeds $\theta_i$ for $i \neq \alpha$ and
$z_4'$ denotes all the seeds $\theta_i$ for $i \neq \alpha'$.
Since, $\alpha \neq \alpha'$, the $\Ext$ has access to all the
seeds $\theta_i$ for $i \in \intoneto{N}$.
The extractor therefore can efficiently compute all the permutations $\pi_i$
and masks $\bm{u}_i$ for $i \in \intoneto{N}$, including $\pi_\alpha$ and
$\bm{u}_{\alpha}$. Using these the extractor can efficiently compute and
then output $(\bm{z}_1 - \bm{u}_{\alpha}, \pi_{\alpha}^{-1}[\bm{z}_2])$.

Let $\tilde{\bm{\theta}} = ({\tilde{\theta}}_{i})_{i \in \intoneto{N}}$ denote the seeds used to generate the commitments
$(\com_i)_{i \in \intoneto{N}}$ comprising the value $\aux$.
Note that if there exists an index $j \in \intoneto{N}$, such that
${\theta}_{j} \neq  {\tilde{\theta}}_{j}$ or 
${\theta'}_{j} \neq  {\tilde{\theta}}_{j}$, then that particular
${\theta}_{j}$ or ${\theta'}_{j}$ can be used to break
the binding property of the commitment scheme as any commitment $\com_i$ in $\aux$ where $i \neq j$, can be
opened as valid commitment for two distinct messages where one contains
${\tilde{\theta}}_{j}$ and the other contains ${\theta}_{j}$.
Thus, the binding property of the commitment scheme ensures that
the extracted mask $\bm{u}_{\alpha}$ and the permutation
$\pi_{\alpha}$ are same as the those committed in the common 
information $\aux$ and the first commitment $\com$.
We now explain why the extractor's output is a solution to the considered $\GSD$ problem instance.
Note that $\bm{z}_1$ is computed as $\bm{z}_1 = \bm{u}_{\alpha} + \bm{x}$,
therefore the extractor correctly recovers the secret $\bm{x}$
by computing $\bm{z}_1 - \bm{u}_{\alpha}$.
Also, $\bm{z}_2 = \pi_{\alpha}[\bm{e}]$, which can be inverted by
$\Ext$ after learning $\pi_{\alpha}$ as $\bm{e} = \pi_{\alpha}^{-1}[\bm{z}_2]$.
This proves that extracted solution
$(\bm{z}_1 - \bm{u}_{\alpha}, \pi_{\alpha}^{-1}[\bm{z}_2]) = (\bm{x}, \bm{e})$
is the correct solution to the $\GSD$ problem.

\vspace{\baselineskip}
\noindent \textbf{Special Honest-Verifier Zero-Knowledge.} 
We start by explaining why valid transcripts do not leak anything on the secret $(\bm{x}, \bm{e})$.
A valid transcript contains $(\bm{u}_{\alpha} + \bm{x}, \, \pi_{\alpha}[\bm{e}], \, \bm{s}_{\alpha}, \, (\pi_i, \bm{u}_i, \bm{v}_i)_{i \in \intoneto{N} \setminus \alpha})$ namely the secret $\bm{x}$ is masked by a random value $\bm{u}_{\alpha}$ and the secret $\bm{e}$ is masked by a random permutation $\pi_{\alpha}$.
From the protocol, one can compute $\bm{u}_{\alpha} + \bm{x}$ and $\bm{u} + \bm{x}$ however this does not leak anything on $\bm{x}$ as $\bm{u}_{\alpha}$ and $\bm{u}$ are both unknown.
The main difficulty concerns the permutation $\pi_\alpha$ as the protocol requires $\pi_{\alpha}[(\bm{u} + \bm{x}) \bm{G}]$ to be computed while both $(\bm{u} + \bm{x})$ and $\bm{G}$ are known.
To overcome this issue, the protocol actually computes $\pi_{\alpha}[(\bm{u} + \bm{x}) \bm{G}] + \bm{v}_{\alpha}$ for some random value $\bm{v}_{\alpha}$ hence the transcript does not leak anything on $\bm{e}$.
Formally, one can build a $\ppt$ simulator $\Sim$ that given the public values $(\bm{G}, \bm{y})$, a random seed $\theta$ and a random challenge $\alpha$ outputs a transcript $(\bm{G}, \bm{y}, \aux, \com, \alpha, \rsp)$ such that $\aux = \setup(\theta)$ that is computationally indistinguishable from the transcript of honest executions of the protocol:

\pseudocode{%
  \tx{1. Compute } (\pi_i, \bm{u}_i, \bm{v}_i, \theta_{i^*})_{i \in \intoneto{N}} \tx{ and } \bm{u} \tx{ from } \theta \\
  \tx{2. Compute } \bm{\tilde{x}} \sampler \Ftk, ~ \bm{\tilde{e}} \sampler \swset{\omega}{\Ftn} \tx{ and } \bm{\tilde{y}} = \bm{\tilde{x}} \bm{G} + \bm{\tilde{e}} \\
  \tx{3. Compute } \bm{\tilde{s}}_i = \pi_i[(\bm{u} + \bm{\tilde{x}}) \bm{G}] + \bm{v}_i \tx{ for all } i \in \intoneto{N} \setminus \alpha \\
  \tx{4. Compute } \bm{\tilde{s}}_{\alpha} = \pi_{\alpha}[(\bm{u} + \bm{\tilde{x}}) \bm{G}] + \bm{v}_{\alpha} + \pi_{\alpha}[\bm{y} - \bm{\tilde{y}}] \\
  \tx{5. Compute } r \sampler \bit^{\lambda}, ~ \bm{\tilde{s}} = (\bm{\tilde{s}}_i)_{i \in \intoneto{N}} \tx{ and } \tilde{\com} = \commit{r, \, (\bm{u} + \bm{\tilde{x}}) \, || \, \bm{\tilde{s}}} \\
  \tx{6. Compute } \bm{\tilde{z}}_1 = \bm{u}_{\alpha} + \bm{\tilde{x}}, ~ \bm{\tilde{z}}_2 = \pi_{\alpha}[\bm{\tilde{e}}], ~ \bm{\tilde{z}}_3 = \bm{\tilde{s}}_{\alpha}, ~ z_4 = \theta_{\alpha^*} \\ 
  \tx{7. Compute  } \tilde{\rsp} = (r_{\alpha}, r, \bm{\tilde{z}}_1, \bm{\tilde{z}}_2, \bm{\tilde{z}}_3, z_4) \tx{ and output } (\bm{G}, \bm{y}, \aux, \tilde{\com}, \alpha, \tilde{\rsp})\\[-0.5\baselineskip]
}

\noindent The transcript generated by the simulator $\Sim$ is $(\bm{G}, \bm{y}, \aux, \tilde{\com}, \alpha, \tilde{\rsp})$ where $\aux \samplen \setup(\theta)$.
We now show that $\tilde{\com}$ and $\tilde{\rsp}$ are indistinguishable in the simulation and during the real execution.
If the commitment used is hiding, then $\com$ and $\tilde{\com}$ are indistinguishable in the simulation and during the real execution.
Since $\bm{\tilde{x}}$ (in the simulation) and $\bm{x}$ (in the real execution)
are masked by a random mask $\bm{u}_{\alpha}$ which is unknown
to the verifier, $\bm{\tilde{z}}_1$ and $\bm{z}_1$ 
are computationally indistinguishable.
Similarly, $\bm{\tilde{e}}$ and $\bm{e}$ have same hamming weight $\omega$,
and are masked by random permutation $\pi_{\alpha}$ which is never
known to the verifier. Thus, making $\bm{\tilde{z}}_2$ and $\bm{z}_2$ 
computationally indistinguishable.
In addition, as the mask $\bm{v}_{\alpha}$ is sampled uniformly at random
and is unknown to the verifier, it cannot distinguish between $\bm{\tilde{z}}_3$ and $\bm{z}_3$.
Finally, $z_4 = \theta_{\alpha^*}$ is identical in both cases.
As a consequence, $\rsp$ and $\tilde{\rsp}$ are 
computationally indistinguishable 
in the simulation and during the real execution.
Finally, $\Sim$ runs in polynomial time which completes the proof.
\end{proof}

\subsection{Reducing soundness using a shared permutation} \label{sec:pok3}

The PoK~2 presented in the previous section achieves an arbitrarily small soundness error equal to $1/N$.
As such, it is theoretically equivalent to the proposal from \cite{FJR21}.
Nonetheless, the optimized version of the FJR protocol outperforms the optimized version of our PoK~2 in practice.
This is explained by the fact that some optimizations bring a better improvement for the FJR protocol than for PoK~2.
We defer the interested reader to the paragraph ``Commitment compression'' in Section~\ref{sec:optimizations} for additional details on this topic.
In this section, we show how one can adapt our PoK~2 to the shared permutation setting used in \cite{FJR21} to achieve similar performances.
The resulting protocol is denoted PoK~3 and can be seen as a dual version of the FJR protocol based on the $\GSD$ problem rather than the $\SD$ one.

\begin{figure}[!ht]
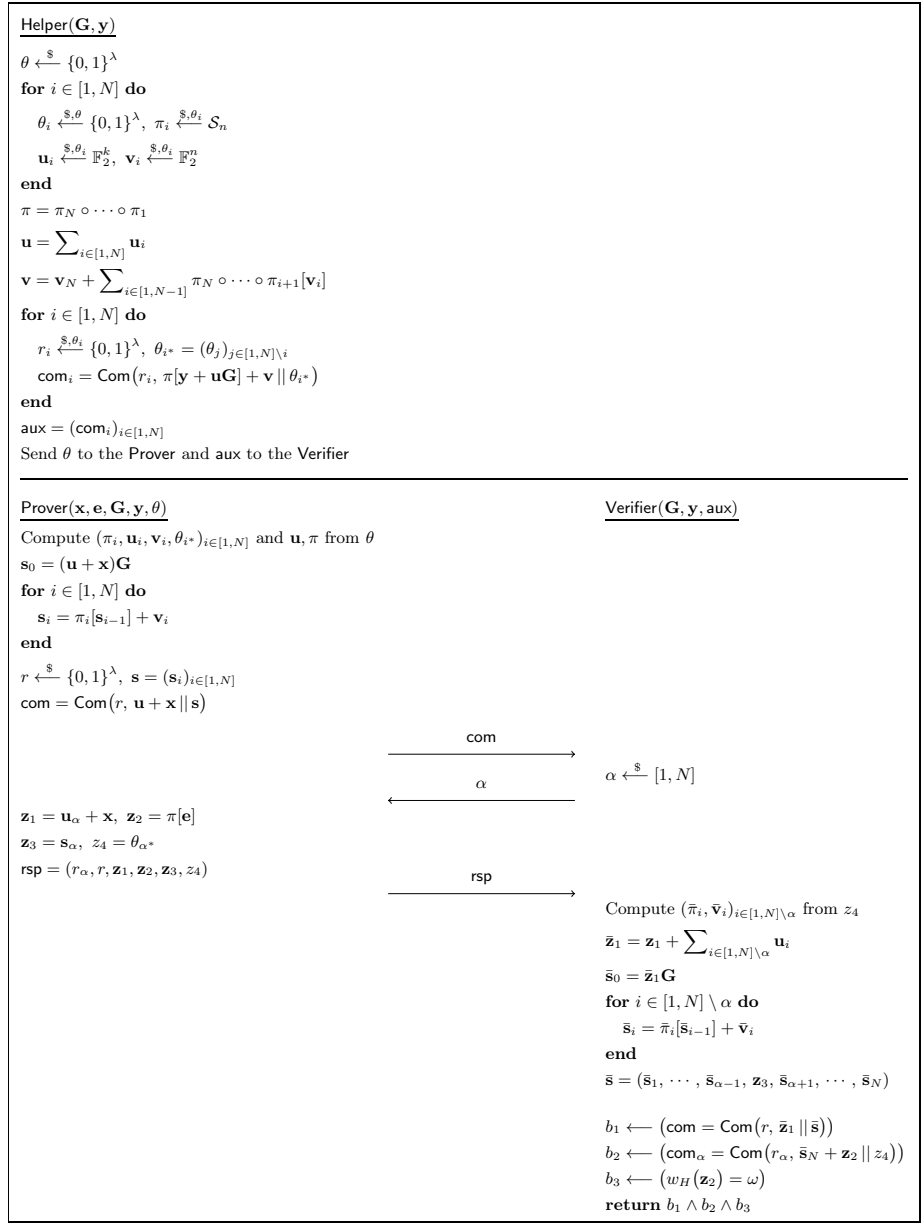
 
  \begin{center}
    \resizebox{1\textwidth}{!}{\fbox{
      \pseudocode{%
        \hspace{160pt} \> \> \hspace{160pt} \\[-0.75\baselineskip] % Hack to force font scaling
        \underline{\helperf(\bm{G}, \bm{y})} \\[0.1\baselineskip]
        \theta \sampler \bit^{\lambda} \\
        \pcfor i \in \intoneto{N} \pcdo \\
        \pcind \theta_{i} \samples{\theta} \bit^{\lambda}, ~ \pi_i \samples{\theta_i} \hperm{n} \\
        \pcind \bm{u}_i \samples{\theta_i} \Ftk, ~ \bm{v}_i \samples{\theta_i} \Ftn \\
        \pcend \\
        \pi = \pi_N \circ \cdots \circ \pi_1 \\
        \bm{u} = \sum\nolimits_{i \in \intoneto{N}} \bm{u}_i \\ 
        \bm{v} = \bm{v}_N + \sum\nolimits_{i \in \intoneto{N - 1}} \pi_N \circ \cdots \circ \pi_{i + 1}[\bm{v}_i] \\
        \pcfor i \in \intoneto{N} \pcdo \\
        \pcind r_{i} \samples{\theta_i} \bit^{\lambda}, ~ \theta_{i^*} = (\theta_j)_{j \in \intoneto{N} \setminus i} \\ 
        \pcind \com_{i} = \commit{r_i, \, \pi[\bm{y} + \bm{u} \bm{G}] + \bm{v} \, || \, \theta_{i^*}} \\
        \pcend \\
        \aux = (\com_i)_{i \in \intoneto{N}} \\
        \tx{Send } \theta \tx{ to the } \proverf \tx{ and } \aux \tx{ to the } \verifierf \\[0.5\baselineskip][\hline]\\[-0.25\baselineskip]
        \underline{\proverf(\bm{x}, \bm{e}, \bm{G}, \bm{y}, \theta)} \> \> \underline{\verifierf(\bm{G}, \bm{y}, \aux)} \\
        \tx{Compute } (\pi_i, \bm{u}_i, \bm{v}_i, \theta_{i^*})_{i \in \intoneto{N}} \tx{ and } \bm{u}, \pi \tx{ from } \theta \\
        \bm{s}_0 = (\bm{u} + \bm{x}) \bm{G} \\
        \pcfor i \in [1, N] \pcdo \\
        \pcind \bm{s}_i = \pi_i[\bm{s}_{i - 1}] + \bm{v}_i \\ 
        \pcend \\
        r \sampler \bit^{\lambda}, ~ \bm{s} = (\bm{s}_i)_{i \in \intoneto{N}} \\
        \com = \commit{r, \, \bm{u} + \bm{x} \, || \, \bm{s}}  \\
        \> \sendmessageright*{\com} \\[-0.5\baselineskip]
        \> \> \alpha \sampler \intoneto{N} \\[-\baselineskip]
        \> \sendmessageleft*{\alpha} \\[-0.5\baselineskip]
        \bm{z}_1 = \bm{u}_\alpha + \bm{x}, ~ \bm{z}_2 = \pi[\bm{e}] \\
        \bm{z}_3 = \bm{s}_\alpha, ~ z_4 = \theta_{\alpha^*} \\
        \rsp = (r_{\alpha}, r, \bm{z}_1, \bm{z}_2, \bm{z}_3, z_4) \\[-\baselineskip]
        \> \sendmessageright*{\rsp} \\[-0.5\baselineskip]
        \> \> \tx{Compute } (\bar{\pi}_i, \bar{\bm{v}}_i)_{i \in \intoneto{N} \setminus \alpha} \tx{ from } z_4 \\
        \> \> \bar{\bm{z}}_1 = \bm{z}_1 + \sum\nolimits_{i \in \intoneto{N} \setminus \alpha} \bm{u}_i \\ 
        \> \> \bar{\bm{s}}_0 = \bar{\bm{z}}_1 \bm{G} \\
        \> \> \pcfor i \in [1, N] \setminus \alpha \pcdo \\
        \> \> \pcind \bar{\bm{s}}_i = \bar{\pi}_i[\bar{\bm{s}}_{i - 1}] + \bar{\bm{v}}_i \\
        \> \> \pcend \\
        \> \> \bar{\bm{s}} = (\bar{\bm{s}}_1, \, \cdots, \, \bar{\bm{s}}_{\alpha - 1}, \, \bm{z}_3, \, \bar{\bm{s}}_{\alpha + 1}, \, \cdots, \, \bar{\bm{s}}_N) \\[0.75\baselineskip]
        \> \> b_1 \samplen \big( \com = \commit{r, \, \bar{\bm{z}}_1 \, || \, \bar{\bm{s}}} \big) \\ 
        \> \> b_2 \samplen \big( \com_\alpha = \commit{r_{\alpha}, \, \bar{\bm{s}}_N + \bm{z}_2 \, || \, z_4} \big) \\
        \> \> b_3 \samplen \big( \hw{\bm{z}_2} = \omega \big) \\
        \> \> \pcreturn b_1 \wedge b_2 \wedge b_3
      }
    }}
    \caption{ZK PoK with Helper for the $\SD$ problem over $\Ft$ \label{fig:pok3-h}}
  \end{center}
\end{figure}

\begin{theorem}[\textbf{Proof of knowledge with helper}] \label{thm:pok3}
  If the commitment is binding and 
  hiding, then the protocol depicted 
  in Figure~\ref{fig:pok3-h} is a proof of knowledge with helper 
  for the $\GSD$ problem with challenge space $\mathcal{C}$ such that $|\mathcal{C}| = N$, with computational soundness error $1/N$ and
  honest-verifier computational zero-knowledge.
\end{theorem}

\begin{proof} One need to prove the correctness, special soundness and special honest-verifier zero-knowledge properties to complete the proof.

\vspace{\baselineskip}
\noindent \textbf{Correctness.} The proof for correctness follows
the same arguments as for the proof of correctness of Theorem~\ref{thm:pok2},
with the only difference that the $\bm{s}_i$ values are computed
by composing the permutations $\pi_i$ in a nested manner.

%The correctness follows straightforwardly from the protocol description once the commitments $\com$ and $\com_{\alpha}$ are verified. 
%In order to validate these commitments, one should respectively verify that $\bar{\bm{z}}_1 \, || \, \bar{\bm{s}} = \bm{u} + \bm{x} \, || \, \bm{s}$ and $\bar{\bm{s}}_N + \bm{z}_2 \, || \, z_4 = \pi[\bm{y} + \bm{u} \bm{G}] + \bm{v} \, || \, \theta_{\alpha^*} $.
%One can see that $\bar{\bm{z}}_1 = (\bm{u}_{\alpha} + \bm{x}) + \sum\nolimits_{i \in \intoneto{N} \setminus \alpha} \bm{u}_i = \bm{u} + \bm{x}$.
%As $\bm{z}_3 = \bm{s}_{\alpha} = \pi_{\alpha}[\bm{s}_{\alpha - 1}] + \bm{v}_{\alpha}$, $\bar{\bm{s}}_0 = (\bm{u} + \bm{x}) \bm{G}$ and $\bar{\bm{s}}_i = \pi_i[\bar{\bm{s}}_{i-1}] + \bm{v}_i$ for all $i \in \intoneto{N} \setminus \alpha$, it follows that $\bar{\bm{s}} = \bm{s}$.
%In addition, $z_4 = \theta_{\alpha^*}$ and $\bar{\bm{s}}_{N} + \bm{z}_2 = \pi[(\bm{u} + \bm{x}) \bm{G}] + \bm{v} + \pi[\bm{e}] = \pi[\bm{x} \bm{G} + \bm{e} + \bm{u} \bm{G}] + \bm{v} = \pi[\bm{y} + \bm{u} \bm{G}] + \bm{v}$. 

\vspace{\baselineskip}
\noindent \textbf{Special soundness.} Given an adversary $\adv$ that outputs 
two valid transcripts $(\bm{G}, \bm{y}, \aux, \com, \alpha, \rsp)$ and $(\bm{G}, \bm{y}, \aux, \com, \allowbreak \alpha', \rsp')$ with $\alpha \neq \alpha'$ and where $\aux = \setup(\theta)$ for some random seed $\theta$, one can build a knowledge extractor $\Ext$ that returns a solution of the $\GSD$ instance defined by $(\bm{G}, \bm{y})$ with high probability as follows:

\vspace{0.5\baselineskip} 
\pseudocode{%
  \tx{1. Compute } (\pi_i, \bm{u}_i)_{i \in \intoneto{N}} \tx{ from } z_4 \tx{ and } z_4' \\
  \tx{2. Compute } \pi = \pi_N \circ \cdots \circ \pi_1 \\
  \tx{3. Output } (\bm{z}_1 - \bm{u}_{\alpha}, \pi^{-1}[\bm{z}_2])
}
\vspace{0.5\baselineskip} 

\noindent The proof of soundness showing that the extractor $\Ext$ is efficient
and returns a valid solution for the $\GSD$ instance, follows
the same ideas as for the proof of soundness of Theorem~\ref{thm:pok2},
with the only difference of computing the permutation
$\pi$ as composition of the permutations $\pi_i$ after
extracting $\pi_\alpha$.

\vspace{\baselineskip}
\noindent \textbf{Special Honest-Verifier Zero-Knowledge.} 
The proof of zero-knowledge follows
the same arguments as for the proof of zero-knowledge of Theorem~\ref{thm:pok2}.
We provide the description of the PPT simulator $\Sim$
which generates the indistinguishable transcript using
only the public information for the completeness below.
%We start by explaining why valid transcripts don't leak anything on the secret $(\bm{x}, \bm{e})$.
%A valid transcript contains $(\bm{u}_{\alpha} + \bm{x}, \, \pi[\bm{e}], \, \bm{s}_{\alpha}, \, (\pi_i, \bm{u}_i, \bm{v}_i)_{i \in \intoneto{N} \setminus \alpha})$ namely the secret $\bm{x}$ is masked by a random value $\bm{u}_{\alpha}$ and the secret $\bm{e}$ is masked by a random permutation $\pi$.
%From the protocol, one can compute $\bm{u}_{\alpha} + \bm{x}$ and $\bm{u} + \bm{x}$ however this does not leak anything on $\bm{x}$ as $\bm{u}_{\alpha}$ and $\bm{u}$ are both unknown.
%The main difficulty concerns the permutation $\pi$ as the protocol requires $\pi[(\bm{u} + \bm{x}) \bm{G}]$ to be computed while both $(\bm{u} + \bm{x})$ and $\bm{G}$ are known.
%To overcome this issue, the protocol actually computes $\pi[(\bm{u} + \bm{x}) \bm{G}] + \bm{v}$ for some random value $\bm{v}$ using the shared permutation setting from \cite{FJR21}.
%As the $\pi_{\alpha}$ and $\bm{v}_{\alpha}$ values used to compute $\bm{s}_{\alpha}$ are not disclosed, the verifier can compute $\pi[(\bm{u} + \bm{x}) \bm{G}] + \bm{v}$ without revealing anything on $\pi$ hence the transcript does not leak anything on $\bm{e}$.
Formally, one can build a $\ppt$ simulator $\Sim$ that given the public values $(\bm{G}, \bm{y})$, a random seed $\theta$ and a random challenge $\alpha$ outputs a transcript $(\bm{G}, \bm{y}, \aux, \com, \alpha, \rsp)$ such that $\aux = \setup(\theta)$ that is computationally indistinguishable from the probability distribution of transcripts of honest executions of the protocol:

\vspace{0.5\baselineskip} 
\pseudocode{%
  \tx{1. Compute } (\pi_i, \bm{u}_i, \bm{v}_i, \theta_{i^*})_{i \in \intoneto{N}} \tx{ and } \bm{u}, \pi \tx{ from } \theta \\
  \tx{2. Compute } \bm{\tilde{x}} \sampler \Ftk, ~ \bm{\tilde{e}} \sampler \swset{\omega}{\Ftn} \tx{ and } \bm{\tilde{y}} = \bm{\tilde{x}} \bm{G} + \bm{\tilde{e}} \\
  \tx{3. Compute } \bm{\tilde{s}}_0 = (\bm{u} + \bm{\tilde{x}})\bm{G} \text{ and } \bm{\tilde{s}}_i = \pi_i[\bm{\tilde{s}}_{i - 1}] + \bm{v}_i \tx{ for all } i \in [1, \alpha - 1] \\
  \tx{4. Compute } \bm{\tilde{s}}_{\alpha} = \pi_{\alpha}[(\bm{u} + \bm{\tilde{x}}) \bm{G}] + \bm{v}_{\alpha} + \pi^{-1}_{\alpha + 1} \circ \cdots \circ \pi^{-1}_{N} \circ \pi[\bm{y} - \bm{\tilde{y}}] \\
  \tx{5. Compute } \bm{\tilde{s}}_i = \pi_i[\bm{\tilde{s}}_{i - 1}] + \bm{v}_i \tx{ for all } i \in [\alpha + 1, N] \\
  \tx{6. Compute } r \sampler \bit^{\lambda}, ~ \bm{\tilde{s}} = (\bm{\tilde{s}}_i)_{i \in \intoneto{N}} \tx{ and } \tilde{\com} = \commit{r, \, (\bm{u} + \bm{\tilde{x}}) \, || \, \bm{\tilde{s}}} \\
  \tx{7. Compute } \bm{\tilde{z}}_1 = \bm{u}_{\alpha} + \bm{\tilde{x}}, ~ \bm{\tilde{z}}_2 = \pi[\bm{\tilde{e}}], ~ \bm{\tilde{z}}_3 = \bm{\tilde{s}}_{\alpha}, ~ z_4 = \theta_{\alpha^*} \\ 
  \tx{8. Compute  } \tilde{\rsp} = (r_{\alpha}, r, \bm{\tilde{z}}_1, \bm{\tilde{z}}_2, \bm{\tilde{z}}_3, z_4) \tx{ and output } (\bm{G}, \bm{y}, \aux, \tilde{\com}, \alpha, \tilde{\rsp})\\[-0.5\baselineskip]
}
\vspace{0.5\baselineskip} 

%\noindent The transcript generated by the simulator $\Sim$ is $(\bm{G}, \bm{y}, \aux, \tilde{\com}, \alpha, \tilde{\rsp})$ where $\aux \samplen \setup(\theta)$.
%One need to check that $\tilde{\com}$ and $\tilde{\rsp}$ are indistinguishable in the simulation and during the real execution.
%If the commitment used is hiding, then $\com$ and $\tilde{\com}$ are indistinguishable in the simulation and during the real execution.
%As $\bm{\tilde{x}}$ follows the same probability distributions as $\bm{x}$, $\bm{\tilde{z}}_1$ and $\bm{z}_1$ are indistinguishable.
%Using the same argument on $\bm{\tilde{e}}$ and $\bm{e}$, one can also see that $\bm{\tilde{z}}_2$ and $\bm{z}_2$ are indistinguishable.
%In addition, as $\bm{v}_{\alpha}$ is sampled uniformly at random, one cannot distinguish between $\bm{\tilde{z}}_3$ and $\bm{z}_3$.
%As a consequence, $\rsp$ and $\tilde{\rsp}$ are indistinguishable in the simulation and during the real execution.
%Finally, $\Sim$ runs in polynomial time which completes the proof.
\end{proof}

%--------------------------------------------------------------------%
\section{PoK without Trusted Helper} \label{sec:no-helper}
%--------------------------------------------------------------------%

PoK with helper can be transformed into either 3-round HVZK PoK without helper or 5-round HVZK PoK without helper using the cut-and-choose paradigm as explained in \cite{KKW18, Beullens20}.
When the 5-round transformation is used, one must take into account the attack from \cite{KZ20} that specifically exploits the fact that the proof has a 5-round structure.
Choosing to use the 3-round or the 5-round transformation leads to different communication costs depending on the underlying proof and therefore must be decided on a case by case basis.
Hereafter, we present both transformations and discuss which ones to consider for the proof of knowledge introduced in Sections \ref{sec:pok1} and \ref{sec:pok23} respectively.

The main idea is to let the prover run the setup phase multiple times
with many independent seeds $\setup{(\theta^{(k)})}$ for $k \in \intoneto{M}$
and then share the auxiliary information $\aux^{(k)}$ for all the instances
with the verifier. The verifier then picks an arbitrary instance $\kappa$,
and the prover sends all the seeds $\theta^{(k)}$ for $k \neq \kappa$
to the verifier. The verifier can then verify that the received
auxiliary information $\aux^{(k)}$ has been honestly computed
by running the setup algorithm itself, if this check does not pass then
the verifier rejects. Otherwise, the prover and the verifier proceed
with the protocol for PoK using the seed $\theta^{(\kappa)}$.
If the soundness error of the PoK with the helper is $\frac{1}{N}$,
and the protocol without helper runs the setup $M$ times with independent
seeds, then the soundness error of the PoK without helper is
$\max(\frac{1}{M}, \frac{1}{N})$.

\begin{theorem}[\textbf{3-round Proof of Knowledge} \cite{Beullens20}]
  If the commitment used is binding and hiding, then the protocol depicted in Figure~\ref{no-helper-fig-3r} is a 3-round honest-verifier zero-knowledge proof of knowledge with challenge space $\mathcal{C}_1 \times \mathcal{C}_2$ such that $|\mathcal{C}_1| = M$ and $|\mathcal{C}_2| = N$ and soundness error equal to $\max(\frac{1}{M}, \frac{1}{N})$.
\end{theorem}

\begin{proof}
  This is a direct application of Theorem 3 from \cite{Beullens20} that permits to build a 3-round PoK without helper from a PoK with helper.
\end{proof}

\begin{theorem}[\textbf{5-round Proof of Knowledge}]
  If the commitment used is binding and hiding, then the protocol depicted in Figure~\ref{no-helper-fig-5r} is a 5-round honest-verifier zero-knowledge proof of knowledge with challenge space $\mathcal{C}_1 \times \mathcal{C}_2$ such that $|\mathcal{C}_1| = M$ and $|\mathcal{C}_2| = N$ and soundness error equal to $\max(\frac{1}{M}, \frac{1}{N})$.
\end{theorem}

\begin{proof}
  One can straightforwardly adapt the proof of Theorem 3 from \cite{Beullens20} to the 5-round setting.
\end{proof}

\noindent \textbf{Removing the helper from PoK~1.} For our first PoK (see Section \ref{sec:pok1}), we will consider both the 3-round and 5-round transformations.
We defer the reader to Appendices \ref{app:pok1-3r-nopt} and \ref{app:pok1-5r-nopt} for the description of the 3-round PoK~1 and 5-round PoK~1 protocols which are obtained after respectively applying the 3-round and 5-round transformations to our fist proof of knowledge.
The 3-round PoK~1 protocol is a slightly more conservative choice while the 5-round PoK~1 lead to a slightly smaller signature (see Section \ref{sec:parameters}).

\begin{figure}[!ht]
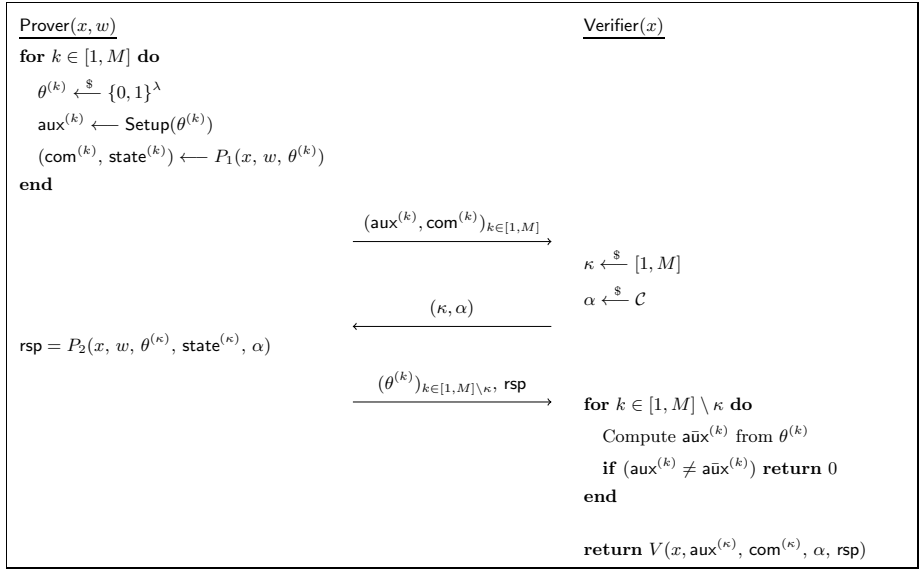
 
  \begin{center}
    \resizebox{1\textwidth}{!}{\fbox{
      \pseudocode{%
        \hspace{160pt} \> \> \hspace{160pt} \\[-0.75\baselineskip] % Hack to force font scaling
        \underline{\proverf(x, w)} \> \> \underline{\verifierf(x)} \\
        \pcfor k \in \intoneto{M} \pcdo \\
        \pcind \theta^{(k)} \sampler \bit^{\lambda} \\ 
        \pcind \aux^{(k)} \samplen \setup(\theta^{(k)}) \\ 
        \pcind (\com^{(k)}, \, \state^{(k)}) \samplen P_1(x, \, w, \, \theta^{(k)}) \\
        \pcend \\
        \> \sendmessageright*{(\aux^{(k)}, \com^{(k)})_{k \in \intoneto{M}}} \\[-0.5\baselineskip]
        \> \> \kappa \sampler \intoneto{M} \\ 
        \> \> \alpha \sampler \mathcal{C} \\[-\baselineskip]
        \> \sendmessageleft*{(\kappa, \alpha)} \\[-0.5\baselineskip]
        \rsp = P_2(x, \, w, \, \theta^{(\kappa)}, \, \state^{(\kappa)}, \, \alpha) \\
        \> \sendmessageright*{(\theta^{(k)})_{k \in \intoneto{M} \setminus \kappa}, \, \rsp} \\[-\baselineskip]
        \> \> \pcfor k \in \intoneto{M} \setminus \kappa \pcdo \\
        \> \> \pcind \tx{Compute } \bar{\aux}^{(k)} \tx{ from } \theta^{(k)} \\
        \> \> \pcind \pcif (\aux^{(k)} \neq \bar{\aux}^{(k)}) ~ \pcreturn 0 \\
        \> \> \pcend \\[0.75\baselineskip]
        \> \> \pcreturn V(x, \aux^{(\kappa)}, \, \com^{(\kappa)}, \, \alpha, \, \rsp)
      }
    }}
    \caption{3-round HVZK PoK from HVZK PoK with Trusted Helper \cite{Beullens20} \label{no-helper-fig-3r}}
  \end{center}
\end{figure}

\begin{figure}[H]
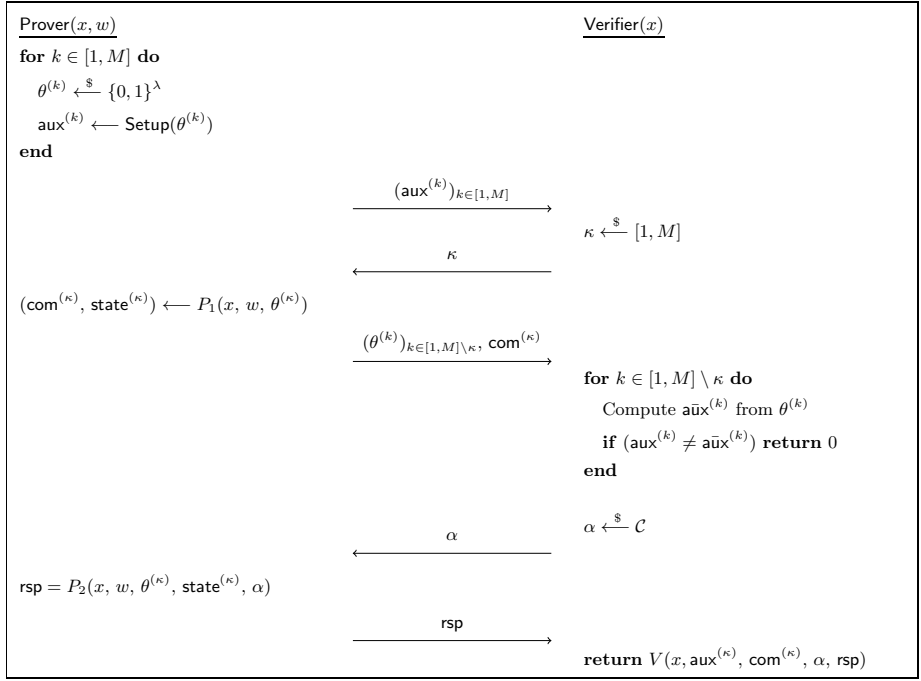
 
  \begin{center}
    \resizebox{1\textwidth}{!}{\fbox{
      \pseudocode{%
        \hspace{160pt} \> \> \hspace{160pt} \\[-0.75\baselineskip] % Hack to force font scaling
        \underline{\proverf(x, w)} \> \> \underline{\verifierf(x)} \\
        \pcfor k \in \intoneto{M} \pcdo \\
        \pcind \theta^{(k)} \sampler \bit^{\lambda} \\ 
        \pcind \aux^{(k)} \samplen \setup(\theta^{(k)}) \\ 
        \pcend \\
        \> \sendmessageright*{(\aux^{(k)})_{k \in \intoneto{M}}} \\[-0.5\baselineskip]
        \> \> \kappa \sampler \intoneto{M} \\[-0.5\baselineskip]
        \> \sendmessageleft*{\kappa} \\
        (\com^{(\kappa)}, \, \state^{(\kappa)}) \samplen P_1(x, \, w, \, \theta^{(\kappa)}) \\
        \> \sendmessageright*{(\theta^{(k)})_{k \in \intoneto{M} \setminus \kappa}, \, \com^{(\kappa)}} \\[-0.5\baselineskip]
        \> \> \pcfor k \in \intoneto{M} \setminus \kappa \pcdo \\
        \> \> \pcind \tx{Compute } \bar{\aux}^{(k)} \tx{ from } \theta^{(k)} \\
        \> \> \pcind \pcif (\aux^{(k)} \neq \bar{\aux}^{(k)}) ~ \pcreturn 0 \\
        \> \> \pcend \\[0.75\baselineskip]
        \> \> \alpha \sampler \mathcal{C} \\[-\baselineskip]
        \> \sendmessageleft*{\alpha} \\
        \rsp = P_2(x, \, w, \, \theta^{(\kappa)}, \, \state^{(\kappa)}, \, \alpha) \\
        \> \sendmessageright*{\rsp} \\[-0.5\baselineskip]
        \> \> \pcreturn V(x, \aux^{(\kappa)}, \, \com^{(\kappa)}, \, \alpha, \, \rsp)
      }
    }}
    \caption{5-round HVZK PoK from HVZK PoK with Trusted Helper \label{no-helper-fig-5r}}
  \end{center}
\end{figure}

\noindent \textbf{Removing the helper from PoK~2 and PoK 3.} For our PoK~2 and PoK~3 (see Section \ref{sec:pok2}), we will only consider the 3-round transformation as it leads to smaller signatures when taking account the attack from \cite{KZ20}.
We defer the reader to Appendices~\ref{app:pok2-3r-nopt} and \ref{app:pok3-3r-nopt} for the description of the protocols.

%--------------------------------------------------------------------%
\section{Communication Cost and Optimizations} \label{sec:optimizations}
%--------------------------------------------------------------------%

In this section, we present optimizations that permits to reduce the communication cost of the aforementioned proofs of knowledge.
Several optimizations are related to the use of the MPC-in-the-head paradigm and were first introduced in \cite{KKW18}. 
We also consider code-based related optimizations that were introduced in \cite{AGS11} and \cite{BGS21}.
Finally, we discuss a performance oriented optimization relying on the use of structured matrices.
The optimized versions of our 3-round PoK~1, 5-round PoK~1, 3-round PoK~2 and 3-round PoK~3 are described in Appendices~\ref{app:pok1-3r-opt}, \ref{app:pok1-5r-opt}, \ref{app:pok2-3r-opt} and \ref{app:pok3-3r-opt} respectively.

\vspace{\baselineskip}
\noindent \textbf{Protocol repetition} \cite{KKW18, Beullens20}\textbf{.}
The 3-round or 5-round PoK constructed by removing the helper (see Section \ref{sec:no-helper}) have a soundness error equal to $\max{(\frac{1}{M}, \frac{1}{N})}$.
In order to obtain a negligible soundness error with respect to security parameter $\lambda$, one can compute $\tau$ parallel executions of the protocol where $\tau = \frac{\lambda}{\log_2{(\min{(M, N)})}}$.
In that case, one needs to execute $\tau \cdot M$ setup steps (the $\helperf$ part in Figure~\ref{preliminaries-pok-fig1}) followed by $\tau$ executions of the protocol (the $\proverf$ part in Figure~\ref{preliminaries-pok-fig1}).
The key idea of the \emph{beating parallel repetition} optimization from \cite{KKW18} is to 
let the verifier choose $\tau$ out of $M$ (instead of only $1$ out of $M$) setups to execute, 
which means the setup
phase is repeated only $M$ times (instead of $\tau \cdot M$ times).
However, this comes at the cost of increasing the number
executions (higher $\tau$) in the protocol as explained below.
Suppose, a malicious prover computes $e \leq \tau$ setup steps incorrectly, it can only convince the verifier if the latter
chooses to execute all these $e$ setups (hence never verifying any of them) which happens with probability $\binom{M - e}{\tau - e} \cdot \binom{M}{\tau}^{-1}$.
Moreover, the malicious prover needs to be accepted for the remaining $\tau - e$ executions with honest setups which can happen with probability $\leq (\frac{1}{N})^{\tau - e}$.
Therefore, using this optimization, one needs to execute $M$ setup steps followed by $\tau$ executions of the protocol and the soundness error is given by $\max\limits_{0 \leq e \leq \tau} \binom{M - e}{\tau - e} \cdot \binom{M}{\tau}^{-1} \cdot N^{-(\tau - e)}$.

\begin{table}[H]
\begin{center}
{\renewcommand{\arraystretch}{1.3}
{\scriptsize
  \begin{tabular}{|l|c|c|c||c|c|c|}
      \cline{2-7}
      \multicolumn{1}{c|}{} & \multicolumn{3}{c||}{Parallel Repetition} & \multicolumn{3}{c|}{\cite{KKW18} Repetition} \\ \hline
      $N$          & 2   & 16  & 32  & 2   & 16  & 32  \\ \hline
      $M$          & 2   & 16  & 32  & 256 & 272 & 389 \\ \hline
      $\tau$       & 128 & 32  & 26  & 128 & 35  & 28  \\ \hline \hline
      \textbf{\# Setup}     & \textbf{256} & \textbf{512} & \textbf{832} & \textbf{256} & \textbf{272} & \textbf{389} \\ \hline
      \textbf{\# Execution} & \textbf{128} & \textbf{32}  & \textbf{26}  & \textbf{128} & \textbf{35}  & \textbf{28}  \\ \hline
  \end{tabular}
  \caption{Example of parameters trade-off for $\lambda = 128$. The optimization greatly reduce the Setup's number for only a small increase of the Execution's number.} %Note that, the optimized version results in almost $50 \%$ reduction in the number of times Setup is run while increasing only
  % about $10 \%$ of the executed instances.} \label{table:protocol-repetition}
}}
\end{center}
\end{table}

From here onward, the reader is advised to keep in mind that
during the protocol, the signer needs to
send information (seeds) corresponding to 
(i) $M - \tau$ instances
of the underlying PoK which are generated during the setup phase but
are only used to verify that the setup was run honestly 
and (ii) information associated with $\tau$
executions of the underlying PoK.

\vspace{\baselineskip}
\noindent \textbf{Seed compression} \cite{KKW18}\textbf{.} For all the aforementioned PoK, one has to send the seeds $(\theta^{(k)})_{k \in \intoneto{M} \setminus K}$ used to recompute the auxiliary information $(\aux^{(k)})_{k \in \intoneto{M} \setminus K}$ (with $K \sampler \{ K \subset \{1, \cdots, M \}, \, |K| = \tau \}$) associated to the setups that are not going to be executed.
As explained in \cite{KKW18}, one can use Merkle trees in order to reduce the cost of sending these seeds.
To this end, the prover samples a root seed and generates a binary tree of depth $\log(M)$ where each node is a seed derived from its parent node. 
Doing so, he can send to the verifier the nodes that allows to recompute all the leafs of the tree except the $\tau$ seeds that should not be revealed thus reducing the cost associated to the $(\theta^{(k)})_{k \in \intoneto{M} \setminus K}$ seeds from $(M - \tau) \cdot \lambda$ to $\tau \log_2(\frac{M}{\tau}) \cdot \lambda$ where $\lambda$ is the security parameter.
A crucial observation is that this optimization works well only when $\tau$ is small with respect to $M$. 
In particular for $\tau = M/2$, using a Merkle tree provides no benefit.
Hence, in the case of PoK~1 where $\tau = M/2$, we employ a variant of this optimization by considering $M/2$ binary trees of depth $1$ instead of single binary tree of depth $\log(M)$ which reduces the expected cost of sending the seeds to $3/4 \cdot (M - \tau) \cdot \lambda$ from $(M - \tau) \cdot \lambda$.

\vspace{\baselineskip}
\noindent \textbf{Commitment compression} \cite{KKW18, AGS11}\textbf{.} When considering 3-round PoK, the first prover's message contains the commitments $(\aux^{(k)}, \allowbreak \com^{(k)})_{k \in \intoneto{M}}$ where $\aux^{(k)} = (\com^{(k)}_i)_{i \in \intoneto{N}}$.
In order to reduce the cost associated to these commitments, one can instead send a unique commitment $h = \commit{r, \, (\aux^{(k)} \, || \, \allowbreak \com^{(k)})_{k \in \intoneto{M}}}$ as suggested in \cite{AGS11}.
Doing so, the prover has to give the verifier all the commitments that the latter cannot recompute himself in order to allow him to check the commitment $h$.
The situation is similar for 5-round PoK with $h = \commit{r, \, (\aux^{(k)})_{k \in \intoneto{M}}}$ and $h' = \commit{r', \, (\com^{(\kappa)})_{\kappa \in K}}$.
In this case, the verifier can recompute $h'$ by himself hence only the commitments contained in $h$ have to be considered.
To reduce the cost associated to these commitments, one can once again leverage Merkle trees as explained in \cite{KKW18}.
Indeed, the prover will generate a Merkle tree of his commitments (from bottom to top contrarily to the previous case) and send the root to the verifier as well as all the nodes of the tree that permits to recompute the root from the commitments that the verifier can obtain by himself.

In our PoK~1~(3-round), $(\aux^{(k)})_{k \in \intoneto{M}}$ contains $2M$ commitments while $(\com^{(k)})_{k \in \intoneto{M}}$ contains $M$ commitments.
For the $M - \tau$ instances that are not executed, these commitments can be recomputed by the verifier from the $(\theta^{(k)})_{k \in \intoneto{M} \setminus K}$ seeds.
For the $\tau$ instances that are executed, one commitment from $(\aux^{(\kappa)})_{\kappa \in K}$ can be recomputed by the verifier while the other one can be given in the prover's response $\rsp$.
In addition, all the commitments from $(\com^{(k)})_{k \in \intoneto{M} \setminus K}$ need to be given to the verifier.
Given that $\tau = M/2$, this reduces the cost of sending the $3M$ commitments of our PoK~1~(3-round) to $(1 + 3/4 \cdot (M - \tau) + 7/8 \cdot \tau) \cdot |\com|$.
For PoK~1~(5-round), the cost is reduced from $3M$ commitments to $(2 + 7/8 \cdot \tau) \cdot |\com|$.

For PoK~2 and PoK~3, sending the $(M - \tau)$ commitments $(\com^{(k)})_{k \in \intoneto{M} \setminus K}$ can be done using Merkle trees hence cost $\tau \log_2(\frac{M}{\tau}) \cdot |\com|$.
For the the $M - \tau$ instances that are not executed, the $N(M - \tau)$ commitments from $(\aux^{(k)})_{k \in \intoneto{M} \setminus K}$ can be recomputed by the verifier from the $(\theta^{(k)})_{k \in \intoneto{M} \setminus K}$ seeds.
Interestingly, for the $\tau$ instances that are executed, the situation differs between PoK~2 and PoK~3 which explains why the optimized version of PoK~3 outperforms the optimized version of PoK~2 although both non optimized versions are equivalent.
In the case of PoK~2, only one commitment from $(\aux^{(\kappa)})_{\kappa \in K}$ can be recomputed by the verifier while the other $(N - 1)$ have to be given to him as he can not recomputed them due to the presence of the  value $\bm{u}$.
Hence, sending these commitments using Merkle trees cost $\log(N) \cdot |\com|$.
In contrast, in the \cite{FJR21} setting, $(N - 1)$ commitments from $(\aux^{(\kappa)})_{\kappa \in K}$ can be recomputed by the verifier hence the cost associated to sending the commitments from $(\aux^{(\kappa)})_{\kappa \in K}$ is only $|\com|$.

\vspace{\baselineskip}
\noindent \textbf{Small weight vector compression.} \cite{AGS11}\textbf{.}
One can leverage the small weight of some vectors such as $\pi[\bm{x}]$ in PoK~1 or $\pi[\bm{e}]$ in PoK~2 and PoK~3 by using a compression algorithm before sending them.
Hence, the cost of sending small weight vectors is reduced from $n$ to $n/2$ approximately.

\vspace{\baselineskip}
\noindent \textbf{Additional vector compression} \cite{BGS21}\textbf{.}
This optimization is specific to PoK~1 in which the prover have to send a permutation of a random vector $\pi[\bm{u}]$.
Instead of doing this, one can sample a random value $\bm{v} \samples{\psi} \Ftn$ from some random seed $\psi$ and compute the value $\bm{u} = \pi^{-1}[\bm{v}]$.
When the prover has to send $\pi[\bm{u}]$, he can send $\bm{v}$ instead which can be substituted by the seed $\psi$.
Doing so, one reduce the cost of sending such vectors from $n$ to $\lambda$.

\vspace{\baselineskip}
\noindent \textbf{Improved performances from structured matrices.} In all aforementioned PoK, a matrix vector multiplication must be computed during each setup. 
In order to improve the performance of these protocols, one may choose to use structured matrices featuring an efficient matrix vector multiplication.
For example, one can use quasi-cyclic matrices as their matrix vector multiplication can be performed efficiently by polynomial multiplication.
In this case, the security of the protocol relies on the quasi-cyclic variants $\QCSD$ or $\QCGSD$ of the syndrome decoding problem.

\newpage
%--------------------------------------------------------------------%
\section{Signature Schemes} \label{sec:signatures}
%--------------------------------------------------------------------%

In this section, we explain how to transform our interactive HVZK PoKs without helper as detailed in Section~\ref{sec:no-helper} into digital signatures using the strong Fiat-Shamir heuristic \cite{FS, sFS}.
We also discuss the security of the resulting signature schemes in both the random oracle model (ROM) and the quantum random oracle model (QROM).

\vspace{\baselineskip}
\noindent The keystone idea of the Fiat-Shamir heuristic \cite{FS} is to ``emulate" the random challenge sampling from the verifier by a call to a hash function modelled as a random oracle, thereby turning an interactive
protocol into non-interactive protocol with access to random oracle.
%\footnote{This technique is applicable only to the protocols where the
%randomness used to compute the verifier's challenge is public. We note
%that since the challenges in our PoK are sampled uniform randomly, one
%can use Fiat-Shamir transformation.}.
%In order to sign a message $m$, the signer acts as the prover and computes
%the commitment $\com$ as the output of the first round of the PoK.
%It then derives the challenge $\ch$ by querying the random oracle on the message $m$, and the commitment $\com$ along with the public parameters.
%Next, it computes the response $\rsp$ and outputs the signature $\sigma = (\com, \rsp)$.
%In order to verify a signature $\sigma$ corresponding to a given message $m$,
%the verifier only needs to verify that $(\com, \ch, \rsp)$ is a valid transcript of the underlying PoK. This can be done by checking if the random oracle
%answers with $\ch$ when queried with $m$, and $\com$ along with the public parameters, and then verifying if $\rsp$ is accepted by the verifier of
%the underlying PoK.
Figures~\ref{fig:signature-3r} and \ref{fig:signature-5r} explain how to apply the Fiat-Shamir heuristic in the context of PoK with trusted helper.
%The signatures obtained by applying the Fiat-Shamir transformation on top of the 3-round HVZK PoK and the 5-round HVZK PoK from Section~\ref{sec:no-helper} are depicted in figures~\ref{fig:signature-3r} and \ref{fig:signature-5r} respectively.
We next present the signatures obtained by applying 
the Fiat-Shamir transform
to HVZK PoK schemes from Sections~\ref{sec:pok1} and~\ref{sec:pok23}.
Starting with our PoK~1 with Helper (Section \ref{sec:pok1}, Figure~\ref{pok1:fig1}), PoK~2 with Helper (Section \ref{sec:pok2}, Figure~\ref{fig:pok2-h}) and PoK~3 with Helper (Section \ref{sec:pok3}, Figure~\ref{fig:pok3-h}), one can remove the helper using the constructions from Section~\ref{sec:no-helper}.
Doing so, one get our non-optimized 3-round PoK~1 (Appendix~\ref{app:pok1-3r-nopt}), 5-round PoK~1 (Appendix~\ref{app:pok1-5r-nopt}), 3-round PoK~2 (Appendix~\ref{app:pok2-3r-nopt}) and 3-round PoK~3 (Appendix~\ref{app:pok3-3r-nopt}).
Hereafter, we assume that these protocols provide a negligible soundness error which in practice implies to perform a parallel repetition of the PoK.
By applying the results from Section~\ref{sec:optimizations}, we obtain the optimized versions of our 3-round PoK~1 (Appendix~\ref{app:pok1-3r-opt}), 5-round PoK~1 (Appendix~\ref{app:pok1-5r-opt}), 3-round PoK~2 (Appendix~\ref{app:pok2-3r-opt}) and 3-round PoK~3 (Appendix~\ref{app:pok3-3r-opt}).
The optimized versions of our PoK include protocol repetitions hence achieve a negligible soundness error.
Then, using the Fiat-Shamir transformation presented in Figures~\ref{fig:signature-3r} and \ref{fig:signature-5r}, we construct four signatures denoted Sig~1~(3-round), Sig~1~(5-round), Sig~2~(3-round) and Sig~3~(3-round) 
that can be found in Appendices \ref{app:sig1-3r}, \ref{app:sig1-5r}, \ref{app:sig2} and \ref{app:sig3} respectively.
%In a nutshell, Sig~1~(3-round) is constructed using Figures~\ref{pok1:fig1}, \ref{no-helper-fig-3r} and \ref{fig:signature-3r}; Sig~1~(5-round) is constructed using Figures~\ref{pok1:fig1}, \ref{no-helper-fig-5r} and~\ref{fig:signature-5r}; Sig~2~(3-round) is constructed using Figures~\ref{fig:pok2-h}, \ref{no-helper-fig-3r} and~\ref{fig:signature-3r} while Sig~3~(3-round) is constructed using Figures~\ref{fig:pok3-h}, \ref{no-helper-fig-3r} and~\ref{fig:signature-3r}.
Similarly to the signatures constructed in \cite{Beullens20} and \cite{GPS21}, our security theorems apply to the non-optimized versions of our protocols while the optimized versions are the ones considered in practice.

\vspace{\baselineskip}
Signatures built from the Fiat-Shamir heuristic have been proven existentially unforgeable in the ROM whenever the underlying HVZK PoK achieves a negligible soundness error, see \cite{PS96,PS00, FSmin}.
These security guarantees can be extended to the QROM model following 
the work of~\cite{EC:Unruh12,FOCS:Zhandry12,EC:Unruh16,AC:Unruh17,EC:KilLyuSch18,C:LiuZha19,C:DFMS19,C:DonFehMaj20}.
Similarly to the signatures constructed in \cite{Beullens20} and \cite{GPS21}, our security theorems apply to the non-optimized versions of our protocols while the optimized versions are the ones considered in practice.
On a high level, this line of research has shown that the (multi-round)
Fiat-Shamir heuristic preserves the soundness and the other proof of knowledge
properties in the QROM setting and that transforming such protocols
into signature schemes provides (strong) existential unforgeability guarantees.
%We refer the interested readers 
%to~\cite{EC:Unruh12,EC:Unruh16,AC:Unruh17,EC:KilLyuSch18,C:DFMS19,C:DonFehMaj20}
%for more details.
Before stating the security properties of the schemes proposed in this work, we define the ``computationally unique responses'' property which is required for the proof.

%\begin{theorem}
%\textcolor{blue}{
%If the hash function used as commitment scheme is modeled as a Quantum Random Oracle and that a Quantum Random Oracle Model is used as PRG, then the non-optimized variants of Sig~1~(3-round), Sig~1~(5-round) and Sig~2~(5-round) signature schemes are strongly existentially unforgeable in the QROM.
%}
%\end{theorem}

\begin{figure}[H] 
  \begin{center}
    \resizebox{1\textwidth}{!}{\fbox{
      \pseudocode{%
        \hspace{160pt} \> \> \hspace{160pt} \\[-0.75\baselineskip] % Hack to force font scaling
        \underline{\keygen(x, w)} \\
        \pk = x, ~ \sk = w \\
        \pcreturn (\pk, \sk) \\[0.75\baselineskip]
        \underline{\sign(\pk, \sk, m)} \\
        \pcfor k \in \intoneto{M} \pcdo \\
        \pcind \theta^{(k)} \sampler \bit^{\lambda} \\ 
        \pcind \aux^{(k)} \samplen \setup(\theta^{(k)}) \\ 
        \pcind (\com^{(k)}, \, \state^{(k)}) \samplen P_1(\pk, \, \sk, \, \theta^{(k)}) \\
        \pcend \\
        \com = (\aux^{(k)}, \com^{(k)})_{k \in \intoneto{M}} \\
        (\kappa, \, \alpha) \samplen \ro{m \, || \, \pk \, || \, \com} \\
        \rsp' = P_2(\pk, \, \sk, \, \theta^{(\kappa)}, \, \state^{(\kappa)}, \, \alpha) \\
        \rsp = ((\theta^{(k)})_{k \in \intoneto{M} \setminus \kappa}, \, \rsp') \\
        \pcreturn \sigma = (\com, \rsp) \\[0.75\baselineskip]
        \underline{\verif(\pk, \sigma, m)} \\
        (\kappa, \, \alpha) \samplen \ro{m \, || \, \pk \, || \, \com} \\
        \pcfor k \in \intoneto{M} \setminus \kappa \pcdo \\
        \pcind \tx{Compute } \bar{\aux}^{(k)} \tx{ from } \theta^{(k)} \\
        \pcind \pcif (\aux^{(k)} \neq \bar{\aux}^{(k)}) ~ \pcreturn 0 \\
        \pcend \\
        \pcreturn V(\pk, \, \aux^{(\kappa)}, \, \com^{(\kappa)}, \, \alpha, \, \rsp')
      }
    }}
    \caption{Signature from Fiat-Shamir Heuristic applied to Figure~\ref{no-helper-fig-3r}} \label{fig:signature-3r}
  \end{center}
\end{figure}

\begin{figure}[H] 
  \begin{center}
    \resizebox{1\textwidth}{!}{\fbox{
      \pseudocode{%
        \hspace{160pt} \> \> \hspace{160pt} \\[-0.75\baselineskip] % Hack to force font scaling
        \underline{\keygen(x, w)} \\
        \pk = x, ~ \sk = w \\
        \pcreturn (\pk, \sk) \\[0.75\baselineskip]
        \underline{\sign(\pk, \sk, m)} \\
        \pcfor k \in \intoneto{M} \pcdo \\
        \pcind \theta^{(k)} \sampler \bit^{\lambda} \\ 
        \pcind \aux^{(k)} \samplen \setup(\theta^{(k)}) \\ 
        \pcend \\
        \com_1 = (\aux^{(k)})_{k \in \intoneto{M}} \\
        \kappa \samplen \ro{m \, || \, \pk \, || \, \com_1} \\
        (\com^{(\kappa)}, \, \state^{(\kappa)}) \samplen P_1(\pk, \, \sk, \, \theta^{(\kappa)}) \\
        \com_2 = \com^{(\kappa)} \\
        \alpha \samplen \ro{m \, || \, \pk \, || \, \com_1 \, || \, \com_2} \\
        \rsp' = P_2(\pk, \, \sk, \, \theta^{(\kappa)}, \, \state^{(\kappa)}, \, \alpha) \\
        \rsp = ((\theta^{(k)})_{k \in \intoneto{M} \setminus \kappa}, \, \rsp') \\
        \pcreturn \sigma = (\com_1, \, \com_2, \, \rsp) \\[0.75\baselineskip]
        \underline{\verif(\pk, \sigma, m)} \\
        \kappa \samplen \ro{m \, || \, \pk \, || \, \com_1} \\
        \alpha \samplen \ro{m \, || \, \pk \, || \, \com_1 \, || \, \com_2} \\
        \pcfor k \in \intoneto{M} \setminus \kappa \pcdo \\
        \pcind \tx{Compute } \bar{\aux}^{(k)} \tx{ from } \theta^{(k)} \\
        \pcind \pcif (\aux^{(k)} \neq \bar{\aux}^{(k)}) ~ \pcreturn 0 \\
        \pcend \\
        \pcreturn V(\pk, \, \aux^{(\kappa)}, \, \com^{(\kappa)}, \, \alpha, \, \rsp')
      }
    }}
    \caption{Signature from Fiat-Shamir Heuristic applied to Figure~\ref{no-helper-fig-5r}} \label{fig:signature-5r}
  \end{center}
\end{figure}

\newpage

\begin{definition} [Computationally Unique Responses] \label{def:comp-uni-res}
	A $(2n+1)$ round public-coin interactive PoK is said to have
	\emph{computationally unique responses} if given a partial transcript
	$(\com, \, \ch_1, \, \rsp_1, \, \ldots, \, \ch_i)$ it is computationally 
	infeasible to find two accepting  conversations which share the first
	$(2i )$ messages as above but differ in at least one position.
	That is, 
	$$\prb \left[ V(\trans_1) = \accept \bigwedge V(\trans_2) = \accept ~\big|~ (\trans_1, \trans_2) \leftarrow \adv \right]$$
	is negligible for computationally bounded (quantum) adversary $\adv$,
	where,
	$\trans_b = (\com, \, \ch_1, \, \rsp_1, \, \ldots, \, \ch_i, \, \rsp^b_i, \ch^b_{i+1}, \, \rsp^b_{i+1}, \, \ldots , \, \ch^b_{n}, \, \rsp^b_{n}, \,)$
	for $b \in \{0,1\}$ such that $\rsp^0_i \neq \rsp^1_i$.
\end{definition}

\begin{theorem} \label{thm:qrom}
If the non-optimized variants of Sig~1~(3-round), Sig~1~(5-round), Sig~2 and Sig~3 signature schemes are instantiated with a collapsing hash function
as commitment scheme, then the  signature schemes are strong
existential unforgeable under chosen message attack (sUF-CMA) in the QROM.

\end{theorem}

\begin{proof}
  The proof of Theorem~\ref{thm:qrom} is similar to those analyzing the
  security of (multi-round) Fiat-Shamir transformation of PoK in the QROM~\cite{C:DFMS19,C:DonFehMaj20}.
  Here, we note that the first message is a commitment generated using
  a collapsing hash function and hence is unpredictable.
  Also, the second message in case of Sig~1~(5-round) is computed as a function of
  the committed values in first message. Similar is the case for the final
  responses, which additionally includes some opening information.
  Due to the binding property  of the commitment, the second message 
  and the response for each of the schemes is computationally unique.
  As shown earlier, the schemes are also HVZK.
  This suffices to prove the 
  strong existential unforgeable under chosen message attack (sUF-CMA)
  property of the schemes following Theorem 23, Theorem 28, Corollary 24,
  Corollary 29, and Corollary 30 from~\cite{C:DonFehMaj20}.
\end{proof}

%--------------------------------------------------------------------%
\section{Parameters and Comparison} \label{sec:parameters}
%--------------------------------------------------------------------%

\subsection{Parameters choice}

The system parameters $(n, k, w, M, N, \tau)$ are chosen such that all known attacks cost more than $2^{\lambda}$ elementary operations for a given security parameter $\lambda$.
The parameters $(n, k, w)$ are related to the difficulty of solving the underlying decoding problems while $(M, N, \tau)$ are related to the soundness of the PoK.
Resulting parameters are given in Table~\ref{table:param1}.

\vspace{\baselineskip}
\noindent \textbf{Decoding attack.} We consider decoding problems instantiated with binary $[n, k]$ codes and secrets of small weight $w$.
Parameters are chosen according to the BJMM generic attack \cite{BJMM12} along with estimates from \cite{HS13}.
For the PoK leveraging quasi-cyclicity, we take into account the DOOM attack from \cite{DOOM} which reduces the complexity by a factor $\sqrt n$.

\vspace{\baselineskip}
\noindent \textbf{Soundness error.} When considering the \emph{beating parallel repetition} optimization (see Section \ref{sec:optimizations}), one has to chose $(M, N, \tau)$ such that the soundness error $\max\limits_{0 \leq e \leq \tau} \binom{M - e}{\tau - e} \cdot \binom{M}{\tau}^{-1} \cdot N^{-(\tau - e)}$ is negligible with respect to $\lambda$.
In practice, this offer a trade-off between performances and signature sizes as one can increase $M$ and $N$ (degrading running time) in order to reduce $\tau$ (improving signature size).
We illustrate this trade-off by providing several parameter sets in Table~\ref{table:param1}.

\vspace{\baselineskip}
\noindent \textbf{Attack against 5-round protocols.} An attack exploiting the structure of 5-round PoK has been identified in \cite{KZ20}.
The main idea is to split the attacker work in two phases: (i) initially it tries to guess the first challenge for several repetitions and then (ii) to guess the second challenge for the remaining repetitions.
Our schemes feature the \emph{capability for early abort} described in \cite{KZ20} therefore the cost of the attack is equal to $|\mathcal{C}_1|^{\tau^*} + |\mathcal{C}_2|^{\tau - \tau^*}$ where $\mathcal{C}_1$ and $\mathcal{C}_2$ are the challenge spaces and $\tau^* \in [0, \tau]$ is chosen by the adversary to minimize the attack's cost.

\subsection{Resulting key and signature sizes}

\vspace{\baselineskip}
\noindent We now explain how to compute the sizes of the signatures Sig~1~(3-round), Sig~1~(5-round), Sig~2 and Sig~3 (see Appendices \ref{app:sig1-3r}, \ref{app:sig1-5r}, \ref{app:sig2} and \ref{app:sig3}).
Hereafter, we consider commitments instantiated from hash functions such that $|\com| = 2\lambda$ bits.
Resulting average sizes are given in Table~\ref{table:param1}.
In practice, the size of these signature varies depending on the challenges received.

\vspace{\baselineskip}
\noindent \textbf{Key pair.} The key pair of the signature built from PoK~1 is defined by $\sk = (\bm{x})$ and $\pk = (\bm{H}, \bm{y}^{\top} = \bm{H} \bm{x}^{\top})$ while the key pair of the signature built from PoK~2 and PoK~3 is given by $\sk = (\bm{x}, \bm{e})$ and $\pk = (\bm{G}, \bm{y} = \bm{x} \bm{G} + \bm{e})$.
Both the secret values ($\bm{x}$ and $(\bm{x}, \bm{e})$) and the public matrices ($\bm{H}$ and $\bm{G}$) can be generated from seeds.
Therefore, $\sk$ has size $\lambda$ bits in both cases while $\pk$ is $(n - k) + \lambda$ bits long for PoK~1 and $n + \lambda$ bits long for PoK~2 and PoK~3.

\vspace{\baselineskip}
\noindent \textbf{Signature from PoK~1.} In our Sig~1 (3-round), the signer has to send $(h, \xi, (\theta^{(k)}, \allowbreak \com^{(k)})_{k \in \intoneto{M} \setminus K}, (\rsp^{(\kappa)})_{\kappa \in K})$.
In the 5-round variant, the signer send $(h, h', \xi, \allowbreak (\theta^{(k)})_{k \in \intoneto{M} \setminus K}, (\rsp^{(\kappa)})_{\kappa \in K})$.
Both $h$ and $h'$ are commitments of size $|\com|$ bits while $\xi$ is a $\lambda$ bits long seed.  
The $M - \tau$ values $\theta^{(k)}$ are the seeds corresponding to the setups that are checked by the verifier without being executed.
They can be sent using the aforementioned Merkle tree based optimization (see Section~\ref{sec:optimizations}).
In the case of our Sig~1 parameters where $\tau = M/2$, the cost associated to the $\theta^{(k)}$ seeds is equal to $3/4 \cdot (M - \tau) \cdot \lambda$ bits (following the same argument as the one used for the seed compression in the response optimization).
The $(M - \tau)$ commitments $\com^{(k)}$ used in the 3-round variant can be sent in a similar way using $3/4 \cdot (M - \tau) \cdot |\com|$ bits.
The response $\rsp^{(\kappa)}$ differs with respect to the value of the challenge $\alpha$.
It either contains $(\phi^{(\kappa)}, \bm{u}^{(\kappa)} + \bm{x}, \com^{(\kappa)}_1)$ when $\alpha = 0$ or $(\psi^{(\kappa)}, \pi^{(\kappa)}[\bm{x}], \com^{(\kappa)}_0)$ when $\alpha = 1$.
The values $\phi^{(\kappa)}$ and $\psi^{(\kappa)}$ are seeds, the value $\bm{u}^{(\kappa)} + \bm{x}$ is a vector of size $n$ and $\pi^{(\kappa)}[\bm{x}]$ can be sent using $n/2$ bits thanks to the small weight vector compression optimization.
The cost of sending $\com^{(\kappa)}_0$ and $\com^{(\kappa)}_1$ can be reduced from $2|\com|$ to $7/8 \cdot 2|\com|$ using binary trees of length $1$ although this is less efficient than for $(\com^{(k)})_{k \in \intoneto{M} \setminus K}$ as one have to take into account the fact that some instances are not going to be executed.
Thus, the cost associated to the $\tau$ responses $\rsp^{(\kappa)}$ used in the PoK~1 is equal to $\tau/2 \cdot (3n/2 + 2 \lambda + 7/8 \cdot 2|\com|)$.
Overall, the signature constructed from 3-round PoK~1 has a size equal to $(1 + 3/4 \cdot (M - \tau)) \cdot (\lambda + |\com|) + \tau/2 \cdot (3n/2 + 2 \lambda + 7/8 \cdot 2|\com|) \approx \tau \cdot (0.75 n + 5 \lambda) $ bits.
Similarly, the signature constructed from 5-round PoK~1 has a size equal to $(2|\com| + \lambda) + 3/4 \cdot (M - \tau) \cdot \lambda + \tau/2 \cdot (3n/2 + 2 \lambda + 7/8 \cdot 2|\com|) \approx \tau \cdot (0.75 n + 3.5 \lambda)$ bits.
We have provided parameters suitable for the $\SD$ problem and its quasi-cyclic variant $\QCSD$ in Table~\ref{table:param1}.

\vspace{\baselineskip}
\noindent \textbf{Signature from PoK~2.} In our Sig~2, the signer has to send $(h, \xi, (\theta^{(k)}, \allowbreak \com^{(k)})_{k \in \intoneto{M} \setminus K}, (\rsp^{(\kappa)})_{\kappa \in K})$.
The value $h$ is commitment of size $|\com|$ bits while $\xi$ is a $\lambda$ bits long seed.  
The $(M - \tau)$ values $\theta^{(k)}$ and $\com^{(k)}$ are respectively the seeds and commitments  corresponding to the setups that are checked by the verifier without being executed.
By leveraging the Merkle tree based optimization, their cost is respectively equal to $\tau \log_2(M/\tau) \cdot \lambda$ bits and $\tau \log_2(M/\tau) \cdot |\com|$ bits.
The prover's response contains $(\bm{u}^{(\kappa)}_\alpha + \bm{x}, \pi^{(\kappa)}_{\alpha}[\bm{e}], \bm{s}^{(\kappa)}_\alpha, \allowbreak \theta^{(\kappa)}_{\alpha^*}, \aux^{(\kappa)}_{\alpha^*})$. 
The vectors $\bm{u}^{(\kappa)}_\alpha + \bm{x}$ and $\bm{s}^{(\kappa)}_\alpha$ are of size $k = n/2$ and $n$ respectively.
In addition, $\pi^{(\kappa)}_{\alpha}[\bm{e}]$ can be sent using $n/2$ bits thanks to the small weight vector compression optimization.
The values $\theta^{(\kappa)}_{\alpha^*}$ and $\aux^{(\kappa)}_{\alpha^*}$ contains respectively $N - 1$ seeds and commitments that can be sent using $\lambda \cdot \log_2(N)$ bits and $|\com| \cdot \log_2(N)$ bits.
Hence the cost associated to the $\tau$ responses is equal to $\tau \cdot (2n + (\lambda + |\com|) \cdot \log_2(N))$ bits.
Overall, the signature constructed from 3-round PoK~2 has a size equal to $(|\com| + \lambda) \cdot (1 + \tau \log_2(M/\tau)) + \tau \cdot (2n + (\lambda + |\com|) \cdot \log_2(N)) \approx \tau \cdot (2n + 3 \lambda \cdot [\log_2(M/\tau) + \log_2(N)])$ bits.

\vspace{\baselineskip}
\noindent \textbf{Signature from PoK~3.} The case of Sig~3 is very similar to the one of Sig~2 except that it benefits more of the commitment compression optimization as explained in Section~\ref{sec:optimizations}.
As a result, the cost associated to the $\tau$ responses is equal to $\tau \cdot (2n + \lambda \cdot \log_2(N) + |\com|)$ instead of $\tau \cdot (2n + (\lambda + |\com|) \cdot \log_2(N))$.
Hence, the signature constructed from PoK~3 is of size $(|\com| + \lambda) \cdot (1 + \tau \log_2(M/\tau)) + \tau \cdot (2n + \lambda \cdot \log_2(N) + |\com|) \approx \tau \cdot (2n + \lambda \cdot [3\log_2(M/\tau) + \log_2(N) + 2])$ bits.

\vspace{0.5\baselineskip}
\begin{table}[ht]
\begin{center}
{\renewcommand{\arraystretch}{1.3}
{\scriptsize
  \begin{tabular}{|l|c|c|c|c|c|c|c|c|}
    \cline{2-9}
      \multicolumn{1}{c|}{} & $n$ & $k$ & $w$ & $M$ & $N$ & $\tau$ & $\pk$ & $\sigma$ \\ \hline 
      Sig 1 (3-round) [$\QCSD$] & 1238 & 619 & 137 & \multirow{2}{*}{256}  & \multirow{2}{*}{2}   & \multirow{2}{*}{128} & 0.1 kB & 25.2 kB \\ \cline{1-4} \cline{8-9}
      Sig 1 (3-round) [$\SD$]   & 1190 & 595 & 132 &                       &                      &                      & 0.1 kB & 24.6 kB \\ \cline{1-4} \cline{5-9}
      Sig 1 (5-round) [$\QCSD$] & 1238 & 619 & 137 & \multirow{2}{*}{256}  & \multirow{2}{*}{2}   & \multirow{2}{*}{143} & 0.1 kB & 24.3 kB \\ \cline{1-4} \cline{8-9}
      Sig 1 (5-round) [$\SD$]   & 1190 & 595 & 132 &                       &                      &                      & 0.1 kB & 23.7 kB \\ \hline 
      \multirow{3}{*}{Sig 2 (3-round) [$\QCGSD$]}  & \multirow{3}{*}{1238} & \multirow{3}{*}{619} & \multirow{3}{*}{137} & 272 & 16 & 35 & \multirow{3}{*}{0.2 kB}  & 22.6 kB \\ \cline{5-7} \cline{9-9}
                                                   &                       &                      &                      & 389 & 32 & 28 &                          & 20.6 kB \\ \cline{5-7} \cline{9-9}	
                                                   &                       &                      &                      & 631 & 64 & 23 &                          & 19.3 kB \\ \hline
      \multirow{3}{*}{Sig 3 (3-round) [$\QCGSD$]}  & \multirow{3}{*}{1238} & \multirow{3}{*}{619} & \multirow{3}{*}{137} & 272 & 16 & 35 & \multirow{3}{*}{0.2 kB}  & 19.3 kB \\ \cline{5-7} \cline{9-9}
                                                   &                       &                      &                      & 389 & 32 & 28 &                          & 17.0 kB \\ \cline{5-7} \cline{9-9}	
                                                   &                       &                      &                      & 631 & 64 & 23 &                          & 15.6 kB \\ \hline
  \end{tabular}
  \caption{Parameters and sizes for Sig~1, Sig~2 and Sig~3 ($\lambda = 128$)} \label{table:param1}
}}
\end{center}
\end{table}
\vspace{-\baselineskip}

\newpage
\subsection{Comparison to other code-based signatures}

We first compare our new signatures to existing code-based signatures based on proofs on knowledge for the syndrome decoding. Next, we extend this comparison to any code-based signatures.

\vspace{\baselineskip}
\noindent \textbf{Comparison between code-based signatures using PoK.} We compare the Stern \cite{Stern93}, Véron \cite{Veron97,ISIT21}, AGS \cite{AGS11,ISIT21}, BGS \cite{BGS21}, CVE \cite{CVE11}, GPS \cite{GPS21} and FJR \cite{FJR21} schemes to our new proposals according to their security, sizes and performances.
To provide a meaningful comparison, we have updated the parameters of old schemes so that they achieve a security level $\lambda = 128$ and have applied recent optimizations to old schemes whenever relevant.
We have used the parameters $n = 1190$, $k = 595$ and $w = 132$ for the schemes based on the $\SD$ problem (without quasi-cyclicity), $n = 1238$, $k = 619$ and $w = 137$ for the schemes based on the $\QCSD$ problem and $n = 226$, $q = 256$, $k = 113$ and $w = 86$ for the schemes based on the $\SD$ problem over $\Fq$.
For recent GPS and FJR proposals, we have used the parameters proposed by their respective authors in \cite{GPS21} and \cite{FJR21}.
As the parameters used in \cite{FJR21} differ from the ones used here, this induces small differences between FJR and Sig~3 although both schemes can achieve similar results.
As several of these schemes can be based either on the $\SD$ or $\QCSD$ problem, we have indicated the case considered in the comparison (which matches the initial design of the scheme) and have indicated the other possible case using parenthesis whenever relevant.
Since there is no implementation implementation available for most of these schemes yet, we provide an estimate of their relative performances.
For all these schemes, the first step (every operations executed by the prover before he outputs its first commitment) can be seen as repeating $\mu$ times the computation of $\nu$ operations whose cost is arbitrarily denoted as one $\texttt{cost\_unit}$.
This first step hence costs $\mu \cdot \nu \cdot \texttt{cost\_unit}$.
Using our PoK~2 for illustrative purposes, one can see that $\mu = M$, $\nu = N$ and the $\texttt{cost\_unit}$ encompasses all the operations required to compute $\com^{(k)}_{i}$ for a given $k \in \intoneto{M}$ and $i \in \intoneto{N}$.
We believe, this constitutes a good estimate of the relative performances between these schemes as this step will likely dominates their overall cost.
Indeed, the $\texttt{cost\_unit}$ generally contains the most costly operations (matrix / vector multiplication, randomness sampling and hash computation) and is repeated $\mu \cdot \nu$ time in order to obtain a negligible soundness error.
We define our cost estimate as $\mu \cdot \nu$ thus assuming that the $\texttt{cost\_unit}$ is similar for each schemes.
This introduces an approximation in our comparison which could only be solved by providing and benchmarking actual implementations of the aforementioned schemes.
In particular, this approximation hides the performance difference between using plain matrices and structured ones which is not negligible in practice.
As such, one should compare the schemes whose $\texttt{cost\_unit}$ includes a matrix / vector multiplication (the ones based on the plain $\SD$ problem) separately from the schemes whose $\texttt{cost\_unit}$ features an efficient one thanks to structured matrices (the ones based on the $\QCSD$ problem).
Moreover, the GPS scheme does not include such a multiplication hence its real $\texttt{cost\_unit}$ is likely to be smaller than the one of other schemes which means that the proposed estimate might overestimate its real cost.
Results are displayed in Table~\ref{table:param2}.
One can see that our new constructions provide various trade-offs between security assumptions, performances and sizes for code-based signatures built from PoK.
In particular, Sig~1 brings improvement with respect to Stern, Véron, AGS and BGS protocols at the cost of a very small overhead.
If one is willing to accept a greater performance overhead, then PoK~3 (and PoK~2 to a lesser extent) permits to achieve even smaller signature sizes.

\vspace{0.5\baselineskip}
\begin{table}[ht]
\begin{center}
{\renewcommand{\arraystretch}{1.3}
{\scriptsize
  \begin{tabular}{|l|c|c|c|c|c|l|}
    \cline{2-7}
    \multicolumn{1}{c|}{}  & \multicolumn{3}{c|}{Performance} & \multicolumn{2}{c|}{Size} & \multirow{2}{*}{Security Assumption} \\ \cline{2-6}
    \multicolumn{1}{c|}{}  & $\mu$ & $\nu$ & Cost & $\pk$ & $\sigma$ & \\ \hline 
    Stern                            & 219  & 2    & 438       & 0.1 kB                   & 37.6 kB                   & $\SD$ (or $\QCSD$) over $\Ft$ \\ \hline 
    Véron                            & 219  & 2    & 438       & 0.2 kB                   & 31.2 kB                   & $\SD$ (or $\QCSD$) over $\Ft$ \\ \hline 
    CVE                              & 156  & 2    & 312       & 0.2 kB                   & 32.6 KB                   & $\SD$ (or $\QCSD$) over $\Fq$ \\ \hline
    \multirow{2}{*}{AGS}             & 151  & 2    & 302       & 0.2 kB                   & 30.5 kB                   & \multirow{2}{*}{$\QCSD$/$\DiffSD$ over $\Ft$} \\ \cline{2-6} 
                                     & 141  & 2    & 282       & 3.1 kB                   & 28.5 kB                   & \\ \hline
    \multirow{2}{*}{BGS}             & 151  & 2    & 302       & 0.1 kB                   & 25.2 kB                   & \multirow{2}{*}{$\QCSD$/$\DiffSD$ over $\Ft$} \\ \cline{2-6} 
                                     & 141  & 2    & 282       & 1.7 kB                   & 23.5 kB                   & \\ \hline 
    \multirow{4}{*}{GPS}             & 512  & 128  & 65 536    & 0.2 kB                   & 27.1 kB                   & \multirow{4}{*}{$\SD$ (or $\QCSD$) over $\Fq$} \\ \cline{2-6}
                                     & 1024 & 256  & 262 144   & 0.2 kB                   & 24.0 kB                   & \\ \cline{2-6}
                                     & 2048 & 512  & 1 048 576 & 0.2 kB                   & 21.3 kB                   & \\ \cline{2-6}
                                     & 4096 & 1024 & 4 194 304 & 0.2 kB                   & 19.8 kB                   & \\ \hline
    \multirow{2}{*}{FJR}             & 187  & 8    & 1496      & 0.1 kB                   & 24.4 kB                   & \multirow{2}{*}{$\SD$ (or $\QCSD$) over $\Ft$} \\ \cline{2-6}
                                     & 389  & 32   & 12 448    & 0.1 kB                   & 17.6 kB                   & \\ \hline
    \multirow{2}{*}{Sig 1 (3-round)} & 256  & 2    & 512       & 0.1 kB                   & 25.2 kB                   & $\QCSD$ over $\Ft$ \\ \cline{2-7}
                                     & 256  & 2    & 512       & 0.1 kB                   & 24.6 kB                   & $\SD$ over $\Ft$ \\ \hline
    \multirow{2}{*}{Sig 1 (5-round)} & 256  & 2    & 512       & 0.1 kB                   & 24.3 kB                   & $\QCSD$ over $\Ft$ \\ \cline{2-7}
                                     & 256  & 2    & 512       & 0.1 kB                   & 23.7 kB                   & $\SD$ over $\Ft$ \\ \hline
    \multirow{3}{*}{Sig 2}           & 272  & 16   & 4352      & 0.2 kB                   & 22.6 kB                   & \multirow{3}{*}{$\QCSD$ (or $\SD$) over $\Ft$} \\ \cline{2-6}
                                     & 389  & 32   & 12 448    & 0.2 kB                   & 20.6 kB                   & \\ \cline{2-6}	
                                     & 631  & 64   & 40 384    & 0.2 kB                   & 19.3 kB                   & \\ \hline
    \multirow{3}{*}{Sig 3}           & 272  & 16   & 4352      & 0.2 kB                   & 19.3 kB                   & \multirow{3}{*}{$\QCSD$ (or $\SD$) over $\Ft$} \\ \cline{2-6}	
                                     & 389  & 32   & 12 448    & 0.2 kB                   & 17.0 kB                   & \\ \cline{2-6}	
                                     & 631  & 64   & 40 384    & 0.2 kB                   & 15.6 kB                   & \\ \hline	
  \end{tabular}
  \caption{Comparison between code-based signatures built from PoK ($\lambda = 128$)} \label{table:param2}
}}
\end{center}
\end{table}

\noindent \textbf{Comparison to other signatures.} We provide a comparison with existing code-based signatures (including Wave \cite{wave19}, LESS \cite{less,less-fm} and Durandal \cite{Durandal}) in Table~\ref{table:param3}.
The LESS scheme relies on the code equivalence problem while Wave relies on both the $\SD$ problem (with large weight) and the indistinguishability of generalized $(U,U + V )$-codes.
We also include Durandal in our comparison even if it is a scheme based on the rank metric.
As such, it would be better to compare it with the rank-metric variants of our schemes.
Such variants are straightforward and are discussed in Section~\ref{sec:generalization}.

\vspace{0.5\baselineskip}
\begin{table}[ht]
\begin{center}
{\renewcommand{\arraystretch}{1.3}
{\scriptsize
  \begin{tabular}{|l|c|c|c|l|}
    \cline{2-5}
    \multicolumn{1}{c|}{} & $\pk$ & $\sigma$ & $\pk + \sigma$ & \multicolumn{1}{c|}{Security Assumption} \\ \hline 
    \multirow{2}{*}{Wave}     & \multirow{2}{*}{3.2 MB}  & \multirow{2}{*}{0.93 kB}  & \multirow{2}{*}{3.3 MB}   & Syndrome decoding over $\mathbb{F}_3$ (large weight) \\
                              &                          &                           &                           & Generalized $(U,U + V )$-codes indistinguishability  \\ \hline 
    \multirow{3}{*}{LESS}     & 9.8 kB                   & 15.2 kB                   & 25.0 kB                   & Linear Code Equivalence \\ \cline{2-5}
                              & 206 kB                   & 5.3 kB                    & 212 kB                    & \multirow{2}{*}{Permutation Code Equivalence} \\ \cline{2-4}
                              & 11.6 kB                  & 10.4 kB                   & 22.0 kB                   &  \\ \hline
    \multirow{2}{*}{Durandal} & 15.3 kB                  & 4.1 kB                    & 19.4 kB                   & Rank syndrome decoding over $\mathbb{F}_{2^m}$ (small weight) \\ \cline{2-4}
                              & 18.6 kB                  & 5.1 kB                    & 23.7 kB                   & Product spaces subspaces indistinguishability \\ \hline
    \multirow{2}{*}{Sig 1}    & 0.1 kB                   & 24.6 kB                   & 24.6 kB                   & \multirow{2}{*}{Syndrome decoding over $\Ft$ (small weight)} \\ \cline{2-4}
                              & 0.1 kB                   & 22.1 kB                   & 23.7 kB                   & \\ \hline 
    \multirow{3}{*}{Sig 2}    & 0.2 kB                   & 19.8 kB                   & 22.6 kB                   & \multirow{3}{*}{Syndrome decoding over $\Ft$ (small weight)} \\ \cline{2-4}
                              & 0.2 kB                   & 17.8 kB                   & 20.6 kB                   & \\ \cline{2-4}
                              & 0.2 kB                   & 15.6 kB                   & 19.3 kB                   & \\ \hline
    \multirow{3}{*}{Sig 3}    & 0.2 kB                   & 13.7 kB                   & 19.3 kB                   & \multirow{3}{*}{Syndrome decoding over $\Ft$ (small weight)} \\ \cline{2-4}
                              & 0.2 kB                   & 12.0 kB                   & 17.0 kB                   & \\ \cline{2-4}
                              & 0.2 kB                   & 10.9 kB                   & 15.6 kB                   & \\ \hline	
  \end{tabular}
  \caption{Comparison between code-based signatures ($\lambda = 128$)} \label{table:param3}
}}
\end{center}
\end{table}

%\begin{center}
%\begin{table}
%  \caption{$\lambda = 128, n = 1238$ ($1306$ for AGS and QCS)}
%  \begin{tabular}{|c|c|c|}
%      \hline
%      \multicolumn{3}{|c|}{\textbf{PoK 1 (3-round Non-QC)}} \\
%      \multicolumn{3}{|c|}{26.2 kB - Cost x1.5 + Matrix Product} \\
%      \multicolumn{3}{|c|}{Most Conservative} \\ \hline
%      \textbf{QCS (5-round QC)} & \textbf{PoK 1 (5-round QC)} & \textbf{PoK 2 (5-round QC)} \\
%      23.5/26.6 kB - Cost x1 & 22.1 kB - Cost x1.5 & 18.2 kB - Cost x20 \\
%      Fastest     & Best Trade-off  & Smallest \\ \hline
%  \end{tabular}
%\end{table}
%\end{center}

%--------------------------------------------------------------------%
\section{Generalization and Variants} \label{sec:generalization}
%--------------------------------------------------------------------%

In this section, we briefly discuss the generalization of the proposed signatures to additional metrics as well several possible variants.

\vspace{\baselineskip}
\noindent \textbf{Generalization to additional metrics.} Following the work of \cite{BGS21}, we define a Full Domain Linear Isometry (FDLI) set as a set of linear isometries $I$ which has the property that given a random element $\phi \in I$, the image by $\phi$ of a random word $\bm{x}$ of weight $\omega$ is a random word $\bm{y}$ of weight $\omega$. 
One can adapt our protocols to other metrics by (i) redefining the weight of vectors and by (ii) replacing the permutation $\pi \in \hperm{n}$ by a random element $\phi$ of a FDLI set $I$.
For instance, our protocols can be adapted easily to the rank metric setting using the FDLI set for rank metric described in \cite{RankStern11} along with the rank weight and vectors over $\Fqmn$ in place of the Hamming weight and vectors over $\Ftn$.
This allows to further reduce the size of our signatures as illustrated in Table~\ref{table:rank-metric}.

\begin{table}[ht]
\begin{center}
{\renewcommand{\arraystretch}{1.3}
{\scriptsize
  \begin{tabular}{|l|c|c|c|c|c|c|c|c|c|c|}
    \cline{2-11}
      \multicolumn{1}{c|}{} & $q$ & $m$ & $n$ & $k$ & $w$ & $M$ & $N$ & $\tau$ & $\pk$ & $\sigma$ \\ \hline 
      Sig 1 (3-round)                  & \multirow{6}{*}{2} & \multirow{6}{*}{31} & \multirow{6}{*}{32} & \multirow{6}{*}{16} & \multirow{6}{*}{9} & \multirow{2}{*}{256} & \multirow{2}{*}{2} & \multirow{2}{*}{128} & 0.1 kB & 22.8 kB \\ \cline{1-1} \cline{10-11}
      Sig 1 (5-round)                  & & & &                       &                       &      &      &    & 0.1 kB                   & 19.7 kB \\ \cline{1-1} \cline{7-9} \cline{10-11} 
      \multirow{2}{*}{Sig 2 (3-round)} & & & & &                                             & 230  & 8   & 45  & \multirow{2}{*}{0.2 kB}  & 20.8 kB \\ \cline{7-9} \cline{11-11}
                                       & & & &                       &                       & 207  & 16   & 39 &                           & 19.1 kB \\ \cline{1-1} \cline{7-9} \cline{10-11}
      \multirow{2}{*}{Sig 3 (3-round)} & & & & &                                             & 230  & 8   & 45  & \multirow{2}{*}{0.2 kB}  & 17.4 kB \\ \cline{7-9} \cline{11-11}
                                       & & & &                       &                       & 389  & 32   & 28 &                           & 15.5 kB \\ \hline
  \end{tabular}
  \caption{Parameters and sizes in the rank metric setting ($\lambda = 128$)} \label{table:rank-metric}
}}
\end{center}
\end{table}

%\vspace{\baselineskip}
%\noindent \textbf{Variants from the $\SD$/$\GSD$ and $\QCSD$/$\QCGSD$ problems.} As the $\SD$ and $\GSD$ problems differs only by their representation, one can design code-based PoK relying on both problems as initially shown by the Stern and Véron protocols.
%As a consequence, it is possible to design a $\GSD$ variant of our PoK~1 as well as a $\SD$ variant of our PoK~2.
%In addition, one can also choose to use structured variants of the syndrome decoding problem to straightforwardly built variants of our proposed PoK.

\noindent \textbf{Variants with others structured matrices.} We have explained in Section~\ref{sec:optimizations} how one can use structured matrices in order to improve the performances of our proofs of knowledge.
To this end, we suggest the use of quasi-cyclic matrices as they are commonly used in code-based cryptography.
One should note that this performance improvement could be achieved with any set of matrices that benefit from an efficient matrix / vector product.
Therefore, our protocols can be adapted to work with other kind of structured matrices such as Toeplitz matrices which is generalization of quasi-cyclic matrices.

%--------------------------------------------------------------------%
\newpage
\bibliographystyle{alpha}
\bibliography{ref}
\appendix
%--------------------------------------------------------------------%

%--------------------------------------------------------------------%
\newpage
\section{PoK 1 (3-round, without optimization)} \label{app:pok1-3r-nopt}

\begin{figure}[!ht] 
  \begin{center}
    \resizebox{1\textwidth}{!}{\fbox{
      \pseudocode{%
        \hspace{160pt} \> \> \hspace{160pt} \\[-0.75\baselineskip] % Hack to force font scaling
        \underline{\proverf(\bm{x}, \bm{H}, \bm{y})} \> \> \underline{\verifierf(\bm{H}, \bm{y})} \\
        \pcfor k \in \intoneto{M} \pcdo \\
        \pcind \theta^{(k)} \sampler \bit^{\lambda} \\
        \pcind \pi^{(k)} \samples{\theta^{(k)}} \hperm{n}, ~ \bm{u}^{(k)} \samples{\theta^{(k)}} \Ftn \\ 
        \pcind r_0 \samples{\theta^{(k)}} \bit^{\lambda}, ~ r_1 \samples{\theta^{(k)}} \bit^{\lambda} \\
        \pcind \com^{(k)}_{0} = \commit{r^{(k)}_0, \, \pi^{(k)} \, || \, \bm{H} (\bm{u}^{(k)})^{\top}} \\
        \pcind \com^{(k)}_{1} = \commit{r^{(k)}_1, \, \pi^{(k)}[\bm{u}^{(k)}]} \\
        \pcind \aux^{(k)} = (\com^{(k)}_0, \com^{(k)}_1) \\[0.75\baselineskip]
        \pcind r^{(k)} \sampler \bit^{\lambda} \\
        \pcind \com^{(k)} = \commit{r^{(k)}, \, \pi^{(k)}[\bm{u}^{(k)} + \bm{x}]} \\
        \pcend \\
        \> \sendmessageright*{(\aux^{(k)}, \com^{(k)})_{k \in \intoneto{M}}} \\[-0.5\baselineskip]
        \> \> \kappa \sampler \intoneto{M} \\ 
        \> \> \alpha \sampler \bit \\[-\baselineskip]
        \> \sendmessageleft*{(\kappa, \alpha)} \\[-0.5\baselineskip]
        \pcif \alpha = 0 \pcthen \\
        \pcind z_1 = \pi^{(\kappa)}, ~ \bm{z}_2 = \bm{u}^{(\kappa)} + \bm{x} \\
        \pcind \rsp = (r_0^{(\kappa)}, \, r^{(\kappa)}, \, z_1, \, \bm{z}_2) \\
        \pcend \\[0.5\baselineskip]
        \pcif \alpha = 1 \pcthen \\
        \pcind \bm{z}_3 = \pi^{(\kappa)}[\bm{u}^{(\kappa)}], ~ \bm{z}_4 = \pi^{(\kappa)}[\bm{x}] \\ 
        \pcind \rsp = (r_1^{(\kappa)}, \, r^{(\kappa)}, \, \bm{z}_3, \, \bm{z}_4) \\
        \pcend \\
        \> \sendmessageright*{(\theta^{(k)})_{k \in \intoneto{M} \setminus \kappa}, \, \rsp} \\[-\baselineskip]
        \> \> \pcfor k \in \intoneto{M} \setminus \kappa \pcdo \\
        \> \> \pcind \tx{Compute } \bar{\aux}^{(k)} \tx{ from } \theta^{(k)} \\
        \> \> \pcind \pcif (\aux^{(k)} \neq \bar{\aux}^{(k)}) ~ \pcreturn 0 \\
        \> \> \pcend \\[0.75\baselineskip]
        \> \> \pcif \alpha = 0 \pcthen \\
        \> \> \pcind b_1 \samplen \big( \com^{(\kappa)}_0 = \commit{r_0^{(\kappa)}, \, z_1 \, || \, \bm{H} \bm{z}_2^{\top} - \bm{y}^{\top}} \big) \\ 
        \> \> \pcind b_2 \samplen \big( \com^{(\kappa)} = \commit{r^{(\kappa)}, \, z_1[\bm{z}_2]} \big) \\ 
        \> \> \pcind \pcreturn b_1 \wedge b_2 \\
        \> \> \pcend \\[0.5\baselineskip]
        \> \> \pcif \alpha = 1 \pcthen \\
        \> \> \pcind b_1 \samplen \big( \com^{(\kappa)}_1 = \commit{r_1^{(\kappa)}, \, \bm{z}_3} \big) \\ 
        \> \> \pcind b_2 \samplen \big( \com^{(\kappa)} = \commit{r^{(\kappa)}, \, \bm{z}_3 + \bm{z}_4} \big) \\ 
        \> \> \pcind b_3 \samplen \big( \hw{\bm{z}_4} = \omega \big) \\
        \> \> \pcind \pcreturn b_1 \wedge b_2 \wedge b_3 \\
        \> \> \pcend
      }
    }}
    \caption{3-round HVZK PoK for the $\SD$ problem (without optimization) \label{fig:pok1-3r-nopt}}
  \end{center}
\end{figure}

%--------------------------------------------------------------------%

%--------------------------------------------------------------------%
\newpage
\section{PoK 1 (3-round, with optimizations)} \label{app:pok1-3r-opt}

\begin{figure}[!ht] 
  \begin{center}
    \resizebox{1\textwidth}{!}{\fbox{
      \pseudocode{%
        \hspace{160pt} \> \> \hspace{160pt} \\[-0.75\baselineskip] % Hack to force font scaling
        \underline{\proverf(\bm{x}, \bm{H}, \bm{y})} \> \> \underline{\verifierf(\bm{H}, \bm{y})} \\
        \xi \sampler \bit^{\lambda} \\ 
        \pcfor k \in \intoneto{M} \pcdo \\
        \pcind \theta^{(k)} \sampler \bit^{\lambda}, ~ \phi^{(k)} \samples{\theta^{(k)}} \bit^{\lambda}, ~ \psi^{(k)} \samples{\theta^{(k)}} \bit^{\lambda} \\
        \pcind \pi^{(k)} \samples{\phi^{(k)}} \hperm{n}, ~ \bm{v}^{(k)} \samples{\psi^{(k)}} \Ftn, ~ \bm{u}^{(k)} = (\pi^{(k)})^{-1}[\bm{v}^{(k)}] \\
        \pcind r^{(k)}_0 \samples{\phi^{(k)}} \bit^{\lambda}, ~ \com^{(k)}_{0} = \commit{r^{(k)}_0, \, \pi^{(k)} \, || \, \bm{H} (\bm{u}^{(k)})^{\top}} \\
        \pcind r^{(k)}_1 \samples{\psi^{(k)}} \bit^{\lambda}, ~ \com^{(k)}_{1} = \commit{r^{(k)}_1, \, \pi^{(k)}[\bm{u}^{(k)}]} \\
        \pcind \aux^{(k)} = (\com^{(k)}_0, \com^{(k)}_1) \\
        \pcind r^{(k)} \samples{\xi} \bit^{\lambda}, ~ \com^{(k)} = \commit{r^{(k)}, \, \pi^{(k)}[\bm{u}^{(k)} + \bm{x}]} \\
        \pcend \\
        r \samples{\xi} \bit{^\lambda}, ~ h = \commit{r, \, (\aux^{(k)} \, || \, \com^{(k)})_{k \in \intoneto{M}}} \\[-0.75\baselineskip]
        \> \sendmessageright*{h} \\[-0.5\baselineskip]
        \> \> K \sampler \{ K \subset \{1, \cdots, M \}, \, |K| = \tau \} \\ 
        \> \> A \sampler \bit^{\tau} \\[-\baselineskip]
        \> \sendmessageleft*{(K, A)} \\[-0.5\baselineskip]
        \pcfor (\kappa, \alpha) \in K \times A \pcdo \\
        \pcind \pcif \alpha = 0 \pcthen \\
        \pcind \pcind z^{(\kappa)}_1 = \phi^{(\kappa)}, ~ \bm{z}^{(\kappa)}_2 = \bm{u}^{(\kappa)} + \bm{x} \\
        \pcind \pcind \rsp^{(\kappa)} = (z^{(\kappa)}_1, \, \bm{z}^{(\kappa)}_2, \, \com^{(\kappa)}_1) \\
        \pcind \pcend \\[0.5\baselineskip]
        \pcind \pcif \alpha = 1 \pcthen \\
        \pcind \pcind z^{(\kappa)}_3 = \psi^{(\kappa)}, ~ \bm{z}^{(\kappa)}_4 = \pi^{(\kappa)}[\bm{x}] \\ 
        \pcind \pcind \rsp^{(\kappa)} = (z^{(\kappa)}_3, \, \bm{z}^{(\kappa)}_4, \, \com^{(\kappa)}_0) \\
        \pcind \pcend \\
        \pcend \\
        \rsp = (\xi, (\theta^{(k)}, \com^{(k)})_{k \in \intoneto{M} \setminus K}, \, (\rsp^{(\kappa)})_{\kappa \in K})
        \> \sendmessageright*{\rsp} \\[-\baselineskip]
        \> \> \tx{Compute } \bar{r}, (\bar{r}^{(\kappa)})_{\kappa \in K} \tx{ from } \xi \\
        \> \> \pcfor k \in \intoneto{M} \setminus K \pcdo \\
        \> \> \pcind \tx{Compute } \bar{\aux}^{(k)} \tx{ from } \theta^{(k)} \\
        \> \> \pcend \\[0.5\baselineskip]
        \> \> \pcfor (\kappa, \alpha) \in K \times A \pcdo \\
        \> \> \pcind \pcif \alpha = 0 \pcthen \\
        \> \> \pcind \pcind \bar{z}^{(\kappa)}_1 \samples{z^{(\kappa)}_1} \hperm{n}, ~ \bar{r}^{(\kappa)}_0 \samples{z^{(\kappa)}_1} \bit^{\lambda} \\
        \> \> \pcind \pcind \com^{(\kappa)}_0 = \commit{\bar{r}_0^{(\kappa)}, \, \bar{z}^{(\kappa)}_1 \, || \, \bm{H} (\bm{z}^{(\kappa)}_2)^{\top} - \bm{y}^{\top}} \\ 
        \> \> \pcind \pcind \com^{(\kappa)} = \commit{\bar{r}^{(\kappa)}, \, \bar{z}^{(\kappa)}_1[\bm{z}^{(\kappa)}_2]} \\ 
        \> \> \pcind \pcend \\[0.5\baselineskip]
        \> \> \pcind \pcif \alpha = 1 \pcthen \\
        \> \> \pcind \pcind \bar{\bm{z}}^{(\kappa)}_3 \samples{z^{(\kappa)}_3} \Ftn, ~ \bar{r}^{(\kappa)}_1 \samples{z^{(\kappa)}_3} \bit^{\lambda} \\
        \> \> \pcind \pcind \com^{(\kappa)}_1 = \commit{\bar{r}_1^{(\kappa)}, \, \bar{\bm{z}}_3} \\ 
        \> \> \pcind \pcind \com^{(\kappa)} = \commit{\bar{r}^{(\kappa)}, \, \bar{\bm{z}}_3 + \bm{z}_4} \\ 
        \> \> \pcind \pcind b^{(\kappa)}_1 \samplen \big( \hw{\bm{z}^{(\kappa)}_4} = \omega \big) \\
        \> \> \pcind \pcend \\
        \> \> \pcend \\[0.5\baselineskip]
        \> \> b_1 = \bigwedge\nolimits_{\kappa \in K} b^{(\kappa)}_1 \\ 
        \> \> b_2 \samplen \big( h = \commit{r, \, (\aux^{(k)} \, || \, \com^{(k)})_{k \in \intoneto{M}}} \big) \\
        \> \> \pcreturn b_1 \wedge b_2
      }
    }}
    \caption{3-round HVZK PoK for the $\SD$ problem (with optimizations) \label{fig:pok1-3r-opt}}
  \end{center}
\end{figure}

%--------------------------------------------------------------------%

%--------------------------------------------------------------------%
\newpage
\section{Sig 1 (3-round)} \label{app:sig1-3r}
\begin{figure}[!ht] 
  \begin{center}
    \resizebox{1\textwidth}{!}{\fbox{
      \pseudocode{%
        \hspace{160pt} \> \> \hspace{160pt} \\[-0.75\baselineskip] % Hack to force font scaling
        \underline{\keygen(\lambda)} \\
        \rho_1 \sampler \bit^{\lambda}, ~ \bm{x} \samples{\rho_1} \Ftn \tx{ such that } \hw{\bm{x}} = \omega \\
        \rho_2 \sampler \bit^{\lambda}, ~ \bm{H} \samples{\rho_2} \Ft^{\nmktn}, ~ \bm{y}^{\top} = \bm{H} \bm{x}^{\top} \\
        \pcreturn (\sk, \pk) = (\rho_1, (\rho_2, \bm{y})) \\[0.75\baselineskip]
        \underline{\sign(\sk, \pk, m)} \\
        \xi \sampler \bit^{\lambda} \\ 
        \pcfor k \in \intoneto{M} \pcdo \\
        \pcind \theta^{(k)} \sampler \bit^{\lambda}, ~ \phi^{(k)} \samples{\theta^{(k)}} \bit^{\lambda}, ~ \psi^{(k)} \samples{\theta^{(k)}} \bit^{\lambda} \\
        \pcind \pi^{(k)} \samples{\phi^{(k)}} \hperm{n}, ~ \bm{v}^{(k)} \samples{\psi^{(k)}} \Ftn, ~ \bm{u}^{(k)} = (\pi^{(k)})^{-1}[\bm{v}^{(k)}] \\
        \pcind r^{(k)}_0 \samples{\phi^{(k)}} \bit^{\lambda}, ~ \com^{(k)}_{0} = \commit{r^{(k)}_0, \, \pi^{(k)} \, || \, \bm{H} (\bm{u}^{(k)})^{\top}} \\
        \pcind r^{(k)}_1 \samples{\psi^{(k)}} \bit^{\lambda}, ~ \com^{(k)}_{1} = \commit{r^{(k)}_1, \, \pi^{(k)}[\bm{u}^{(k)}]} \\
        \pcind \aux^{(k)} = (\com^{(k)}_0, \com^{(k)}_1) \\
        \pcind r^{(k)} \samples{\xi} \bit^{\lambda}, ~ \com^{(k)} = \commit{r^{(k)}, \, \pi^{(k)}[\bm{u}^{(k)} + \bm{x}]} \\
        \pcend \\
        r \samples{\xi} \bit{^\lambda}, ~ h = \commit{r, \, (\aux^{(k)} \, || \, \com^{(k)})_{k \in \intoneto{M}}} \\
        (K, \, A) \samplen \ro{m \, || \, \pk \, || \, h} \\
        \pcfor (\kappa, \alpha) \in K \times A \pcdo \\
        \pcind \pcif \alpha = 0 \pcthen \\
        \pcind \pcind z^{(\kappa)}_1 = \phi^{(\kappa)}, ~ \bm{z}^{(\kappa)}_2 = \bm{u}^{(\kappa)} + \bm{x} \\
        \pcind \pcind \rsp^{(\kappa)} = (z^{(\kappa)}_1, \, \bm{z}^{(\kappa)}_2, \, \com^{(\kappa)}_1) \\
        \pcind \pcend \\[0.5\baselineskip]
        \pcind \pcif \alpha = 1 \pcthen \\
        \pcind \pcind z^{(\kappa)}_3 = \psi^{(\kappa)}, ~ \bm{z}^{(\kappa)}_4 = \pi^{(\kappa)}[\bm{x}] \\ 
        \pcind \pcind \rsp^{(\kappa)} = (z^{(\kappa)}_3, \, \bm{z}^{(\kappa)}_4, \, \com^{(\kappa)}_0) \\
        \pcind \pcend \\
        \pcend \\
        \rsp = (\xi, (\theta^{(k)}, \com^{(k)})_{k \in \intoneto{M} \setminus K}, \, (\rsp^{(\kappa)})_{\kappa \in K}) \\
        \pcreturn \sigma = (h, \, \rsp)
        }
    }}
    \caption{$\keygen$ and $\sign$ algorithms for Sig 1 (3-round)} \label{fig:signature-1-3r-keygen-sign}
  \end{center}
\end{figure}

\newpage
        
\begin{figure}[!ht] 
  \begin{center}
    \resizebox{1\textwidth}{!}{\fbox{
      \pseudocode{%
        \hspace{160pt} \> \> \hspace{60pt} \\[-0.75\baselineskip] % Hack to force font scaling
        \underline{\verif(\pk, \sigma, m)} \\
        \tx{Parse } \sigma \tx{ as } \sigma = (h, \rsp) \tx{ and } \rsp \tx{ as } \rsp = (\xi, (\theta^{(k)}, \com^{(k)})_{k \in \intoneto{M} \setminus K}, \, (\rsp^{(\kappa)})_{\kappa \in K}) \\
        (K, \, A) \samplen \ro{m \, || \, \pk \, || \, h} \\[0.5\baselineskip]
        \tx{Compute } \bar{r}, (\bar{r}^{(\kappa)})_{\kappa \in K} \tx{ from } \xi \\
        \pcfor k \in \intoneto{M} \setminus K \pcdo \\
        \pcind \tx{Compute } \bar{\aux}^{(k)} \tx{ from } \theta^{(k)} \\
        \pcend \\[0.5\baselineskip]
        \pcfor (\kappa, \alpha) \in K \times A \pcdo \\
        \pcind \pcif \alpha = 0 \pcthen \\
        \pcind \pcind \bar{z}^{(\kappa)}_1 \samples{z^{(\kappa)}_1} \hperm{n}, ~ \bar{r}^{(\kappa)}_0 \samples{z^{(\kappa)}_1} \bit^{\lambda} \\
        \pcind \pcind \com^{(\kappa)}_0 = \commit{\bar{r}_0^{(\kappa)}, \, \bar{z}^{(\kappa)}_1 \, || \, \bm{H} (\bm{z}^{(\kappa)}_2)^{\top} - \bm{y}^{\top}} \\ 
        \pcind \pcind \com^{(\kappa)} = \commit{\bar{r}^{(\kappa)}, \, \bar{z}^{(\kappa)}_1[\bm{z}^{(\kappa)}_2]} \\ 
        \pcind \pcend \\[0.5\baselineskip]
        \pcind \pcif \alpha = 1 \pcthen \\
        \pcind \pcind \bar{\bm{z}}^{(\kappa)}_3 \samples{z^{(\kappa)}_3} \Ftn, ~ \bar{r}^{(\kappa)}_1 \samples{z^{(\kappa)}_3} \bit^{\lambda} \\
        \pcind \pcind \com^{(\kappa)}_1 = \commit{\bar{r}_1^{(\kappa)}, \, \bar{\bm{z}}_3} \\ 
        \pcind \pcind \com^{(\kappa)} = \commit{\bar{r}^{(\kappa)}, \, \bar{\bm{z}}_3 + \bm{z}_4} \\ 
        \pcind \pcind b^{(\kappa)}_1 \samplen \big( \hw{\bm{z}^{(\kappa)}_4} = \omega \big) \\
        \pcind \pcend \\
        \pcend \\[0.5\baselineskip]
        b_1 = \bigwedge\nolimits_{\kappa \in K} b^{(\kappa)}_1 \\ 
        b_2 \samplen \big( h = \commit{r, \, (\aux^{(k)} \, || \, \com^{(k)})_{k \in \intoneto{M}}} \big) \\
        \pcreturn b_1 \wedge b_2
      }
    }}
    \caption{$\verif$ algorithm for Sig 1 (3-round)} \label{fig:signature-pok1-3r-verify}
  \end{center}
\end{figure}

%--------------------------------------------------------------------%

%--------------------------------------------------------------------%
\newpage
\section{PoK 1 (5-round, without optimization)} \label{app:pok1-5r-nopt}

\vspace{-0.5\baselineskip}
\begin{figure}[!ht] 
  \begin{center}
    \resizebox{1\textwidth}{!}{\fbox{
      \pseudocode{%
        \hspace{160pt} \> \> \hspace{160pt} \\[-0.75\baselineskip] % Hack to force font scaling
        \underline{\proverf(\bm{x}, \bm{H}, \bm{y})} \> \> \underline{\verifierf(\bm{H}, \bm{y})} \\
        \pcfor k \in \intoneto{M} \pcdo \\
        \pcind \theta^{(k)} \sampler \bit^{\lambda} \\
        \pcind \pi^{(k)} \samples{\theta^{(k)}} \hperm{n}, ~ \bm{u}^{(k)} \samples{\theta^{(k)}} \Ftn \\ 
        \pcind r_0 \samples{\theta^{(k)}} \bit^{\lambda}, ~ r_1 \samples{\theta^{(k)}} \bit^{\lambda} \\
        \pcind \com^{(k)}_{0} = \commit{r^{(k)}_0, \, \pi^{(k)} \, || \, \bm{H} (\bm{u}^{(k)})^{\top}} \\
        \pcind \com^{(k)}_{1} = \commit{r^{(k)}_1, \, \pi^{(k)}[\bm{u}^{(k)}]} \\
        \pcind \aux^{(k)} = (\com^{(k)}_0, \com^{(k)}_1) \\
        \pcend \\
        \> \sendmessageright*{(\aux^{(k)})_{k \in \intoneto{M}}} \\[-0.5\baselineskip]
        \> \> \kappa \sampler \intoneto{M} \\[-\baselineskip]
        \> \sendmessageleft*{\kappa} \\[-0.5\baselineskip]
        r^{(\kappa)} \sampler \bit^{\lambda} \\
        \com^{(\kappa)} = \commit{r^{(\kappa)}, \, \pi^{(\kappa)}[\bm{u}^{(\kappa)} + \bm{x}]} \\
        \> \sendmessageright*{(\theta^{(k)})_{k \in \intoneto{M} \setminus \kappa}, \, \com^{(\kappa)}} \\[-0.5\baselineskip]
        \> \> \pcfor k \in \intoneto{M} \setminus \kappa \pcdo \\
        \> \> \pcind \tx{Compute } \bar{\aux}^{(k)} \tx{ from } \theta^{(k)} \\
        \> \> \pcind \pcif (\aux^{(k)} \neq \bar{\aux}^{(k)}) ~ \pcreturn 0 \\
        \> \> \pcend \\[0.5\baselineskip]
        \> \> \alpha \sampler \bit \\[-\baselineskip]
        \> \sendmessageleft*{\alpha} \\[-0.5\baselineskip]
        \pcif \alpha = 0 \pcthen \\
        \pcind z_1 = \pi^{(\kappa)}, ~ \bm{z}_2 = \bm{u}^{(\kappa)} + \bm{x} \\
        \pcind \rsp = (r_0^{(\kappa)}, \, r^{(\kappa)}, \, z_1, \, \bm{z}_2) \\
        \pcend \\[0.5\baselineskip]
        \pcif \alpha = 1 \pcthen \\
        \pcind \bm{z}_3 = \pi^{(\kappa)}[\bm{u}^{(\kappa)}], ~ \bm{z}_4 = \pi^{(\kappa)}[\bm{x}] \\ 
        \pcind \rsp = (r_1^{(\kappa)}, \, r^{(\kappa)}, \, \bm{z}_3, \, \bm{z}_4) \\
        \pcend \\
        \> \sendmessageright*{\rsp} \\[-0.5\baselineskip]
        \> \> \pcif \alpha = 0 \pcthen \\
        \> \> \pcind b_1 \samplen \big( \com^{(\kappa)}_0 = \commit{r_0^{(\kappa)}, \, z_1 \, || \, \bm{H} \bm{z}_2^{\top} - \bm{y}^{\top}} \big) \\ 
        \> \> \pcind b_2 \samplen \big( \com^{(\kappa)} = \commit{r^{(\kappa)}, \, z_1[\bm{z}_2]} \big) \\ 
        \> \> \pcind \pcreturn b_1 \wedge b_2 \\
        \> \> \pcend \\[0.5\baselineskip]
        \> \> \pcif \alpha = 1 \pcthen \\
        \> \> \pcind b_1 \samplen \big( \com^{(\kappa)}_1 = \commit{r_1^{(\kappa)}, \, \bm{z}_3} \big) \\ 
        \> \> \pcind b_2 \samplen \big( \com^{(\kappa)} = \commit{r^{(\kappa)}, \, \bm{z}_3 + \bm{z}_4} \big) \\ 
        \> \> \pcind b_3 \samplen \big( \hw{\bm{z}_4} = \omega \big) \\
        \> \> \pcind \pcreturn b_1 \wedge b_2 \wedge b_3 \\
        \> \> \pcend
      }
    }}
    \caption{5-round HVZK PoK for the $\SD$ problem (without optimization) \label{fig:pok1-5r-nopt}}
  \end{center}
\end{figure}
\vspace{-0.5\baselineskip}

%--------------------------------------------------------------------%

%--------------------------------------------------------------------%
\newpage
\section{PoK 1 (5-round, with optimizations)} \label{app:pok1-5r-opt}

\vspace{-\baselineskip}
\begin{figure}[!ht] 
  \begin{center}
    \resizebox{1\textwidth}{!}{\fbox{
      \pseudocode{%
        \hspace{160pt} \> \> \hspace{160pt} \\[-0.75\baselineskip] % Hack to force font scaling
        %
        % Prover 1
        \underline{\proverf(\bm{x}, \bm{H}, \bm{y})} \> \> \underline{\verifierf(\bm{H}, \bm{y})} \\
        \xi \sampler \bit^{\lambda} \\
        \pcfor k \in \intoneto{M} \pcdo \\
        \pcind \theta^{(k)} \sampler \bit^{\lambda}, ~ \phi^{(k)} \samples{\theta^{(k)}} \bit^{\lambda}, ~ \psi^{(k)} \samples{\theta^{(k)}} \bit^{\lambda} \\
        \pcind \pi^{(k)} \samples{\phi^{(k)}} \hperm{n}, ~ \bm{v}^{(k)} \samples{\psi^{(k)}} \Ftn, ~ \bm{u}^{(k)} = (\pi^{(k)})^{-1}[\bm{v}^{(k)}] \\
        \pcind r^{(k)}_0 \samples{\phi^{(k)}} \bit^{\lambda}, ~ \com^{(k)}_{0} = \commit{r^{(k)}_0, \, \pi^{(k)} \, || \, \bm{H} (\bm{u}^{(k)})^{\top}} \\
        \pcind r^{(k)}_1 \samples{\psi^{(k)}} \bit^{\lambda}, ~ \com^{(k)}_{1} = \commit{r^{(k)}_1, \, \pi^{(k)}[\bm{u}^{(k)}]} \\
        \pcind \aux^{(k)} = (\com^{(k)}_0, \com^{(k)}_1) \\
        \pcend \\
        r \samples{\xi} \bit{^\lambda}, ~ h = \commit{r, \, (\aux^{(k)})_{k \in \intoneto{M}}} \\[-0.75\baselineskip]
        % Prover 1
        %
        % Exchange 1
        \> \sendmessageright*{h} \\[-0.5\baselineskip]
        % Exchange 1
        %
        % Verifier 1
        \> \> K \sampler \{ K \subset \{ 1, \cdots, M \}, \, |K| = \tau \} \\[-\baselineskip]
        % Verifier 1
        %
        % Exchange 2
        \> \sendmessageleft*{K} \\[-0.75\baselineskip]
        % Exchange 2
        %
        % Prover 2
        \pcfor \kappa \in K \pcdo \\
        \pcind r^{(\kappa)} \samples{\xi} \bit^{\lambda}, ~ \com^{(\kappa)} = \commit{r^{(\kappa)}, \, \pi^{(\kappa)}[\bm{u}^{(\kappa)} + \bm{x}]} \\
        \pcend \\
        r' \samples{\xi} \bit^{\lambda}, ~ h' = \commit{r', \, (\com^{(\kappa)})_{\kappa \in K}} \\[-0.75\baselineskip]
        % Prover 2
        %
        % Exchange 3
        \> \sendmessageright*{h'} \\[-0.5\baselineskip]
        % Exchange 3
        %
        % Verifier 2
        \> \> A \sampler \bit^{\tau} \\[-\baselineskip] 
        % Verifier 2
        %
        % Exchange 4
        \> \sendmessageleft*{A} \\[-0.75\baselineskip]
        % Exchange 4
        %
        % Prover 3
        \pcfor (\kappa, \alpha) \in K \times A \pcdo \\
        \pcind \pcif \alpha = 0 \pcthen \\
        \pcind \pcind z^{(\kappa)}_1 = \phi^{(\kappa)}, ~ \bm{z}^{(\kappa)}_2 = \bm{u}^{(\kappa)} + \bm{x}, ~ \rsp^{(\kappa)} = (z^{(\kappa)}_1, \, \bm{z}^{(\kappa)}_2, \, \com^{(\kappa)}_1) \\
        \pcind \pcend \\[0.5\baselineskip]
        \pcind \pcif \alpha = 1 \pcthen \\
        \pcind \pcind z^{(\kappa)}_3 = \psi^{(\kappa)}, ~ \bm{z}^{(\kappa)}_4 = \pi^{(\kappa)}[\bm{x}], ~ \rsp^{(\kappa)} = (z^{(\kappa)}_3, \, \bm{z}^{(\kappa)}_4, \, \com^{(\kappa)}_0) \\
        \pcind \pcend \\
        \pcend \\
        \rsp = \big( \xi, (\theta^{(k)})_{k \in \intoneto{M} \setminus K}, ~ (\rsp^{(\kappa)})_{\kappa \in K} \big) \\[-0.75\baselineskip]
        %
        % Prover 3
        %
        % Exchange 5
        \> \sendmessageright*{\rsp} \\[-1.25\baselineskip]
        % Exchange 5
        %
        % Verifier 3
        \> \> \tx{Compute } \bar{r}, (\bar{r}^{(\kappa)})_{\kappa \in K} \tx{ and } \bar{r'} \tx{ from } \xi \\
        \> \> \pcfor k \in \intoneto{M} \setminus K \pcdo \\
        \> \> \pcind \tx{Compute } \bar{\aux}^{(k)} \tx{ from } \theta^{(k)} \\
        \> \> \pcend \\[0.5\baselineskip]
        \> \> \pcfor (\kappa, \alpha) \in K \times A \pcdo \\
        \> \> \pcind \pcif \alpha = 0 \pcthen \\
        \> \> \pcind \pcind \bar{z}^{(\kappa)}_1 \samples{z^{(\kappa)}_1} \hperm{n}, ~ \bar{r}^{(\kappa)}_0 \samples{z^{(\kappa)}_1} \bit^{\lambda} \\
        \> \> \pcind \pcind \com^{(\kappa)}_0 = \commit{\bar{r}_0^{(\kappa)}, \, \bar{z}^{(\kappa)}_1 \, || \, \bm{H} (\bm{z}^{(\kappa)}_2)^{\top} - \bm{y}^{\top}} \\ 
        \> \> \pcind \pcind \com^{(\kappa)} = \commit{\bar{r}^{(\kappa)}, \, \bar{z}^{(\kappa)}_1[\bm{z}^{(\kappa)}_2]} \\ 
        \> \> \pcind \pcend \\[0.5\baselineskip]
        \> \> \pcind \pcif \alpha = 1 \pcthen \\
        \> \> \pcind \pcind \bar{\bm{z}}^{(\kappa)}_3 \samples{z^{(\kappa)}_3} \Ftn, ~ \bar{r}^{(\kappa)}_1 \samples{z^{(\kappa)}_3} \bit^{\lambda} \\
        \> \> \pcind \pcind \com^{(\kappa)}_1 = \commit{\bar{r}_1^{(\kappa)}, \, \bar{\bm{z}}_3} \\ 
        \> \> \pcind \pcind \com^{(\kappa)} = \commit{\bar{r}^{(\kappa)}, \, \bar{\bm{z}}_3 + \bm{z}_4} \\ 
        \> \> \pcind \pcind b^{(\kappa)}_1 \samplen \big( \hw{\bm{z}^{(\kappa)}_4} = \omega \big) \\
        \> \> \pcind \pcend \\
        \> \> \pcend \\[0.5\baselineskip]
        \> \> b_1 = \bigwedge\nolimits_{\kappa \in K} b^{(\kappa)}_1 \\ 
        \> \> b_2 \samplen \big( h = \commit{\bar{r}, \, (\aux^{(k)})_{k \in \intoneto{M}} } \big)\\
        \> \> b_3 \samplen \big( h' = \commit{\bar{r'}, \, (\com^{(\kappa)})_{\kappa \in K} } \big)\\
        \> \> \pcreturn b_1 \wedge b_2 \wedge b_3
        % Verifier 3
      }
    }}
    \caption{5-round HVZK PoK for the $\SD$ problem (with optimizations) \label{fig:pok1-5r-opt}}
  \end{center}
\end{figure}
\vspace{-\baselineskip}

%--------------------------------------------------------------------%

%--------------------------------------------------------------------%
\newpage
\section{Sig 1 (5-round)} \label{app:sig1-5r}
\begin{figure}[!ht] 
  \begin{center}
    \resizebox{1\textwidth}{!}{\fbox{
      \pseudocode{%
        \hspace{160pt} \> \> \hspace{160pt} \\[-0.75\baselineskip] % Hack to force font scaling
        \underline{\keygen(\lambda)} \\
        \rho_1 \sampler \bit^{\lambda}, ~ \bm{x} \samples{\rho_1} \Ftn \tx{ such that } \hw{\bm{x}} = \omega \\
        \rho_2 \sampler \bit^{\lambda}, ~ \bm{H} \samples{\rho_2} \Ft^{\nmktn}, ~ \bm{y}^{\top} = \bm{H} \bm{x}^{\top} \\
        \pcreturn (\sk, \pk) = (\rho_1, (\rho_2, \bm{y})) \\[0.75\baselineskip]
        \underline{\sign(\sk, \pk, m)} \\
        \xi \sampler \bit^{\lambda} \\
        \pcfor k \in \intoneto{M} \pcdo \\
        \pcind \theta^{(k)} \sampler \bit^{\lambda}, ~ \phi^{(k)} \samples{\theta^{(k)}} \bit^{\lambda}, ~ \psi^{(k)} \samples{\theta^{(k)}} \bit^{\lambda} \\
        \pcind \pi^{(k)} \samples{\phi^{(k)}} \hperm{n}, ~ \bm{v}^{(k)} \samples{\psi^{(k)}} \Ftn, ~ \bm{u}^{(k)} = (\pi^{(k)})^{-1}[\bm{v}^{(k)}] \\
        \pcind r^{(k)}_0 \samples{\phi^{(k)}} \bit^{\lambda}, ~ \com^{(k)}_{0} = \commit{r^{(k)}_0, \, \pi^{(k)} \, || \, \bm{H} (\bm{u}^{(k)})^{\top}} \\
        \pcind r^{(k)}_1 \samples{\psi^{(k)}} \bit^{\lambda}, ~ \com^{(k)}_{1} = \commit{r^{(k)}_1, \, \pi^{(k)}[\bm{u}^{(k)}]} \\
        \pcind \aux^{(k)} = (\com^{(k)}_0, \com^{(k)}_1) \\
        \pcend \\
        r \samples{\xi} \bit{^\lambda}, ~ h = \commit{r, \, (\aux^{(k)})_{k \in \intoneto{M}}} \\
        K \samplen \ro{m \, || \, \pk \, || \, h} \\
        \pcfor \kappa \in K \pcdo \\
        \pcind r^{(\kappa)} \samples{\xi} \bit^{\lambda}, ~ \com^{(\kappa)} = \commit{r^{(\kappa)}, \, \pi^{(\kappa)}[\bm{u}^{(\kappa)} + \bm{x}]} \\
        \pcend \\
        r' \samples{\xi} \bit^{\lambda}, ~ h' = \commit{r', \, (\com^{(\kappa)})_{\kappa \in K}} \\
        A \samplen \ro{m \, || \, \pk \, || \, h \, || \, h'} \\
        \pcfor (\kappa, \alpha) \in K \times A \pcdo \\
        \pcind \pcif \alpha = 0 \pcthen \\
        \pcind \pcind z^{(\kappa)}_1 = \phi^{(\kappa)}, ~ \bm{z}^{(\kappa)}_2 = \bm{u}^{(\kappa)} + \bm{x}, ~ \rsp^{(\kappa)} = (z^{(\kappa)}_1, \, \bm{z}^{(\kappa)}_2, \, \com^{(\kappa)}_1) \\
        \pcind \pcend \\[0.5\baselineskip]
        \pcind \pcif \alpha = 1 \pcthen \\
        \pcind \pcind z^{(\kappa)}_3 = \psi^{(\kappa)}, ~ \bm{z}^{(\kappa)}_4 = \pi^{(\kappa)}[\bm{x}], ~ \rsp^{(\kappa)} = (z^{(\kappa)}_3, \, \bm{z}^{(\kappa)}_4, \, \com^{(\kappa)}_0) \\
        \pcind \pcend \\
        \pcend \\
        \rsp = \big( \xi, (\theta^{(k)})_{k \in \intoneto{M} \setminus K}, ~ (\rsp^{(\kappa)})_{\kappa \in K} \big) \\
        \pcreturn \sigma = (h, \, h', \, \rsp)
        }
    }}
    \caption{$\keygen$ and $\sign$ algorithms for Sig 1 (5-round)} \label{fig:signature-1-5r-keygen-sign}
  \end{center}
\end{figure}

\newpage

\begin{figure}[!ht] 
  \begin{center}
    \resizebox{1\textwidth}{!}{\fbox{
      \pseudocode{%
        \hspace{160pt} \> \> \hspace{60pt} \\[-0.75\baselineskip] % Hack to force font scaling
        \underline{\verif(\pk, \sigma, m)} \\
        \tx{Parse } \sigma \tx{ as } \sigma = (h, h', \rsp) \tx{ and } \rsp \tx{ as } \rsp = \big( \xi, (\theta^{(k)})_{k \in \intoneto{M} \setminus K}, ~ (\rsp^{(\kappa)})_{\kappa \in K} \big) \\
        K \samplen \ro{m \, || \, \pk \, || \, h} \\
        A \samplen \ro{m \, || \, \pk \, || \, h \, || \, h'} \\[0.5\baselineskip]
        \tx{Compute } \bar{r}, (\bar{r}^{(\kappa)})_{\kappa \in K} \tx{ and } \bar{r'} \tx{ from } \xi \\
        \pcfor k \in \intoneto{M} \setminus K \pcdo \\
        \pcind \tx{Compute } \bar{\aux}^{(k)} \tx{ from } \theta^{(k)} \\
        \pcend \\[0.5\baselineskip]
        \pcfor (\kappa, \alpha) \in K \times A \pcdo \\
        \pcind \pcif \alpha = 0 \pcthen \\
        \pcind \pcind \bar{z}^{(\kappa)}_1 \samples{z^{(\kappa)}_1} \hperm{n}, ~ \bar{r}^{(\kappa)}_0 \samples{z^{(\kappa)}_1} \bit^{\lambda} \\
        \pcind \pcind \com^{(\kappa)}_0 = \commit{\bar{r}_0^{(\kappa)}, \, \bar{z}^{(\kappa)}_1 \, || \, \bm{H} (\bm{z}^{(\kappa)}_2)^{\top} - \bm{y}^{\top}} \\ 
        \pcind \pcind \com^{(\kappa)} = \commit{\bar{r}^{(\kappa)}, \, \bar{z}^{(\kappa)}_1[\bm{z}^{(\kappa)}_2]} \\ 
        \pcind \pcend \\[0.5\baselineskip]
        \pcind \pcif \alpha = 1 \pcthen \\
        \pcind \pcind \bar{\bm{z}}^{(\kappa)}_3 \samples{z^{(\kappa)}_3} \Ftn, ~ \bar{r}^{(\kappa)}_1 \samples{z^{(\kappa)}_3} \bit^{\lambda} \\
        \pcind \pcind \com^{(\kappa)}_1 = \commit{\bar{r}_1^{(\kappa)}, \, \bar{\bm{z}}_3} \\ 
        \pcind \pcind \com^{(\kappa)} = \commit{\bar{r}^{(\kappa)}, \, \bar{\bm{z}}_3 + \bm{z}_4} \\ 
        \pcind \pcind b^{(\kappa)}_1 \samplen \big( \hw{\bm{z}^{(\kappa)}_4} = \omega \big) \\
        \pcind \pcend \\
        \pcend \\[0.5\baselineskip]
        b_1 = \bigwedge\nolimits_{\kappa \in K} b^{(\kappa)}_1 \\ 
        b_2 \samplen \big( h = \commit{\bar{r}, \, (\aux^{(k)})_{k \in \intoneto{M}} } \big)\\
        b_3 \samplen \big( h' = \commit{\bar{r'}, \, (\com^{(\kappa)})_{\kappa \in K} } \big)\\
        \pcreturn b_1 \wedge b_2 \wedge b_3
      }
    }}
    \caption{$\verif$ algorithm for Sig 1 (5-round)} \label{fig:signature-pok1-5r-verify}
  \end{center}
\end{figure}

%--------------------------------------------------------------------%

%--------------------------------------------------------------------%
\newpage
\section{PoK 2 (3-round, without optimization)} \label{app:pok2-3r-nopt}

\begin{figure}[!ht] 
  \begin{center}
    \resizebox{1\textwidth}{!}{\fbox{
      \pseudocode{%
        \hspace{160pt} \> \> \hspace{160pt} \\[-0.75\baselineskip] % Hack to force font scaling
        \underline{\proverf(\bm{x}, \bm{e}, \bm{G}, \bm{y})} \> \> \underline{\verifierf(\bm{G}, \bm{y})} \\
        \pcfor k \in \intoneto{M} \pcdo \\
        \pcind \theta^{(k)} \sampler \bit^{\lambda} \\
        \pcind \pcfor i \in \intoneto{N} \pcdo \\
        \pcind \pcind \theta^{(k)}_{i} \samples{\theta^{(k)}} \bit^{\lambda}, ~ \pi^{(k)}_i \samples{\theta^{(k)}_i} \hperm{n} \\
        \pcind \pcind \bm{u}^{(k)}_i \samples{\theta^{(k)}_i} \Ftk, ~ \bm{v}^{(k)}_i \samples{\theta^{(k)}_i} \Ftn \\
        \pcind \pcend \\
        \pcind \bm{u}^{(k)} = \sum\nolimits_{i \in \intoneto{N}} \bm{u}^{(k)}_i \\ 
        \pcind \pcfor i \in \intoneto{N} \pcdo \\
        \pcind \pcind r^{(k)}_{i} \samples{\theta_i} \bit^{\lambda}, ~ \theta^{(k)}_{i^*} = (\theta^{(k)}_j)_{j \in \intoneto{N} \setminus i} \\ 
        \pcind \pcind \com^{(k)}_{i} = \commit{r^{(k)}_i, \, \pi^{(k)}_i[\bm{y} + \bm{u}^{(k)} \bm{G}] + \bm{v}^{(k)}_i \, || \, \theta^{(k)}_{i^*}} \\
        \pcind \pcend \\
        \pcind \aux^{(k)} = (\com^{(k)}_i)_{i \in \intoneto{N}} \\[\baselineskip]
        \pcind \pcfor i \in [1, N] \pcdo \\
        \pcind \pcind \bm{s}^{(k)}_i = \pi^{(k)}_i[(\bm{u}^{(k)} + \bm{x}) \bm{G}] + \bm{v}^{(k)}_i \\ 
        \pcind \pcend \\
        \pcind r^{(k)} \sampler \bit^{\lambda}, ~ \bm{s}^{(k)} = (\bm{s}^{(k)}_i)_{i \in \intoneto{N}} \\ 
        \pcind \com^{(k)} = \commit{r^{(k)}, \, \bm{u}^{(k)} + \bm{x} \, || \, \bm{s}^{(k)}} \\
        \pcend \\
        \> \sendmessageright*{(\aux^{(k)}, \com^{(k)})_{k \in \intoneto{M}}} \\[-0.5\baselineskip]
        \> \> \kappa \sampler \intoneto{M} \\ 
        \> \> \alpha \sampler \bit \\[-\baselineskip]
        \> \sendmessageleft*{(\kappa, \alpha)} \\[-0.5\baselineskip]
        \bm{z}_1 = \bm{u}^{(\kappa)}_\alpha + \bm{x}, ~ \bm{z}_2 = \pi^{(\kappa)}_{\alpha}[\bm{e}], ~ \bm{z}_3 = \bm{s}^{(\kappa)}_\alpha, ~ z_4 = \theta^{(\kappa)}_{\alpha^*} \\
        \rsp = (r_{\alpha}^{(\kappa)}, r^{(\kappa)}, \bm{z}_1, \bm{z}_2, \bm{z}_3, z_4) \\[-\baselineskip]
        \> \sendmessageright*{(\theta^{(k)})_{k \in \intoneto{M} \setminus \kappa}, \, \rsp} \\[-\baselineskip]
        \> \> \pcfor k \in \intoneto{M} \setminus \kappa \pcdo \\
        \> \> \pcind \tx{Compute } \bar{\aux}^{(k)} \tx{ from } \theta^{(k)} \\
        \> \> \pcind \pcif (\aux^{(k)} \neq \bar{\aux}^{(k)}) ~ \pcreturn 0 \\
        \> \> \pcend \\[0.75\baselineskip]
        \> \> \tx{Compute } (\bar{\pi}_i^{(\kappa)}, \bar{\bm{u}}_i^{(\kappa)}, \bar{\bm{v}}_i^{(\kappa)})_{i \in \intoneto{N} \setminus \alpha} \tx{ from } z_4 \\
        \> \> \bar{\bm{z}}_1 = \bm{z}_1 + \sum\nolimits_{i \in \intoneto{N} \setminus \alpha} \bar{\bm{u}}^{(\kappa)}_i \\ 
        \> \> \pcfor i \in [1, N] \setminus \alpha \pcdo \\
        \> \> \pcind \bar{\bm{s}}^{(\kappa)}_i = \bar{\pi}^{(\kappa)}_i[\bar{\bm{z}}^{(\kappa)}_{1}\bm{G}] + \bar{\bm{v}}^{(\kappa)}_i \\
        \> \> \pcend \\
        \> \> \bar{\bm{s}}^{(\kappa)} = (\bar{\bm{s}}^{(\kappa)}_1, \, \cdots, \, \bar{\bm{s}}^{(\kappa)}_{\alpha - 1}, \, \bm{z}_3, \, \bar{\bm{s}}^{(\kappa)}_{\alpha + 1}, \, \cdots, \, \bar{\bm{s}}^{(\kappa)}_N) \\[0.75\baselineskip]
        \> \> b_1 \samplen \big( \com^{(\kappa)} = \commit{r^{(\kappa)}, \, \bar{\bm{z}}_1 \, || \, \bar{\bm{s}}^{(\kappa)}} \big) \\ 
        \> \> b_2 \samplen \big( \com^{(\kappa)}_\alpha = \commit{r^{(\kappa)}_{\alpha}, \, \bm{z}_3 + \bm{z}_2 \, || \, z_4} \big) \\
        \> \> b_3 \samplen \big( \hw{\bm{z}_2} = \omega \big) \\
        \> \> \pcreturn b_1 \wedge b_2 \wedge b_3
      }
    }}
    \caption{3-round ZK PoK for the $\GSD$ problem (without optimization) \label{fig:pok2-3r-nopt}}
  \end{center}
\end{figure}

%--------------------------------------------------------------------%

%--------------------------------------------------------------------%
\newpage
\section{PoK 2 (3-round, with optimizations)} \label{app:pok2-3r-opt}

\vspace{-\baselineskip}
\begin{figure}[!ht]
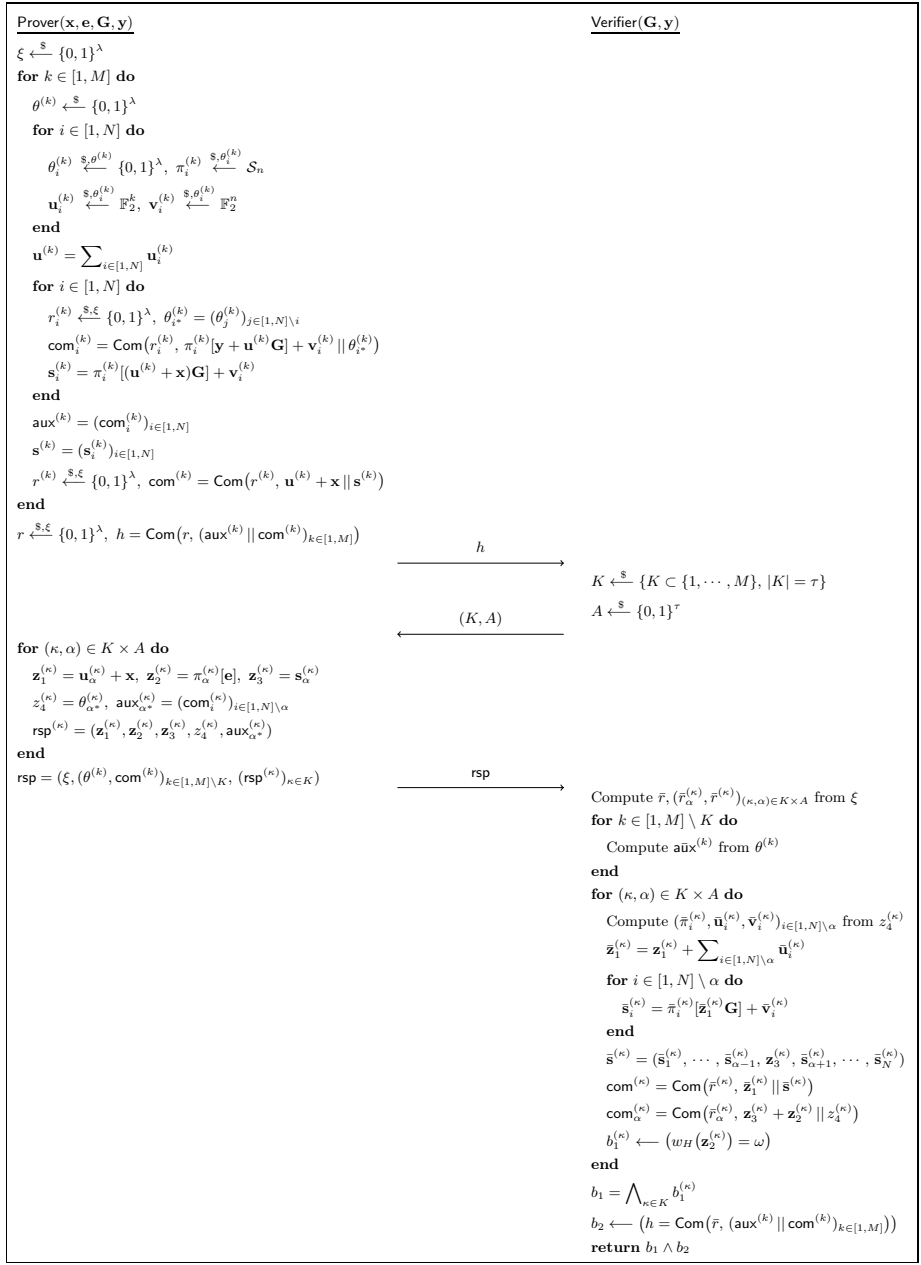
 
  \begin{center}
    \resizebox{1\textwidth}{!}{\fbox{
      \pseudocode{%
        \hspace{160pt} \> \> \hspace{160pt} \\[-0.75\baselineskip] % Hack to force font scaling
        %
        % Prover 1
        \underline{\proverf(\bm{x}, \bm{e}, \bm{G}, \bm{y})} \> \> \underline{\verifierf(\bm{G}, \bm{y})} \\
        \xi \sampler \bit^{\lambda} \\
        \pcfor k \in \intoneto{M} \pcdo \\
        \pcind \theta^{(k)} \sampler \bit^{\lambda} \\
        \pcind \pcfor i \in \intoneto{N} \pcdo \\
        \pcind \pcind \theta^{(k)}_{i} \samples{\theta^{(k)}} \bit^{\lambda}, ~ \pi^{(k)}_i \samples{\theta^{(k)}_i} \hperm{n} \\
        \pcind \pcind \bm{u}^{(k)}_i \samples{\theta^{(k)}_i} \Ftk, ~ \bm{v}^{(k)}_i \samples{\theta^{(k)}_i} \Ftn \\
        \pcind \pcend \\
        \pcind \bm{u}^{(k)} = \sum\nolimits_{i \in \intoneto{N}} \bm{u}^{(k)}_i \\ 
        \pcind \pcfor i \in \intoneto{N} \pcdo \\
        \pcind \pcind r^{(k)}_{i} \samples{\xi} \bit^{\lambda}, ~ \theta^{(k)}_{i^*} = (\theta^{(k)}_{j})_{j \in \intoneto{N} \setminus i} \\
        \pcind \pcind \com^{(k)}_{i} = \commit{r^{(k)}_{i}, \, \pi^{(k)}_i[\bm{y} + \bm{u}^{(k)} \bm{G}] + \bm{v}^{(k)}_i \, || \, \theta^{(k)}_{i^*}} \\
        \pcind \pcind \bm{s}^{(k)}_i = \pi^{(k)}_i[(\bm{u}^{(k)} + \bm{x}) \bm{G}] + \bm{v}^{(k)}_i \\ 
        \pcind \pcend \\
        \pcind \aux^{(k)} = (\com^{(k)}_{i})_{i \in \intoneto{N}} \\
        \pcind \bm{s}^{(k)} = (\bm{s}^{(k)}_i)_{i \in \intoneto{N}} \\
        \pcind r^{(k)} \samples{\xi} \bit^{\lambda}, ~ \com^{(k)} = \commit{r^{(k)}, \, \bm{u}^{(k)} + \bm{x} \, || \, \bm{s}^{(k)}} \\
        \pcend \\
        r \samples{\xi} \bit{^\lambda}, ~ h = \commit{r, \, (\aux^{(k)} \, || \, \com^{(k)})_{k \in \intoneto{M}}} \\[-0.75\baselineskip]
        \> \sendmessageright*{h} \\[-0.5\baselineskip]
        \> \> K \sampler \{ K \subset \{1, \cdots, M \}, \, |K| = \tau \} \\ 
        \> \> A \sampler \bit^{\tau} \\[-\baselineskip]
        \> \sendmessageleft*{(K, A)} \\[-0.5\baselineskip]
        \pcfor (\kappa, \alpha) \in K \times A \pcdo \\
        \pcind \bm{z}^{(\kappa)}_1 = \bm{u}^{(\kappa)}_\alpha + \bm{x}, ~ \bm{z}^{(\kappa)}_2 = \pi^{(\kappa)}_{\alpha}[\bm{e}], ~ \bm{z}^{(\kappa)}_3 = \bm{s}^{(\kappa)}_\alpha \\
        \pcind z^{(\kappa)}_4 = \theta^{(\kappa)}_{\alpha^*}, ~ \aux^{(\kappa)}_{\alpha^*} = (\com^{(\kappa)}_{i})_{i \in \intoneto{N} \setminus \alpha} \\
        \pcind \rsp^{(\kappa)} = (\bm{z}^{(\kappa)}_1, \bm{z}^{(\kappa)}_2, \bm{z}^{(\kappa)}_3, z^{(\kappa)}_4, \aux^{(\kappa)}_{\alpha^*}) \\
        \pcend \\[-0.5\baselineskip]
        \rsp = (\xi, (\theta^{(k)}, \com^{(k)})_{k \in \intoneto{M} \setminus K}, \, (\rsp^{(\kappa)})_{\kappa \in K})
        \> \sendmessageright*{\rsp} \\[-0.75\baselineskip]
        \> \> \tx{Compute } \bar{r}, (\bar{r}^{(\kappa)}_{\alpha}, \bar{r}^{(\kappa)})_{(\kappa, \alpha) \in K \times A} \tx{ from } \xi \\
        \> \> \pcfor k \in \intoneto{M} \setminus K \pcdo \\
        \> \> \pcind \tx{Compute } \bar{\aux}^{(k)} \tx{ from } \theta^{(k)} \\
        \> \> \pcend \\
        \> \> \pcfor (\kappa, \alpha) \in K \times A \pcdo \\
        \> \> \pcind \tx{Compute } (\bar{\pi}_i^{(\kappa)}, \bar{\bm{u}}_i^{(\kappa)}, \bar{\bm{v}}_i^{(\kappa)})_{i \in \intoneto{N} \setminus \alpha} \tx{ from } z^{(\kappa)}_4 \\
        \> \> \pcind \bar{\bm{z}}^{(\kappa)}_1 = \bm{z}^{(\kappa)}_1 + \sum\nolimits_{i \in \intoneto{N} \setminus \alpha} \bar{\bm{u}}^{(\kappa)}_i \\
        \> \> \pcind \pcfor i \in [1, N] \setminus \alpha \pcdo \\
        \> \> \pcind \pcind \bar{\bm{s}}^{(\kappa)}_i = \bar{\pi}^{(\kappa)}_i[\bar{\bm{z}}^{(\kappa)}_1 \bm{G}] + \bar{\bm{v}}^{(\kappa)}_i \\
        \> \> \pcind \pcend \\
        \> \> \pcind \bar{\bm{s}}^{(\kappa)} = (\bar{\bm{s}}^{(\kappa)}_1, \, \cdots, \, \bar{\bm{s}}^{(\kappa)}_{\alpha - 1}, \, \bm{z}^{(\kappa)}_3, \, \bar{\bm{s}}^{(\kappa)}_{\alpha + 1}, \, \cdots, \, \bar{\bm{s}}^{(\kappa)}_N) \\
        \> \> \pcind \com^{(\kappa)} = \commit{\bar{r}^{(\kappa)}, \, \bar{\bm{z}}^{(\kappa)}_1 \, || \, \bar{\bm{s}}^{(\kappa)}} \\ 
        \> \> \pcind \com^{(\kappa)}_{\alpha} = \commit{\bar{r}^{(\kappa)}_{\alpha}, \, \bm{z}^{(\kappa)}_3 + \bm{z}^{(\kappa)}_2 \, || \, z^{(\kappa)}_4} \\
        \> \> \pcind b_1^{(\kappa)} \samplen \big( \hw{\bm{z}_2^{(\kappa)}} = \omega \big) \\
        \> \> \pcend \\
        \> \> b_1 = \bigwedge\nolimits_{\kappa \in K} b_1^{(\kappa)} \\ 
        \> \> b_2 \samplen \big( h = \commit{\bar{r}, \, (\aux^{(k)} \, || \, \com^{(k)})_{k \in \intoneto{M}}} \big) \\
        \> \> \pcreturn b_1 \wedge b_2
      }
    }}
    \vspace{-0.25\baselineskip}
    \caption{ZK PoK for the $\GSD$ problem over $\Ft$ (with optimizations) \label{fig:pok2-3r-opt}}
  \end{center}
\end{figure}

%--------------------------------------------------------------------%

%--------------------------------------------------------------------%
\newpage
\section{Sig 2} \label{app:sig2}
\begin{figure}[!ht] 
  \begin{center}
    \resizebox{1\textwidth}{!}{\fbox{
      \pseudocode{%
        \hspace{160pt} \> \> \hspace{160pt} \\[-0.75\baselineskip] % Hack to force font scaling
        \underline{\keygen(\lambda)} \\
        \rho_1 \sampler \bit^{\lambda}, ~ \bm{x} \samples{\rho_1} \Ftk, ~ \bm{e} \samples{\rho_1} \Ftn \tx{ such that } \hw{\bm{e}} = \omega \\
        \rho_2 \sampler \bit^{\lambda}, ~ \bm{G} \samples{\rho_2} \Ft^{\ktn}, ~ \bm{y} = \bm{x} \bm{G} + \bm{e} \\
        \pcreturn (\sk, \pk) = (\rho_1, (\rho_2, \bm{y})) \\[0.75\baselineskip]
        \underline{\sign(\sk, \pk, m)} \\
        \xi \sampler \bit^{\lambda} \\
        \pcfor k \in \intoneto{M} \pcdo \\
        \pcind \theta^{(k)} \sampler \bit^{\lambda} \\
        \pcind \pcfor i \in \intoneto{N} \pcdo \\
        \pcind \pcind \theta^{(k)}_{i} \samples{\theta^{(k)}} \bit^{\lambda}, ~ \pi^{(k)}_i \samples{\theta^{(k)}_i} \hperm{n} \\
        \pcind \pcind \bm{u}^{(k)}_i \samples{\theta^{(k)}_i} \Ftk, ~ \bm{v}^{(k)}_i \samples{\theta^{(k)}_i} \Ftn \\
        \pcind \pcend \\
        \pcind \bm{u}^{(k)} = \sum\nolimits_{i \in \intoneto{N}} \bm{u}^{(k)}_i \\ 
        \pcind \pcfor i \in \intoneto{N} \pcdo \\
        \pcind \pcind r^{(k)}_{i} \samples{\xi} \bit^{\lambda}, ~ \theta^{(k)}_{i^*} = (\theta^{(k)}_{j})_{j \in \intoneto{N} \setminus i} \\
        \pcind \pcind \com^{(k)}_{i} = \commit{r^{(k)}_{i}, \, \pi^{(k)}_i[\bm{y} + \bm{u}^{(k)} \bm{G}] + \bm{v}^{(k)}_i \, || \, \theta^{(k)}_{i^*}} \\
        \pcind \pcind \bm{s}^{(k)}_i = \pi^{(k)}_i[(\bm{u}^{(k)} + \bm{x}) \bm{G}] + \bm{v}^{(k)}_i \\ 
        \pcind \pcend \\
        \pcind \aux^{(k)} = (\com^{(k)}_{i})_{i \in \intoneto{N}} \\
        \pcind \bm{s}^{(k)} = (\bm{s}^{(k)}_i)_{i \in \intoneto{N}} \\
        \pcind r^{(k)} \samples{\xi} \bit^{\lambda}, ~ \com^{(k)} = \commit{r^{(k)}, \, \bm{u}^{(k)} + \bm{x} \, || \, \bm{s}^{(k)}} \\
        \pcend \\
        r \samples{\xi} \bit{^\lambda}, ~ h = \commit{r, \, (\aux^{(k)} \, || \, \com^{(k)})_{k \in \intoneto{M}}} \\
        (K, A) \samplen \ro{m \, || \, \pk \, || \, h} \\
        \pcfor (\kappa, \alpha) \in K \times A \pcdo \\
        \pcind \bm{z}^{(\kappa)}_1 = \bm{u}^{(\kappa)}_\alpha + \bm{x}, ~ \bm{z}^{(\kappa)}_2 = \pi^{(\kappa)}_{\alpha}[\bm{e}], ~ \bm{z}^{(\kappa)}_3 = \bm{s}^{(\kappa)}_\alpha \\
        \pcind z^{(\kappa)}_4 = \theta^{(\kappa)}_{\alpha^*}, ~ \aux^{(\kappa)}_{\alpha^*} = (\com^{(\kappa)}_{i})_{i \in \intoneto{N} \setminus \alpha} \\
        \pcind \rsp^{(\kappa)} = (\bm{z}^{(\kappa)}_1, \bm{z}^{(\kappa)}_2, \bm{z}^{(\kappa)}_3, z^{(\kappa)}_4, \aux^{(\kappa)}_{\alpha^*}) \\
        \pcend \\
        \rsp = (\xi, (\theta^{(k)}, \com^{(k)})_{k \in \intoneto{M} \setminus K}, \, (\rsp^{(\kappa)})_{\kappa \in K}) \\
        \pcreturn \sigma = (h, \rsp)
        }
    }}
    \caption{$\keygen$ and $\sign$ algorithms for Sig 2} \label{fig:signature-2-keygen-sign}
  \end{center}
\end{figure}

\newpage

\begin{figure}[!ht] 
  \begin{center}
    \resizebox{1\textwidth}{!}{\fbox{
      \pseudocode{%
        \hspace{160pt} \> \> \hspace{60pt} \\[-0.75\baselineskip] % Hack to force font scaling
        \underline{\verif(\pk, \sigma, m)} \\
        \tx{Parse } \sigma \tx{ as } \sigma = (h, \rsp) \tx{ and } \rsp \tx{ as } \rsp = (\xi, (\theta^{(k)}, \com^{(k)})_{k \in \intoneto{M} \setminus K}, \, (\rsp^{(\kappa)})_{\kappa \in K}) \\
        (K, \, A) \samplen \ro{m \, || \, \pk \, || \, h} \\[0.5\baselineskip]
        \tx{Compute } \bar{r}, (\bar{r}^{(\kappa)}_{\alpha}, \bar{r}^{(\kappa)})_{(\kappa, \alpha) \in K \times A} \tx{ from } \xi \\
        \pcfor k \in \intoneto{M} \setminus K \pcdo \\
        \pcind \tx{Compute } \bar{\aux}^{(k)} \tx{ from } \theta^{(k)} \\
        \pcend \\
        \pcfor (\kappa, \alpha) \in K \times A \pcdo \\
        \pcind \tx{Compute } (\bar{\pi}_i^{(\kappa)}, \bar{\bm{u}}_i^{(\kappa)}, \bar{\bm{v}}_i^{(\kappa)})_{i \in \intoneto{N} \setminus \alpha} \tx{ from } z^{(\kappa)}_4 \\
        \pcind \bar{\bm{z}}^{(\kappa)}_1 = \bm{z}^{(\kappa)}_1 + \sum\nolimits_{i \in \intoneto{N} \setminus \alpha} \bar{\bm{u}}^{(\kappa)}_i \\
        \pcind \pcfor i \in [1, N] \setminus \alpha \pcdo \\
        \pcind \pcind \bar{\bm{s}}^{(\kappa)}_i = \bar{\pi}^{(\kappa)}_i[\bar{\bm{z}}^{(\kappa)}_1 \bm{G}] + \bar{\bm{v}}^{(\kappa)}_i \\
        \pcind \pcend \\
        \pcind \bar{\bm{s}}^{(\kappa)} = (\bar{\bm{s}}^{(\kappa)}_1, \, \cdots, \, \bar{\bm{s}}^{(\kappa)}_{\alpha - 1}, \, \bm{z}^{(\kappa)}_3, \, \bar{\bm{s}}^{(\kappa)}_{\alpha + 1}, \, \cdots, \, \bar{\bm{s}}^{(\kappa)}_N) \\
        \pcind \com^{(\kappa)} = \commit{\bar{r}^{(\kappa)}, \, \bar{\bm{z}}^{(\kappa)}_1 \, || \, \bar{\bm{s}}^{(\kappa)}} \\ 
        \pcind \com^{(\kappa)}_{\alpha} = \commit{\bar{r}^{(\kappa)}_{\alpha}, \, \bm{z}^{(\kappa)}_3 + \bm{z}^{(\kappa)}_2 \, || \, z^{(\kappa)}_4} \\
        \pcind b_1^{(\kappa)} \samplen \big( \hw{\bm{z}_2^{(\kappa)}} = \omega \big) \\
        \pcend \\
        b_1 = \bigwedge\nolimits_{\kappa \in K} b_1^{(\kappa)} \\ 
        b_2 \samplen \big( h = \commit{\bar{r}, \, (\aux^{(k)} \, || \, \com^{(k)})_{k \in \intoneto{M}}} \big) \\
        \pcreturn b_1 \wedge b_2
      }
    }}
    \caption{$\verif$ algorithm for Sig 2} \label{fig:signature-2-verify}
  \end{center}
\end{figure}

%--------------------------------------------------------------------%

%--------------------------------------------------------------------%
\newpage
\section{PoK 3 (3-round, without optimization)} \label{app:pok3-3r-nopt}

\begin{figure}[!ht] 
  \begin{center}
    \resizebox{1\textwidth}{!}{\fbox{
      \pseudocode{%
        \hspace{160pt} \> \> \hspace{160pt} \\[-0.75\baselineskip] % Hack to force font scaling
        \underline{\proverf(\bm{x}, \bm{e}, \bm{G}, \bm{y})} \> \> \underline{\verifierf(\bm{G}, \bm{y})} \\
        \pcfor k \in \intoneto{M} \pcdo \\
        \pcind \theta^{(k)} \sampler \bit^{\lambda} \\
        \pcind \pcfor i \in \intoneto{N} \pcdo \\
        \pcind \pcind \theta^{(k)}_{i} \samples{\theta^{(k)}} \bit^{\lambda}, ~ \pi^{(k)}_i \samples{\theta^{(k)}_i} \hperm{n} \\
        \pcind \pcind \bm{u}^{(k)}_i \samples{\theta^{(k)}_i} \Ftk, ~ \bm{v}^{(k)}_i \samples{\theta^{(k)}_i} \Ftn \\
        \pcind \pcend \\
        \pcind \pi^{(k)} = \pi^{(k)}_N \circ \cdots \circ \pi^{(k)}_1 \\
        \pcind \bm{u}^{(k)} = \sum\nolimits_{i \in \intoneto{N}} \bm{u}^{(k)}_i \\ 
        \pcind \bm{v}^{(k)} = \bm{v}^{(k)}_N + \sum\nolimits_{i \in \intoneto{N - 1}} \pi^{(k)}_N \circ \cdots \circ \pi^{(k)}_{i + 1}[\bm{v}^{(k)}_i] \\
        \pcind \bm{s}^{(k)}_0 = (\bm{u}^{(k)} + \bm{x}) \bm{G} \\
        \pcind \pcfor i \in \intoneto{N} \pcdo \\
        \pcind \pcind r^{(k)}_{i} \samples{\theta^{(k)}} \bit^{\lambda}, ~ \theta^{(k)}_{i^*} = (\theta^{(k)}_j)_{j \in \intoneto{N} \setminus i} \\ 
        \pcind \pcind \bm{s}^{(k)}_i = \pi^{(k)}_i[\bm{s}^{(k)}_{i - 1}] + \bm{v}^{(k)}_i \\ 
        \pcind \pcind \com^{(k)}_{i} = \commit{r^{(k)}_i, \, \pi^{(k)}[\bm{y} + \bm{u}^{(k)} \bm{G}] + \bm{v}^{(k)} \, || \, \theta^{(k)}_{i^*}} \\
        \pcind \pcend \\
        \pcind \aux^{(k)} = (\com^{(k)}_i)_{i \in \intoneto{N}} \\
        \pcind r^{(k)} \sampler \bit^{\lambda}, ~ \bm{s}^{(k)} = (\bm{s}^{(k)}_i)_{i \in \intoneto{N}} \\
        \pcind \com^{(k)} = \commit{r^{(k)}, \, \bm{u}^{(k)} + \bm{x} \, || \, \bm{s}^{(k)}}  \\
        \pcend \\
        \> \sendmessageright*{(\aux^{(k)}, \com^{(k)})_{k \in \intoneto{M}}} \\[-0.5\baselineskip]
        \> \> \kappa \sampler \intoneto{M} \\ 
        \> \> \alpha \sampler \bit \\[-\baselineskip]
        \> \sendmessageleft*{(\kappa, \alpha)} \\[-0.5\baselineskip]
        \bm{z}_1 = \bm{u}^{(\kappa)}_\alpha + \bm{x}, ~ \bm{z}_2 = \pi^{(\kappa)}[\bm{e}] \\
        \bm{z}_3 = \bm{s}^{(\kappa)}_\alpha, ~ z_4 = \theta^{(\kappa)}_{\alpha^*} \\
        \rsp = (r^{(\kappa)}_{\alpha}, r^{(\kappa)}, \bm{z}_1, \bm{z}_2, \bm{z}_3, z_4) \\[-\baselineskip]
        \> \sendmessageright*{(\theta^{(k)})_{k \in \intoneto{M} \setminus \kappa}, \, \rsp} \\[-\baselineskip]
        \> \> \pcfor k \in \intoneto{M} \setminus \kappa \pcdo \\
        \> \> \pcind \tx{Compute } \bar{\aux}^{(k)} \tx{ from } \theta^{(k)} \\
        \> \> \pcind \pcif (\aux^{(k)} \neq \bar{\aux}^{(k)}) ~ \pcreturn 0 \\
        \> \> \pcend \\[0.75\baselineskip]
        \> \> \tx{Compute } (\bar{\pi}_i, \bar{\bm{v}}_i)_{i \in \intoneto{N} \setminus \alpha} \tx{ from } z_4 \\
        \> \> \bar{\bm{z}}_1 = \bm{z}_1 + \sum\nolimits_{i \in \intoneto{N} \setminus \alpha} \bm{u}^{(\kappa)}_i \\ 
        \> \> \bar{\bm{s}}_0 = \bar{\bm{z}}_1 \bm{G} \\
        \> \> \pcfor i \in [1, N] \setminus \alpha \pcdo \\
        \> \> \pcind \bar{\bm{s}}^{(\kappa)}_i = \bar{\pi}^{(\kappa)}_i[\bar{\bm{s}}^{(\kappa)}_{i - 1}] + \bar{\bm{v}}^{(\kappa)}_i \\
        \> \> \pcend \\
        \> \> \bar{\bm{s}}^{(\kappa)} = (\bar{\bm{s}}^{(\kappa)}_1, \, \cdots, \, \bar{\bm{s}}^{(\kappa)}_{\alpha - 1}, \, \bm{z}_3, \, \bar{\bm{s}}^{(\kappa)}_{\alpha + 1}, \, \cdots, \, \bar{\bm{s}}^{(\kappa)}_N) \\[0.75\baselineskip]
        \> \> b_1 \samplen \big( \com^{(\kappa)} = \commit{r^{(\kappa)}, \, \bar{\bm{z}}_1 \, || \, \bar{\bm{s}}^{(\kappa)}} \big) \\ 
        \> \> b_2 \samplen \big( \com^{(\kappa)}_\alpha = \commit{r^{(\kappa)}_{\alpha}, \, \bar{\bm{s}}^{(\kappa)}_N + \bm{z}_2 \, || \, z_4} \big) \\
        \> \> b_3 \samplen \big( \hw{\bm{z}_2} = \omega \big)
      }
    }}
    \caption{3-round HVZK PoK for the $\GSD$ problem (without optimization) \label{fig:pok3-3r-nopt}}
  \end{center}
\end{figure}

%--------------------------------------------------------------------%

%--------------------------------------------------------------------%
\newpage
\section{PoK 3 (3-round, with optimizations)} \label{app:pok3-3r-opt}

\vspace{-\baselineskip}
\begin{figure}[!ht] 
  \begin{center}
    \resizebox{1\textwidth}{!}{\fbox{
      \pseudocode{%
        \hspace{160pt} \> \> \hspace{160pt} \\[-0.75\baselineskip] % Hack to force font scaling
        \underline{\proverf(\bm{x}, \bm{e}, \bm{G}, \bm{y})} \> \> \underline{\verifierf(\bm{G}, \bm{y})} \\
        \xi \sampler \bit^{\lambda} \\ 
        \pcfor k \in \intoneto{M} \pcdo \\
        \pcind \theta^{(k)} \sampler \bit^{\lambda} \\
        \pcind \pcfor i \in \intoneto{N} \pcdo \\
        \pcind \pcind \theta^{(k)}_{i} \samples{\theta^{(k)}} \bit^{\lambda}, ~ \pi^{(k)}_i \samples{\theta^{(k)}_i} \hperm{n} \\
        \pcind \pcind \bm{u}^{(k)}_i \samples{\theta^{(k)}_i} \Ftk, ~ \bm{v}^{(k)}_i \samples{\theta^{(k)}_i} \Ftn \\
        \pcind \pcend \\
        \pcind \pi^{(k)} = \pi^{(k)}_N \circ \cdots \circ \pi^{(k)}_1, ~ \bm{u}^{(k)} = \sum\nolimits_{i \in \intoneto{N}} \bm{u}^{(k)}_i \\ 
        \pcind \bm{v}^{(k)} = \bm{v}^{(k)}_N + \sum\nolimits_{i \in \intoneto{N - 1}} \pi^{(k)}_N \circ \cdots \circ \pi^{(k)}_{i + 1}[\bm{v}^{(k)}_i] \\
        \pcind \bm{s}^{(k)}_0 = (\bm{u}^{(k)} + \bm{x}) \bm{G} \\
        \pcind \pcfor i \in \intoneto{N} \pcdo \\
        \pcind \pcind r^{(k)}_{1, i} \samples{\xi} \bit^{\lambda}, ~ \com^{(k)}_{1,i} = \commit{r^{(k)}_{1,i}, \, \theta^{(k)}_{i}} \\
        \pcind \pcind \bm{s}^{(k)}_i = \pi^{(k)}_i[\bm{s}^{(k)}_{i - 1}] + \bm{v}^{(k)}_i \\ 
        \pcind \pcend \\
        \pcind r^{(k)}_{2} \samples{\xi} \bit^{\lambda}, ~ \com^{(k)}_{2} = \commit{r^{(k)}_{2}, \, \pi^{(k)}[\bm{y} + \bm{u}^{(k)} \bm{G}] + \bm{v}^{(k)}} \\
        \pcind \bm{s}^{(k)} = (\bm{s}^{(k)}_i)_{i \in \intoneto{N}}, ~ \aux^{(k)} = (\com^{(k)}_{1,i}, \, \com^{(k)}_{2})_{i \in \intoneto{N}} \\
        \pcind r^{(k)} \samples{\xi} \bit^{\lambda}, ~ \com^{(k)} = \commit{r^{(k)}, \, \bm{u}^{(k)} + \bm{x} \, || \, \bm{s}^{(k)}}  \\
        \pcend \\
        r \samples{\xi} \bit{^\lambda}, ~ h = \commit{r, \, (\aux^{(k)} \, || \, \com^{(k)})_{k \in \intoneto{M}}} \\[-0.75\baselineskip]
        \> \sendmessageright*{h} \\[-0.5\baselineskip]
        \> \> K \sampler \{ K \subset \{1, \cdots, M \}, \, |K| = \tau \} \\ 
        \> \> A \sampler \bit^{\tau} \\[-\baselineskip]
        \> \sendmessageleft*{(K, A)} \\[-0.5\baselineskip]
        \pcfor (\kappa, \alpha) \in K \times A \pcdo \\
        \pcind \bm{z}^{(\kappa)}_1 = \bm{u}^{(\kappa)}_\alpha + \bm{x}, ~ \bm{z}^{(\kappa)}_2 = \pi^{(\kappa)}[\bm{e}], ~ \bm{z}^{(\kappa)}_3 = \bm{s}^{(\kappa)}_\alpha \\ 
        \pcind z^{(\kappa)}_4 = \theta^{(\kappa)}_{\alpha^*} = (\theta^{(\kappa)}_j)_{j \in \intoneto{N} \setminus \alpha} \\
        \pcind \rsp^{(\kappa)} = (\bm{z}^{(\kappa)}_1, \bm{z}^{(\kappa)}_2, \bm{z}^{(\kappa)}_3, z^{(\kappa)}_4, \com^{(\kappa)}_{1,\alpha}) \\
        \pcend \\
        \rsp = (\xi, (\theta^{(k)}, \com^{(k)})_{k \in \intoneto{M} \setminus K}, \, (\rsp^{(\kappa)})_{\kappa \in K}) \\[-0.75\baselineskip]
        \> \sendmessageright*{\rsp} \\[-\baselineskip]
        \> \> \tx{Compute } \bar{r}, \bar{r'}, (\bar{r}^{(\kappa)}_{1,i}, \bar{r}^{(\kappa)}_{2})^{\kappa \in K}_{i \in \intoneto{N}} \tx{ from } \xi \\
        \> \> \pcfor k \in \intoneto{M} \setminus K \pcdo \\
        \> \> \pcind \tx{Compute } \aux^{(k)} \tx{ from } \theta^{(k)} \\
        \> \> \pcend \\
        \> \> \pcfor (\kappa, \alpha) \in K \times A \pcdo \\
        \> \> \pcind \tx{Compute } (\theta^{(\kappa)}_{i}, \bar{\pi}^{(\kappa)}_i, \bar{\bm{v}}^{(\kappa)}_i)_{i \in \intoneto{N} \setminus \alpha} \tx{ from } z_4 \\
        \> \> \pcind \bar{\bm{z}}^{(\kappa)}_1 = \bm{z}^{(\kappa)}_1 + \sum\nolimits_{i \in \intoneto{N} \setminus \alpha} \bm{u}^{(\kappa)}_i, ~ \bar{\bm{s}}^{(\kappa)}_0 = \bar{\bm{z}}^{(\kappa)}_1 \bm{G} \\
        \> \> \pcind \pcfor i \in [1, N] \setminus \alpha \pcdo \\
        \> \> \pcind \pcind \com^{(\kappa)}_{1,i} = \commit{r^{(\kappa)}_{1,i}, \, \theta^{(\kappa)}_{i}} \\
        \> \> \pcind \pcind \bar{\bm{s}}^{(\kappa)}_i = \bar{\pi}^{(\kappa)}_i[\bar{\bm{s}}^{(\kappa)}_{i - 1}] + \bar{\bm{v}}^{(\kappa)}_i \\
        \> \> \pcind \pcend \\
        \> \> \pcind \bar{\bm{s}}^{(\kappa)} = (\bar{\bm{s}}^{(\kappa)}_1, \, \cdots, \, \bar{\bm{s}}^{(\kappa)}_{\alpha - 1}, \, \bm{z}_3, \, \bar{\bm{s}}^{(\kappa)}_{\alpha + 1}, \, \cdots, \, \bar{\bm{s}}^{(\kappa)}_N) \\
        \> \> \pcind \com^{(\kappa)} = \commit{\bar{r}^{(\kappa)}, \, \bar{\bm{z}}_1 \, || \, \bar{\bm{s}}^{(\kappa)}} \\ 
        \> \> \pcind \com^{(\kappa)}_{2} = \commit{\bar{r}^{(\kappa)}_{2}, \, \bar{\bm{s}}^{(\kappa)}_N + \bm{z}_2} \\
        \> \> \pcind b_1^{(\kappa)} \samplen \big( \hw{\bm{z}_2^{(\kappa)}} = \omega \big) \\
        \> \> \pcend \\
        \> \> b_1 = \bigwedge\nolimits_{\kappa \in K} b_1^{(\kappa)} \\ 
        \> \> b_2 \samplen \big( h = \commit{\bar{r}, \, (\aux^{(k)} \, || \, \com^{(k)})_{k \in \intoneto{M}} } \big)\\
        \> \> \pcreturn b_1 \wedge b_2
      }
    }}
    \caption{3-round HVZK PoK for the $\GSD$ problem (with optimizations) \label{fig:pok3-3r-opt}}
  \end{center}
\end{figure}

%--------------------------------------------------------------------%

%--------------------------------------------------------------------%
\newpage
\section{Sig 3} \label{app:sig3}
\begin{figure}[!ht] 
  \begin{center}
    \resizebox{1\textwidth}{!}{\fbox{
      \pseudocode{%
        \hspace{160pt} \> \> \hspace{160pt} \\[-0.75\baselineskip] % Hack to force font scaling
        \underline{\keygen(\lambda)} \\
        \rho_1 \sampler \bit^{\lambda}, ~ \bm{x} \samples{\rho_1} \Ftk, ~ \bm{e} \samples{\rho_1} \Ftn \tx{ such that } \hw{\bm{e}} = \omega \\
        \rho_2 \sampler \bit^{\lambda}, ~ \bm{G} \samples{\rho_2} \Ft^{\ktn}, ~ \bm{y} = \bm{x} \bm{G} + \bm{e} \\
        \pcreturn (\sk, \pk) = (\rho_1, (\rho_2, \bm{y})) \\[0.75\baselineskip]
        \underline{\sign(\sk, \pk, m)} \\
        \xi \sampler \bit^{\lambda} \\ 
        \pcfor k \in \intoneto{M} \pcdo \\
        \pcind \theta^{(k)} \sampler \bit^{\lambda} \\
        \pcind \pcfor i \in \intoneto{N} \pcdo \\
        \pcind \pcind \theta^{(k)}_{i} \samples{\theta^{(k)}} \bit^{\lambda}, ~ \pi^{(k)}_i \samples{\theta^{(k)}_i} \hperm{n} \\
        \pcind \pcind \bm{u}^{(k)}_i \samples{\theta^{(k)}_i} \Ftk, ~ \bm{v}^{(k)}_i \samples{\theta^{(k)}_i} \Ftn \\
        \pcind \pcend \\
        \pcind \pi^{(k)} = \pi^{(k)}_N \circ \cdots \circ \pi^{(k)}_1, ~ \bm{u}^{(k)} = \sum\nolimits_{i \in \intoneto{N}} \bm{u}^{(k)}_i \\ 
        \pcind \bm{v}^{(k)} = \bm{v}^{(k)}_N + \sum\nolimits_{i \in \intoneto{N - 1}} \pi^{(k)}_N \circ \cdots \circ \pi^{(k)}_{i + 1}[\bm{v}^{(k)}_i] \\
        \pcind \bm{s}^{(k)}_0 = (\bm{u}^{(k)} + \bm{x}) \bm{G} \\
        \pcind \pcfor i \in \intoneto{N} \pcdo \\
        \pcind \pcind r^{(k)}_{1, i} \samples{\xi} \bit^{\lambda}, ~ \com^{(k)}_{1,i} = \commit{r^{(k)}_{1,i}, \, \theta^{(k)}_{i}} \\
        \pcind \pcind \bm{s}^{(k)}_i = \pi^{(k)}_i[\bm{s}^{(k)}_{i - 1}] + \bm{v}^{(k)}_i \\ 
        \pcind \pcend \\
        \pcind r^{(k)}_{2} \samples{\xi} \bit^{\lambda}, ~ \com^{(k)}_{2} = \commit{r^{(k)}_{2}, \, \pi^{(k)}[\bm{y} + \bm{u}^{(k)} \bm{G}] + \bm{v}^{(k)}} \\
        \pcind \bm{s}^{(k)} = (\bm{s}^{(k)}_i)_{i \in \intoneto{N}}, ~ \aux^{(k)} = (\com^{(k)}_{1,i}, \, \com^{(k)}_{2})_{i \in \intoneto{N}} \\
        \pcind r^{(k)} \samples{\xi} \bit^{\lambda}, ~ \com^{(k)} = \commit{r^{(k)}, \, \bm{u}^{(k)} + \bm{x} \, || \, \bm{s}^{(k)}}  \\
        \pcend \\
        r \samples{\xi} \bit{^\lambda}, ~ h = \commit{r, \, (\aux^{(k)} \, || \, \com^{(k)})_{k \in \intoneto{M}}} \\
        (K, A) \samplen \ro{m \, || \, \pk \, || \, h} \\
        \pcfor (\kappa, \alpha) \in K \times A \pcdo \\
        \pcind \bm{z}^{(\kappa)}_1 = \bm{u}^{(\kappa)}_\alpha + \bm{x}, ~ \bm{z}^{(\kappa)}_2 = \pi^{(\kappa)}[\bm{e}], ~ \bm{z}^{(\kappa)}_3 = \bm{s}^{(\kappa)}_\alpha \\ 
        \pcind z^{(\kappa)}_4 = \theta^{(\kappa)}_{\alpha^*} = (\theta^{(\kappa)}_j)_{j \in \intoneto{N} \setminus \alpha} \\
        \pcind \rsp^{(\kappa)} = (\bm{z}^{(\kappa)}_1, \bm{z}^{(\kappa)}_2, \bm{z}^{(\kappa)}_3, z^{(\kappa)}_4, \com^{(\kappa)}_{1,\alpha}) \\
        \pcend \\
        \rsp = (\xi, (\theta^{(k)}, \com^{(k)})_{k \in \intoneto{M} \setminus K}, \, (\rsp^{(\kappa)})_{\kappa \in K}) \\
        \pcreturn \sigma = (h, \rsp)
        }
    }}
    \caption{$\keygen$ and $\sign$ algorithms for Sig 3} \label{fig:signature-3-keygen-sign}
  \end{center}
\end{figure}

\newpage

\begin{figure}[!ht] 
  \begin{center}
    \resizebox{1\textwidth}{!}{\fbox{
      \pseudocode{%
        \hspace{160pt} \> \> \hspace{60pt} \\[-0.75\baselineskip] % Hack to force font scaling
        \underline{\verif(\pk, \sigma, m)} \\
        \tx{Parse } \sigma \tx{ as } \sigma = (h, \rsp) \tx{ and } \rsp \tx{ as } \rsp = (\xi, (\theta^{(k)}, \com^{(k)})_{k \in \intoneto{M} \setminus K}, \, (\rsp^{(\kappa)})_{\kappa \in K}) \\
        (K, \, A) \samplen \ro{m \, || \, \pk \, || \, h} \\[0.5\baselineskip]
        \tx{Compute } \bar{r}, \bar{r'}, (\bar{r}^{(\kappa)}_{1,i}, \bar{r}^{(\kappa)}_{2})^{\kappa \in K}_{i \in \intoneto{N}} \tx{ from } \xi \\
        \pcfor k \in \intoneto{M} \setminus K \pcdo \\
        \pcind \tx{Compute } \aux^{(k)} \tx{ from } \theta^{(k)} \\
        \pcend \\
        \pcfor (\kappa, \alpha) \in K \times A \pcdo \\
        \pcind \tx{Compute } (\theta^{(\kappa)}_{i}, \bar{\pi}^{(\kappa)}_i, \bar{\bm{v}}^{(\kappa)}_i)_{i \in \intoneto{N} \setminus \alpha} \tx{ from } z_4 \\
        \pcind \bar{\bm{z}}^{(\kappa)}_1 = \bm{z}^{(\kappa)}_1 + \sum\nolimits_{i \in \intoneto{N} \setminus \alpha} \bm{u}^{(\kappa)}_i, ~ \bar{\bm{s}}^{(\kappa)}_0 = \bar{\bm{z}}^{(\kappa)}_1 \bm{G} \\
        \pcind \pcfor i \in [1, N] \setminus \alpha \pcdo \\
        \pcind \pcind \com^{(\kappa)}_{1,i} = \commit{r^{(\kappa)}_{1,i}, \, \theta^{(\kappa)}_{i}} \\
        \pcind \pcind \bar{\bm{s}}^{(\kappa)}_i = \bar{\pi}^{(\kappa)}_i[\bar{\bm{s}}^{(\kappa)}_{i - 1}] + \bar{\bm{v}}^{(\kappa)}_i \\
        \pcind \pcend \\
        \pcind \bar{\bm{s}}^{(\kappa)} = (\bar{\bm{s}}^{(\kappa)}_1, \, \cdots, \, \bar{\bm{s}}^{(\kappa)}_{\alpha - 1}, \, \bm{z}_3, \, \bar{\bm{s}}^{(\kappa)}_{\alpha + 1}, \, \cdots, \, \bar{\bm{s}}^{(\kappa)}_N) \\
        \pcind \com^{(\kappa)} = \commit{\bar{r}^{(\kappa)}, \, \bar{\bm{z}}_1 \, || \, \bar{\bm{s}}^{(\kappa)}} \\ 
        \pcind \com^{(\kappa)}_{2} = \commit{\bar{r}^{(\kappa)}_{2}, \, \bar{\bm{s}}^{(\kappa)}_N + \bm{z}_2} \\
        \pcind b_1^{(\kappa)} \samplen \big( \hw{\bm{z}_2^{(\kappa)}} = \omega \big) \\
        \pcend \\
        b_1 = \bigwedge\nolimits_{\kappa \in K} b_1^{(\kappa)} \\ 
        b_2 \samplen \big( h = \commit{\bar{r}, \, (\aux^{(k)} \, || \, \com^{(k)})_{k \in \intoneto{M}} } \big)\\
        \pcreturn b_1 \wedge b_2
      }
    }}
    \caption{$\verif$ algorithm for Sig 3} \label{fig:signature-3-verify}
  \end{center}
\end{figure}

%--------------------------------------------------------------------%

\end{document}